\documentclass[11pt]{article}

\usepackage{amsmath, amsfonts, amssymb, amsthm}
\usepackage{mathrsfs}
\usepackage{xcolor}
\usepackage{enumitem}

\usepackage{newtxtext}
\usepackage{subcaption}
\usepackage{tikz}
\usepackage[plain,noend]{algorithm2e}
\usepackage[left=1in,top=1in,right=1in,bottom=1in,head=.1in]{geometry}
\usepackage{graphicx}
\graphicspath{{./figs/}}

\usepackage{hyperref}
\usepackage[numbers]{natbib}

\usepackage[textsize=scriptsize, textwidth=2cm, shadow]{todonotes}

\newtheorem{theorem}{Theorem}
\newtheorem{lemma}{Lemma}
\newtheorem{proposition}{Proposition}
\newtheorem{corollary}{Corollary}

\theoremstyle{definition}
\newtheorem{definition}{Definition}

\theoremstyle{remark}

\DeclareMathOperator{\var}{var}
\DeclareMathOperator{\cov}{cov}
\DeclareMathOperator{\cor}{cor}

\DeclareMathOperator{\tr}{tr}

\def\ind{\begin{picture}(9,8)
    \put(0,0){\line(1,0){9}}
    \put(3,0){\line(0,1){8}}
    \put(6,0){\line(0,1){8}}
    \end{picture}
   }

\def\top{{ \mathrm{\scriptscriptstyle T} }}

\begin{document}

\title{Identifiability of Treatment Effects with Unobserved Spatially Varying Confounders}

\author{
Tommy Tang\thanks{Department of Statistics, University of Illinois Urbana-Champaign, 
\texttt{tommymt2@illinois.edu}}
\and
Xinran Li\thanks{Department of Statistics, University of Chicago,
\texttt{xinranli@uchicago.edu}}
\and
Bo Li\thanks{Department of Statistics and Data Science, Washington University,
\texttt{bol@wustl.edu}}
}

\date{}  %

\maketitle

\begin{abstract}
The study of causal effects in the presence of unmeasured spatially varying confounders has garnered increasing attention. However, a general framework for identifiability, which is critical for reliable causal inference from observational data, has yet to be advanced. 
In this paper, we study a linear model with various parametric model assumptions on the covariance structure between the unmeasured confounder and the exposure of interest. 
We establish identifiability of the treatment effect for many commonly used spatial models for both discrete and continuous data, under mild conditions on the structure of observation locations and the exposure-confounder association. We also emphasize models or scenarios where identifiability may not hold, under which statistical inference should be conducted with caution.
\end{abstract}

\section{Introduction}
Causal inference in spatial setting has become increasingly important due to the rapid growth of spatial data in various fields such as environmental science and epidemiology \citep{reich2021review, gilbert2021approaches, schnell2020mitigating,guan2023spectral}.
Although the literature for causal inference with independent observations has been advanced extensively \citep[see, e.g.,][]{imbens2015causal, hernan2020causal}, methods and theory for spatial causal inference are more reserved, especially in the presence of unmeasured confounding.  
On the other hand, the spatial structure of the observations has also brought new opportunities for dealing with unmeasured confounding, as explored in the recent literature \citep{bind2019causal, gao2022causal, davis2019addressing}.

Conceptually, 
there are at least two views on modelling the unmeasured spatial confounder.
One is to treat the unmeasured confounder as a fixed but unknown function of its spatial location, and to model it using smooth functions such as spatial splines or other basis functions \citep{dupont2022spatial+, gilbert2021approaches}.
The other is to treat the unmeasured confounder as a random process given spatial locations, and to model it using, say, 
a Gaussian process with a certain spatial covariance structure \citep{schnell2020mitigating,guan2023spectral}. 
These two views on spatially varying confounders have rather different implications on the identifiability of the treatment effect of interest. 
Specifically, when the unmeasured confounder is a fixed but unknown function of spatial location, the exposure is unconfounded once given the observed spatial location. Thus, in this case, there is essentially ``no unmeasured confounding'', despite that the functional relation between confounder and spatial location may be challenging to estimate. %

In this paper, we will focus on the case where the unmeasured confounder is a random process given spatial locations. In this case, the treatment effect is generally not identifiable without additional assumptions \citep{schnell2020mitigating}. 
Recent work by \citet{schnell2020mitigating} and \citet{guan2023spectral} has shown that, under appropriate spatial model assumptions on the exposure-confounder association, the treatment effect can be identified. 
Our paper is along this line of research. 
Distinct from the existing literature, we study identifiability based on finite observation locations
and emphasize the critical role of the neighborhood structure of the locations, given the specified class of models.  
In particular, 
our results also 
generalize the identifiability result in \citet{schnell2020mitigating} 
from requiring a ring structure to a more general spatial structure.

Throughout the paper, we will first consider conditional autoregressive models that are often used for areal/aggregated data %
 and then models commonly used for 
 observations over a continuous spatial domain. 
We identify general and mild conditions on the neighborhood structure or distance matrices that ensure identifiability of the treatment effect under these models. We also establish the general non-identifiability of the linear model of coregionalization. 
In the absence of identifiability, 
statistical inference can be unstable, 
with Bayesian inference in particular being sensitive to the choice of prior.

\section{Spatial Model with Unmeasured Confounder}\label{sec:model}
We consider the following spatial model for observations at $n$ locations, denoted by $s_{1:n} = (s_1,\ldots, s_n)$: 
\begin{equation}\label{eq:sp-conf-model}
 Y = f(X_{1:n}, s_{1:n}) + Z \beta_Z + U + \epsilon,
\end{equation}
where $Y \in \mathbb{R}^n$ consists of the outcomes of interest at the $n$ locations, $X_{1:n} = (X_1, \ldots, X_n)$ with $X_i\in \mathbb{R}^p$ containing the measured covariates of dimension $p$ at the $i$th location, $f $ is an unknown function of measured covariates and observed locations that maps into $\mathbb{R}^n$, $Z\in \mathbb{R}^n$ is the exposure of interest, $U \in \mathbb{R}^n$ denotes the unmeasured confounder, $\epsilon \in \mathbb{R}^n$ is the Gaussian white noise independent of $(X_{1:n}, Z, U)$, 
and $\beta_Z \in \mathbb{R}$ is the coefficient of  $Z$.
We are interested in $\beta_Z$, which represents the treatment/causal effect of the exposure under certain assumptions; see \citet{guan2023spectral} for a description of such assumptions using the potential outcomes framework. 
We 
further assume
\begin{equation}\label{eq:ZU_X}
    \begin{pmatrix}
        U\\
        Z
    \end{pmatrix}
    \mid  X_{1:n} 
    \sim \mathcal{N}
    \left(  
    \begin{pmatrix}
        0\\
          g(X_{1:n}, s_{1:n})
    \end{pmatrix}, 
    \Sigma
    \right),  
    \quad 
    \text{where }
    \Sigma \equiv
    \begin{pmatrix}
        \Sigma_{UU} & \Sigma_{UZ}\\
        \Sigma_{ZU} & \Sigma_{ZZ}
    \end{pmatrix}
    \text{ is positive definite,}
\end{equation}
and $g$ is a function of the covariates and locations that maps into $\mathbb{R}^n$.
In \eqref{eq:ZU_X}, 
we assume that the conditional mean of $U$ given $X_{1:n}$ and $s_{1:n}$ is zero; this does not lose generality since we can always incorporate the part of $U$ that depends on $X_{1:n}$ and $s_{1:n}$
to the covariate term $f(X_{1:n}, s_{1:n})$ in \eqref{eq:sp-conf-model}. 

Since the unobserved $U$ can be correlated with $Z$, the coefficient of $Z$ from the regression of $Y$ on only $(Z,X)$ will generally be biased due to the omission of $U$.
Whether we can mitigate this bias depends on the covariance structure $\Sigma$ in \eqref{eq:ZU_X}.
As discussed in \citet{schnell2020mitigating}, as long as we can identify $\Sigma$ from the observed data, we can remove the omitted variable bias and obtain unbiased or consistent estimation for the treatment effect; in other words, the treatment effect is identifiable from the observed data. 
In the remaining of the paper, we will focus on the identifiability of the treatment effect $\beta_Z$ under different model assumptions on $\Sigma$. 
As demonstrated in the supplementary material, it suffices to study the identifiability of $\beta$ in the following simplified model that does not involve the observed covariates $X_{1:n}$:
\begin{equation}\label{eq:simp-model}
    Y = Z \beta + U + \epsilon, 
    \quad \epsilon \sim \mathcal{N}(0, \sigma_{\epsilon}^2 I_n), 
    \quad 
    \begin{pmatrix}
        U\\
        Z
    \end{pmatrix} 
    \sim \mathcal{N}
    \left(  
    \begin{pmatrix}
        0\\
        0
    \end{pmatrix}, 
    \begin{pmatrix}
        \Sigma_{UU} & \Sigma_{UZ}\\
        \Sigma_{ZU} & \Sigma_{ZZ}
    \end{pmatrix}
    \right). 
\end{equation}

Throughout the paper, we say $\beta$ in \eqref{eq:simp-model}
is identifiable from the observed data if it can be uniquely determined by the observed data distribution, i.e., the distribution of $(Y,Z)$. 
Equivalently, if there are two distinct values of $\beta$ that can lead to the same distribution of $(Y,Z)$, then $\beta$ is not identifiable from the observed data. 
In this case, we do not expect an unbiased or consistent estimator for $\beta$. The identifiability for other model parameters are defined analogously.

\section{Conditional Autoregressive Model}
\subsection{Schnell \& Papadogeorgou's Conditional Autoregressive Model}
\label{sec:car-papadogeorgu}

Let $W$ be an $n$ by $n$ symmetric \emph{proximity matrix} representing the 
strength
of connection among $n$ spatial locations, where $W_{ij} \ge 0$ for all $i,j$ and  $W_{ii}=0$ for all $i$ \citep{banerjee2003hierarchical}.  
For example, $W_{ij}$ can be a binary indicator denoting whether locations $i$ and $j$ are neighbors. $W_{ij}$ can also be the reciprocal of the distance between locations $i$ and $j$. 
Let $D$ be a diagonal matrix whose $i$th diagonal element is the sum of the entries in the $i$th column of $W$, representing the total strength of connection from all other locations. 
When $W_{ij}$s are binary denoting the neighboring relation, the $i$th diagonal element of $D$ represents the total number of neighbors for location $i$. 
In this subsection, we assume that all diagonal elements of $D$ are positive, which is necessary for the model introduced below to be well-defined. 

Following \citet{schnell2020mitigating}, we assume a conditional autoregressive (CAR) structure \citep{gelfand2010handbook} on the conditional distribution of $U$ given $Z$ and vice versa: 
\begin{align*}
    \var(U\mid Z)^{-1} = \tau_U(D-\varphi_U W), \  
    \var(Z \mid U)^{-1} = \tau_Z(D-\varphi_Z W), 
\end{align*}
where $\tau_U, \tau_Z > 0$ and $-1 < \varphi_U, \varphi_Z < 1$,
and a cross-Markov property with constant correlation for all $i$:
\begin{align*}
    Z_i \ind U_{-i} \mid Z_{-i}, U_i, \ \ 
    \text{and}  \ \ 
    \cor(U_i, Z_i\mid U_{-i}, Z_{-i}) = \rho, 
    \quad \text{where }  \rho \in (-1,1).  
\end{align*}
These assumptions equivalently impose the following model on the covariance matrix $\Sigma$ in \eqref{eq:simp-model}:
\begin{align}\label{eq:Sigma_SP}
    \Sigma = 
    \begin{pmatrix} 
        \tau_U(D-\varphi_U W) & -\rho\sqrt{\tau_U\tau_Z} D \\
        -\rho\sqrt{\tau_U\tau_Z} D & \tau_Z(D-\varphi_Z W)
    \end{pmatrix}^{-1}.
\end{align}
We refer readers to \citet{schnell2020mitigating} for a comprehensive discussion of the model’s motivation and interpretation, as well as the derivation of \eqref{eq:Sigma_SP} for the form of the covariance structure.
Moreover, to avoid technical clutter, we restrict the model parameters to only values that ensure the resulting covariance matrix $\Sigma$ is positive definite unless otherwise stated, i.e., the model is well-defined. This also applies to all other models studied in the remainder of the paper.

\citet{schnell2020mitigating} considered the case where $W_{ij}$s are binary representing neighboring relation. They showed the identifiability of the treatment effect $\beta$ under the conditions that $\varphi_U \ne 0$, the neighborhood structure forms a ring graph, and the sample size $n$ tends to infinity.
Below we generalize this identification result to 
allow for both a general proximity structure 
and a finite sample size.

To facilitate the discussion, we first introduce a partition of the locations based on their proximity matrix $W$. 
Specifically, we partition the $n$ locations, $s_1,...,s_n$, into disjoint \emph{connected components} $\{\mathcal{M}_b\}_{b=1}^B$ in the following way. 
First, $\mathcal{M}_b$s are mutually exclusive, and their union consists of all the $n$ locations. 
Second, for any two locations $s_i$ and $s_j$ in the same component $\mathcal{M}_b$, where $1\le b \le B$, 
there exists a sequence of locations $s_{l_0}, s_{l_1}, \ldots, s_{l_m}$ such that $l_0 = i$, $l_m = j$, and $W_{l_{r-1}, l_{r}}>0$ for all $1\le r\le m$. 
That is, there is a path connecting $s_i$ and $s_j$. 
Third, for any two locations that are in distinct components, there is no path connecting them. 
In short, each component $\mathcal{M}_b$ can be viewed as an island consisting of locations that are connected either directly or indirectly, and these components are isolated from each other.

For each component $\mathcal{M}_b$, we introduce $W_{[b]}$ and $D_{[b]}$ to denote the submatrices of $W$ and $D$, respectively, containing only the rows and columns associated with locations in $\mathcal{M}_b$. Equivalently, $W_{[b]}$ is the proximity matrix for locations in $\mathcal{M}_b$, and $D_{[b]}$ is defined analogously based on the column sums of $W_{[b]}$.

\begin{theorem}\label{thm:car-identifiability}
Consider $n$ spatial locations with a proximity matrix $W$. 
Let 
$\mathcal{M}_1,...,\mathcal{M}_B$ be the disjoint connected components, which form a partition of all the locations. 
Assume that the data generating process follows \eqref{eq:simp-model} and \eqref{eq:Sigma_SP}. 
If $\varphi_U
\neq
0$ and there exists one component $\mathcal{M}_b$ such that 
either the corresponding proximity matrix $W_{[b]}$ has at least three distinct eigenvalues or the corresponding $D_{[b]}$ contains distinct diagonal elements, then all the model parameters $\tau_U, \tau_Z, \varphi_U, \varphi_Z, \rho, \sigma_\epsilon^2$ and $\beta$ are identifiable. 
\end{theorem}

In Theorem \ref{thm:car-identifiability} and the later theorems, the parameters in the conditions, such as $\varphi_U$ here, refer to the true data generating parameters. 
The identification in Theorem \ref{thm:car-identifiability} needs $\varphi_U$ to be nonzero, which essentially requires the unmeasured confounder $U$ to be spatially correlated across the locations conditional on the observed exposure $Z$. 
Thus, the spatial dependence among unmeasured confounders is critical for causal identification, a departure from settings with no unmeasured confounding.
As shown in the supplementary material, $\varphi_U\ne 0$ is not only sufficient but also necessary for the identification of the treatment effect. 
One key reason is that the observable quantity $E(Y \mid Z)$ is a mixture of $Z$ and $E(U \mid Z)$. When $\varphi_U = 0$, $E(U\mid Z)$ is proportional to $Z$, making them indistinguishable based on observed data alone. 

Whether $\varphi_U\ne 0$ can be determined from the observed data distribution; 
specifically, if $E(YZ^\top)E(ZZ^\top)^{-1}$ is not proportional to the identity matrix, then $\varphi_U \ne 0$; otherwise, it must be that either $\varphi_U = 0$ or $\rho = 0$. 
The conditions on the proximity matrix often hold in practice and can be verified from the known locations.
When $W$ is a binary matrix, we obtain even simpler sufficient conditions on $W$.

\begin{corollary}\label{cor:car-identifiability-binary}
Consider the same model setup as in Theorem \ref{thm:car-identifiability}. If $\varphi_U
\neq
0$, all entries of $W$ are binary, and 
there exists a pair of locations $(i,j)$ within the same connected component such that $W_{ij}=0$ (i.e., they are indirect neighbors), 
then all the model parameters $\tau_U, \tau_Z, \varphi_U, \varphi_Z, \rho, \sigma_\epsilon^2$ and $\beta$ are identifiable from the observed data distribution.  
\end{corollary}

From Corollary \ref{cor:car-identifiability-binary}, 
when $\varphi_U \neq 0$, the treatment effect $\beta$ is identifiable from the observed data as long as at least one component has the neighborhood graph not fully connected.
Figure \ref{fig:adj-structure} illustrates this condition through four examples of neighborhood structures among six locations. 
Note that Figure \ref{fig:adj-structure} (b) represents exchangeable correlation. 
Specifically, the condition holds under Figure \ref{fig:adj-structure}(c)--(d), while fails under Figure \ref{fig:adj-structure}(a)--(b). 
In particular, as indicated by Figure \ref{fig:adj-structure}(d), 
the condition is satisfied under a ring graph with $n > 3$.
Thus, Corollary \ref{cor:car-identifiability-binary} generalizes the identification results of \citet{schnell2020mitigating}.
Moreover, Corollary \ref{cor:car-identifiability-binary}  not only accommodates more flexible neighborhood structures, but also applies to finite sample sizes. 
It will be interesting to investigate whether this more general identification result extends to settings with spatial interference \citep{papadogeorgou2023spatial}, and we leave this for future study.

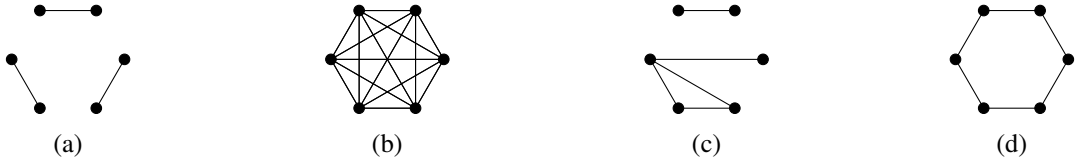
\begin{figure}[h]
    \centering
    
    \begin{subfigure}[b]{0.24\textwidth}
        \centering
        \begin{tikzpicture}[scale=0.75]
            \foreach \angle in {60, 120, 180, 240, 300, 360}
                \node[draw, circle, fill=black, inner sep=1pt, minimum size=4pt] at (\angle:1cm) {};
            \draw (60:1cm) -- (120:1cm);
            \draw (180:1cm) -- (240:1cm);
            \draw (300:1cm) -- (360:1cm);
        \end{tikzpicture}
        \caption{
        }
    \end{subfigure}
    \hfill
    \begin{subfigure}[b]{0.24\textwidth}
        \centering
        \begin{tikzpicture}[scale=0.75]
            \foreach \angle in {60, 120, 180, 240, 300, 360}
                \node[draw, circle, fill=black, inner sep=1pt, minimum size=4pt] at (\angle:1cm) {};
            \foreach \source in {60, 120, 180, 240, 300, 360}
                \foreach \dest in {60, 120, 180, 240, 300, 360}
                    \draw (\source:1cm) -- (\dest:1cm);
        \end{tikzpicture}
        \caption{
        }
    \end{subfigure}
    \hfill
    \begin{subfigure}[b]{0.24\textwidth}
        \centering
        \begin{tikzpicture}[scale=0.75]
            \foreach \angle in {60, 120, 180, 240, 300, 360}
                \node[draw, circle, fill=black, inner sep=1pt, minimum size=4pt] at (\angle:1cm) {};
            \draw (60:1cm) -- (120:1cm);
            \draw (180:1cm) -- (240:1cm) -- (300:1cm) -- cycle;
            \draw (180:1cm) -- (360:1cm);
        \end{tikzpicture}
        \caption{
        }
    \end{subfigure}
    \begin{subfigure}[b]{0.24\textwidth}
        \centering
        \begin{tikzpicture}[scale=0.75]
        \foreach \angle in {60, 120, 180, 240, 300, 360}
            \node[draw, circle, fill=black, inner sep=1pt, minimum size=4pt] at (\angle:1cm) {};
        
        \foreach \angle in {60, 120, 180, 240, 300, 360} {
            \pgfmathsetmacro\next{\angle + 60}
            \draw (\angle:1cm) -- (\next:1cm);
        }
        \end{tikzpicture}
        \caption{
        }
    \end{subfigure}
    \caption{Examples of neighborhood structures among 6 locations. The condition on $W$ specified in Corollary \ref{cor:car-identifiability-binary} is violated in (a) and (b), but satisfied in (c) and (d).
    }\label{fig:adj-structure}
\end{figure}

\subsection{Leroux Conditional Autoregressive Model}\label{sec:leroux_car}
We now consider the Leroux CAR model \citep{leroux2000estimation} for the covariance structure of the unmeasured confounder $U$ and exposure $Z$ in \eqref{eq:simp-model}, as in \citet{guan2023spectral}: 
\begin{equation}\label{eq:lcar-var}
\Sigma_{UU} = \sigma_U^2\{ (1-\lambda_U)I_n + \lambda_U (D-W) \}^{-1}, \quad  \Sigma_{ZZ} = \sigma_Z^2 \{ (1-\lambda_Z)I_n + \lambda_Z (D-W) \}^{-1},
\end{equation} 
where $W$ is a proximity matrix defined as in \S \ref{sec:car-papadogeorgu}, $D$ is the diagonal matrix consisting of the column sums of $W$, and the parameters satisfy that $\sigma_U, \sigma_Z > 0$ and $\lambda_U, \lambda_Z \in [0,1)$.
Here the entries of $W$ do not need to be binary, which is more general than \citet{guan2023spectral}. 
To specify the covariance between $U$ and $Z$, we consider $U$ and $Z$ in their spectral domain. Let $D-W = P \Omega P^\top$ denote the spectral decomposition of $D-W$, where $\Omega$ is a diagonal matrix and $P$ is an orthogonal matrix.
Let $U' = P^\top U$ and $Z' = P^\top Z$, whose covariance structure then has the following form: 
\begin{align}\label{eq:leroux_car}
    \cov\left( \begin{pmatrix}U' \\ Z' \end{pmatrix} \right)
    & \equiv 
    \begin{pmatrix} \Sigma'_{UU} & \Sigma'_{UZ} \\ \Sigma'_{ZU} & \Sigma'_{ZZ} \end{pmatrix}
    = 
    \begin{pmatrix} 
    \sigma^2_U \{(1-\lambda_U)I_n+\lambda_U \Omega\}^{-1} & P^\top \Sigma_{UZ} P 
    \\ 
    P^\top \Sigma_{ZU} P & \sigma^2_Z \{(1-\lambda_Z)I_n+\lambda_Z \Omega\}^{-1}
    \end{pmatrix}. 
\end{align}
Following \citet{guan2023spectral}, we consider the following two model assumptions on the covariance between $U$ and $Z$ in their spectral domain: for some $\rho \in [-1, 1]$
\begin{align}\label{eq:lcar-crosscov-flex}
\Sigma_{UZ}' & = \rho \sigma_U\sigma_Z \{(1-\lambda_{UZ})I_n + \lambda_{UZ}\Omega\}^{-1}, 
\\
\label{eq:lcar-crosscov-pars} 
\Sigma_{UZ}' & = \rho\sigma_U\sigma_Z \big[\{(1-\lambda_U)I_n + \lambda_U\Omega\}^{1/2}\{(1-\lambda_Z)I_n+\lambda_Z \Omega\}^{1/2}\big]^{-1},
\end{align}
which correspond to non-parsimonious and parsimonious models, respectively. The parsimonious model in \eqref{eq:lcar-crosscov-pars} involves one less parameter. In either formulation,  the four submatrices in \eqref{eq:leroux_car} are all diagonal.
Below we study the identification of the causal effect $\beta$ under the above model assumptions. 

We consider first the non-parsimonious model 
for the covariance structure of $U$ and $Z$, whose identifiability has not been studied in \citet{guan2023spectral}. 

\begin{theorem}\label{thm:lcar-flex-identifiability}
Consider $n$ spatial locations with a proximity matrix $W$, and assume the data generating process follows \eqref{eq:simp-model}, \eqref{eq:leroux_car} and \eqref{eq:lcar-crosscov-flex}. 
\begin{enumerate}[label=(\roman*), topsep=1ex,itemsep=-0.3ex,partopsep=1ex,parsep=1ex]
    \item If $\rho \neq 0, \lambda_{UZ} \neq \lambda_Z$, and $D-W$ contains at least three distinct eigenvalues, then the treatment effect $\beta$ along with $\sigma_Z, \lambda_Z, \lambda_{UZ}$ are identifiable. If further $\lambda_U\neq 0$, then $\lambda_U, \rho, \sigma_U$, and $\sigma_\epsilon$ are identifiable. 
    \item If $\rho =0, \lambda_U\neq \lambda_Z, \lambda_Z \neq 0, \lambda_U \neq 0$, and $D-W$ contains at least four distinct eigenvalues, then the treatment effect $\beta$ and all parameters except for $\lambda_{UZ}$ are identifiable.
\end{enumerate}
\end{theorem}

In case (i) of Theorem~\ref{thm:lcar-flex-identifiability}, 
the identification requires that the exposure $Z$ and the confounder $U$ are correlated, 
and that the covariance between $U$ and $Z$ exhibit a spatial dependence structure 
distinct from that of $Z$ itself.
These conditions ensure that the contributions of $Z$ and $E(U \mid Z)$ in $E(Y \mid Z)$ 
are linearly independent, so that we can distinguish these two components and achieve identification. 
In case (ii) of Theorem~\ref{thm:lcar-flex-identifiability}, $\rho=0$ indicates that the exposure $Z$ and the confounder $U$ are uncorrelated, so there is actually no confounding. 
However, the treatment effect $\beta$ may still be unidentifiable, 
because an alternative data-generating process with a correlated unmeasured confounder 
could yield the same observed data distribution. Therefore, case (ii) further requires $U$ and $Z$ to be spatially autocorrelated 
but with distinct dependence structures, 
so that we can achieve identification using the information in $\var(Y \mid Z)$.
As demonstrated in the supplementary material, the conditions on parameters in Theorem \ref{thm:lcar-flex-identifiability} are also necessary for the identification of the treatment effect. 

Again, we can assess from the observed data distribution whether some identification conditions hold. 
Specifically, if $E(Y'Z'^{\top})E(Z'Z'^{\top})^{-1}$ is not proportional to the identity matrix, 
then $\rho \neq 0$ and $\lambda_{UZ} \neq \lambda_Z$, 
indicating that the parameter conditions in case (i) must hold. 
The requirement 
that $D - W$ have a sufficient number of distinct eigenvalues 
can also be readily checked with the observed locations. 
When $W$ is binary indicating neighborhood structure, 
a sufficient condition for $D - W$ to have at least three (or four) distinct eigenvalues 
is the existence of two locations that are indirectly connected through at least one 
(or two) intermediate locations \citep[][Proposition~1.3.3]{brouwer2011spectra}.

We then consider the parsimonious model for the covariance between $U$ and $Z$. 
Unlike \citet{guan2023spectral}, 
we explicitly require the spectrum of $D - W$ to contain a sufficient number of distinct eigenvalues, 
and fully state the sufficient conditions on the model parameters.

\begin{theorem}\label{thm:lcar-pars-identifiability}
Consider $n$ spatial locations with a proximity matrix $W$, and assume that the data generating process follows \eqref{eq:simp-model}, \eqref{eq:leroux_car} and \eqref{eq:lcar-crosscov-pars}. Suppose that $D-W$ contains at least three distinct eigenvalues.
\begin{enumerate}[label=(\roman*), topsep=1ex,itemsep=-0.3ex,partopsep=1ex,parsep=1ex]
\item If $\lambda_U\neq \lambda_Z$ and $\rho \neq 0$, then the treatment effect $\beta$ along with $\sigma_Z$, $\lambda_Z$, $\lambda_U$ and $\rho \sigma_U$ are identifiable.
\item If $\lambda_U \neq \lambda_Z$ and $\lambda_U \neq 0$, then all parameters, including the causal effect $\beta$, are identifiable.
\end{enumerate}
\end{theorem}

The two cases in Theorem~\ref{thm:lcar-pars-identifiability} 
are similar to those in Theorem~\ref{thm:lcar-flex-identifiability}. 
In case (i), identification is achieved by separating the linearly independent contributions 
of $Z$ and $E(U \mid Z)$ in $E(Y \mid Z)$. 
In case (ii), where $\rho$ may be zero, 
identification relies on the information contained in $\var(Y \mid Z)$. 
As shown in the supplementary material, 
the parameter conditions in Theorem~\ref{thm:lcar-pars-identifiability} 
are also necessary for identifying the treatment effect. 

Similarly, we can assess the identification conditions in Theorem~\ref{thm:lcar-pars-identifiability}. 
If \(E(Y'Z'^{\top})E(Z'Z'^{\top})^{-1}\) is not proportional to the identity matrix, 
then we must have \(\lambda_U \neq \lambda_Z\) and \(\rho \neq 0\); 
that is, the parameter conditions in case (i) hold. 
The condition on the proximity structure in Theorem~\ref{thm:lcar-pars-identifiability} often holds in practice and can be easily verified from the observed spatial locations. 
For example, 
as discussed earlier, 
when $W$ is a binary matrix for neighborhood structure, it holds when there are two indirectly connected locations.

\section{Models for Geostatistical Data}
\subsection{Linear Model of Coregionalization}
Let $W$ be an $n$ by $n$ symmetric {\it distance matrix} representing the distances among $n$ spatial locations, where $W_{ij}\ge 0$ for all $i,j$ and $W_{ii} = 0$ for all $i$. 
For example, $W_{ij}$ can be the Euclidean distance between the locations $s_i$ and $s_j$. 
For geostatistical data, it is common to model the unmeasured confounder $U$ and the exposure $Z$
through Gaussian processes with covariance structure depending on the distance between locations. 
Throughout this section, we will consider widely used models for multivariate spatial processes, which can handle potential dependence between $U$ and $Z$.

We first consider the linear model of coregionalization \citep{wackernagel2003multivariate, gelfand2004nonstationary}, 
which represents the spatial processes $U$ and $Z$ as linear combinations of mutually independent latent spatial processes.
Each latent process shares a common covariance function but is characterized by distinct parameters.
Specifically, 
\begin{equation}\label{eq:LMC}
Z = \sum_{t=1}^T a_t L_t,  \qquad U = \sum_{t=1}^T b_t L_t,
\qquad (a_t, b_t \in\mathbb{R} \text{ for } 1\le t \le T) 
\end{equation}
where each $L_t$ is a mean-zero Gaussian process with covariance function $C(\phi_t, \cdot)$ for some parameter vector $\phi_t$. 
That is, the covariance between two observations of $L_t$ at locations $s_i$ and $s_j$ is $C(\phi_t, W_{ij})$. 
The form of the covariance function $C(\cdot, \cdot)$ is assumed known, with $C(\phi, 0) = 1$ for all $\phi$ to avoid parameter redundancy. 
The unknown parameters for the covariance structure of $(Z, U)$ are $a_t$, $b_t$ and $\phi_t$ for $1\le t \le T$.

Unfortunately, under the above model, the treatment effect is not identifiable from the observed data.

\begin{theorem}\label{thm:lmcoreg-nonid}
Consider $n$ spatial locations with a distance matrix $W$, 
and assume that the data generating process follows \eqref{eq:simp-model} and \eqref{eq:LMC}. 
The treatment effect $\beta$ is not identifiable, 
even if we know the true values of $a_t$ and $\phi_t$ for $1\le t \le T$ and $\sigma^2_{\epsilon}$. 
\end{theorem}

From Theorem~\ref{thm:lmcoreg-nonid}, we must exercise caution when using the linear model of coregionalization to characterize the dependence between the exposure and unmeasured confounder.
As shown in the proof, based solely on the observed data distribution, the treatment effect $\beta$ is not identifiable and can take any value on the real line.
Consequently, the likelihood function 
can have 
local maxima at infinitely many parameter configurations, 
and the Bayesian inference for $\beta$ can be highly sensitive to the prior. The intuition behind the nonidentifiability is that even if the true $a_t$, $\phi_t$ and $\sigma^2_{\epsilon}$ are known, the observed data distribution can at most reveal the values of $\beta a_t + b_t$ for all $t$, from which $\beta$ and $b_t$ cannot be separately identified.

\subsection{Bivariate Spatial Covariance Models}\label{sec:bivariate_spatial}
We then consider a general class of bivariate covariance models that can explicitly model the autocovariances of the unmeasured confounder $U$ and the exposure $Z$, as well as their covariance. 
Let $C(\psi, \cdot)$ be a 
covariance 
function indexed by a parameter vector $\psi$, 
with $C(\psi, 0)=1$ to avoid redundancy. 
We assume the following covariance structure for $(U, Z)$ in \eqref{eq:simp-model}: for $1 \le i,j \le n$,
\begin{align}\label{eq:bivariate-cov-framework}
(\Sigma_{UU})_{ij} &= \sigma_U^2 C(\psi_U, W_{ij}), \ (\Sigma_{ZZ})_{ij} = \sigma_Z^2 C(\psi_Z, W_{ij}), \ 
(\Sigma_{UZ})_{ij} = \rho \sigma_U \sigma_Z C(\psi_{UZ}, W_{ij}),
\end{align}  
where $\sigma_U, \sigma_Z > 0$,
$(\Sigma_{UU})_{ij}$ denotes the $(i,j)$th element of $\Sigma_{UU}$,
and $(\Sigma_{ZZ})_{ij}$ and $(\Sigma_{UZ})_{ij}$ are defined analogously.
The model in \eqref{eq:bivariate-cov-framework} resembles that of \citet{guan2023spectral}.
However, unlike their spectral-domain formulation, we focus on the spatial domain and assume that only a finite set of locations rather than the entire spectrum is observed.

We first define the following condition on the covariance functions, which is crucial for identification.
\begin{definition}\label{defn:k-linear-independence}
Given a family of covariance functions $\{C(\psi, \cdot): \psi \in \Psi\}$ and a set $\mathcal{S}$ of 
nonnegative
values, the family is said to be $K$-linearly independent with respect to $\mathcal{S}$ if, for any $K$ distinct parameters $\psi_1, \ldots, \psi_K \in \Psi$, $c_1 C(\psi_1, w) + \cdots + c_K C(\psi_K, w) = 0$ for all $w \in \mathcal{S}$ if and only if $c_1 = \cdots = c_K = 0$.
\end{definition}
Many commonly used covariance families are $K$-linearly independent with respect to any set that contains a sufficiently large number of distinct values. As shown in the supplementary material, the exponential, Gaussian, and powered exponential covariance families are all $K$-linearly independent with respect to any set containing at least $K$ distinct positive values. However, not all covariance families enjoy this property. The spherical covariance family with a restricted domain for its range parameter is only 3-linearly independent with respect to any set containing at least four distinct values smaller than the minimum range parameter value, but for $K\ge 4$, is not $K$-linearly independent with respect to any set that is not dense in $[0, \infty)$. Moreover, the wave covariance family may fail to be linearly independent even with respect to a set containing infinitely many distinct values.

With this definition stated, we now give the following theorem on the identifiability of the model in \eqref{eq:bivariate-cov-framework}.

\begin{theorem}\label{thm:flexible-bivariate-id}
Consider $n$ spatial locations with a distance matrix $W$, 
and assume that the data generating process follows \eqref{eq:simp-model} and \eqref{eq:bivariate-cov-framework}. 
If  $\rho \neq 0$, $\psi_{UZ} \ne \psi_Z$,
and 
the family of covariance functions $C(\psi, \cdot)$ is 3-linearly independent with respect to the set of off-diagonal elements of $W$, 
then all the model parameters, including the treatment effect $\beta$, 
are identifiable. 
\end{theorem}

The identiability conditions in Theorem~\ref{thm:flexible-bivariate-id} are similar to those in case~(i) of Theorem~\ref{thm:lcar-flex-identifiability}: the exposure and the unmeasured confounder must be associated, and the covariance between $U$ and $Z$ must have a spatial structure that differs from that of $Z$ itself. As before, these  conditions can be assessed using observable quantities. Specifically, if
$
E(YZ^\top)E(ZZ^\top)^{-1}
$
is not proportional to the identity matrix, then we must have $\rho \neq 0$ and $\psi_{UZ} \neq \psi_Z$. In addition, if the family of covariance functions is exponential, Gaussian, or powered exponential, then the linear-independence condition in Theorem~\ref{thm:flexible-bivariate-id} holds whenever the distance matrix $W$ contains at least three distinct positive values.

\subsection{Parsimonious Mat\'ern Covariance Models}\label{sec:matern}

We finally consider using the Mat\'ern covariance functions to model the covariance structure of the unmeasured confounder $U$ and the exposure $Z$.  
The Mat\'ern correlation function is defined as
\[ C(\phi, \nu; w)
= 
{2^{1-\nu}\over \Gamma(\nu)}(w/\phi)^\nu K_\nu(w/\phi), 
\qquad (w \ge 0, \phi > 0, \nu >0)
\]
where $\Gamma(\cdot)$ is the gamma function and $K_\nu(\cdot)$ is the modified Bessel function of the second kind of order $\nu$. 
It involves two parameters: 
the spatial range parameter $\phi$ 
and the smoothness parameter $\nu$. 
 Following \citet{gneiting2010matern}, 
we model the exposure $Z$ and unmeasured confounder $U$ as Gaussian processes with 
all the marginal and cross covariance functions being the Mat\'ern function. 
Moreover, we consider the parsimonious bivariate Mat\'ern model 
so that all the marginal and cross covariance functions share a common range parameter and the smoothhness parameter for the cross-covariance is the average of those for the marginal covariances. 
That is, for all $1\le i, j \le n$, 
\begin{equation}\label{eq:pars-matern}
(\Sigma_{UU})_{ij} = \sigma_U^2 C(\phi, \nu_U; W_{ij}), 
\ 
(\Sigma_{ZZ})_{ij} = \sigma_Z^2 C(\phi, \nu_Z; W_{ij}), 
\ 
(\Sigma_{UZ})_{ij} = \rho \sigma_U \sigma_Z C \Big( \phi, \frac{\nu_U+\nu_Z}{2}; W_{ij} \Big), 
\end{equation}  
recalling that $W_{ij}$ denotes the distance between locations $s_i$ and $s_j$. 

The model in \eqref{eq:pars-matern} is a special case of \eqref{eq:bivariate-cov-framework} that uses the Matérn covariance function and imposes additional constraints on the parameters governing the marginal and cross covariance matrices.
Consequently, the identifiability results in Theorem \ref{thm:flexible-bivariate-id} imply identifiability of the parameters of the parsimonious Matérn model in \eqref{eq:pars-matern}, provided that the conditions of the theorem hold.
However, verifying these conditions, particularly the linear independence of the relevant Matérn covariance functions, is challenging because of the complex form of the modified Bessel function. 
In the following theorem, 
we instead consider an asymptotic setting in which there exists a sequence of location pairs whose distances diverge to infinity.

\begin{theorem}\label{thm:bivariate-matern-id-full-asymptotic}
Consider $n$ spatial locations with a distance matrix $W$, and assume that the data generating process follows \eqref{eq:simp-model} and \eqref{eq:pars-matern}.
Suppose that $\nu_U \neq \nu_Z$ and $\rho\neq 0$. 
As $n\rightarrow \infty$, 
if $\max_{1\le i, j \le n} W_{ij}$ converges to infinity, then all the model parameters, including the treatment effect $\beta$, 
are identifiable. 
\end{theorem}

The conditions on the model parameters are essentially the same as those in Theorem \ref{thm:flexible-bivariate-id}, since the smoothness parameter for the cross-covariance is $(\nu_U+\nu_Z)/2$. 
We conjecture that even for a finite set of locations, the model with a parsimonious Mat\'ern covariance structure is identifiable. However, due to the complexity of the Bessel function in Mat\'ern covariance function, we can only prove identifiability when the data contain location pairs with very large distances, which is still common for environmental data.
In the supplementary material, we also establish finite-sample identifiability of the parsimonious Matérn model in \eqref{eq:pars-matern}, provided that $\nu_U$ and $\nu_Z$ are known and distinct.

\appendix
\section*{Supplementary Material}
\addcontentsline{toc}{section}{Supplementary Material}

\renewcommand{\thesection}{A\arabic{section}}
\renewcommand{\thetheorem}{A\arabic{theorem}}
\renewcommand{\thelemma}{A\arabic{lemma}}
\renewcommand{\theproposition}{A\arabic{proposition}}
\renewcommand{\thecorollary}{A\arabic{corollary}}
\renewcommand{\theequation}{A\arabic{equation}}
\renewcommand{\thefigure}{A\arabic{figure}}
\renewcommand{\thetable}{A\arabic{table}}
\renewcommand{\theremark}{A\arabic{remark}}

This appendix is lengthy. We briefly describe the contents of this supplementary material. Section \S \ref{sec:proofs-framework} gives proofs for some of our key assumptions in our framework and shows some computational details for the expressions for our observable parameters. Following it, Section \S \ref{sec:proofs-discrete-car} covers proofs, including edge cases, for the conditionally defined CAR model. Next, Sections \S \ref{sec:proofs-lcar-nonpars} and \S \ref{sec:proofs-lcar-pars} provide proofs for the non-parsimonious and parsimonious Leroux CAR models, respectively, with special attention being given to examples that illuminate the sharpness of some of the conditions. Then, Section \S \ref{sec:proofs-lmcoreg} contains the proof for the linear model of coregionalization. Continuing the results for geostatistical models, Section \S \ref{sec:proofs-continuous-bivariate} contains proofs for the flexible bivariate continuous case and Section \S \ref{sec:proofs-matern} gives proofs for the parsimonious Mat\'ern covariance specifically. Finishing proofs directly related to continuous, geostatistical data, \S \ref{sec:linear-non-ind} details the $K$-linear independence (or non-independence) of various covariance functions often seen in spatial statistics. 
Finally, Section \S \ref{sec:pd} gives brief notes on positive definiteness for our models.

\section{Technical details for \S \ref{sec:model}}\label{sec:proofs-framework}

\subsection{Simplification of the linear structrual model}\label{sec:no-measured-covariates}
We show that for the purposes of identifiability, it suffices to consider the model in \eqref{eq:simp-model} with no observed covariates. 
From \eqref{eq:sp-conf-model} and \eqref{eq:ZU_X}, 
\begin{align*}
    Y 
    & = 
    f(X_{1:n}, s_{1:n}) + Z \beta_Z + U + \epsilon\\
    & = f(X_{1:n}, s_{1:n}) + g(X_{1:n}, s_{1:n}) \beta_Z + \{ Z - E(Z\mid X_{1:n}) \} \beta_Z + U + \epsilon\\
    & = E(Y\mid X_{1:n}) + \{ Z - E(Z\mid X_{1:n}) \} \beta_Z + U + \epsilon,
\end{align*}
where the second equality follows from \eqref{eq:ZU_X}, and the last equality holds because $E(U\mid X_{1:n}) = E(\epsilon \mid X_{1:n}) = 0$. 

Define $\tilde{Y} = Y- E(Y\mid X_{1:n})$ and $\tilde{Z} = Z - E(Z\mid X_{1:n})$. Then
\[
\tilde{Y} =
\tilde{Z}\beta_Z + U + \epsilon, 
\quad 
\epsilon \sim \mathcal{N}(0, \sigma_{\epsilon}^2 I_n), 
\quad 
\begin{pmatrix}
        U\\
        \tilde{Z}
    \end{pmatrix} 
    \sim \mathcal{N}
    \left(  
    \begin{pmatrix}
        0\\
        0
    \end{pmatrix}, 
    \begin{pmatrix}
        \Sigma_{UU} & \Sigma_{UZ}\\
        \Sigma_{ZU} & \Sigma_{ZZ}
    \end{pmatrix}
    \right), 
\] 
which 
has the same form as the simplified model in \eqref{eq:simp-model}.
Moreover, 
compared to the original model in \eqref{eq:sp-conf-model} and \eqref{eq:ZU_X}, 
the parameter of interest $\beta_Z$ is unchanged,  
nor the covariance matrix $\Sigma$ or the noise variance $\sigma^2_\epsilon$.
Because the difference between $(Y,Z)$ and $(\tilde{Y}, \tilde{Z})$ is fully determined by the observed data distribution, 
the identifiability of the parameters $\beta_Z, \Sigma$ and $\sigma^2_\epsilon$ is the same under both the original and the simplified models.

\subsection{
Representation for the observed data distribution
}
Under the model in \eqref{eq:simp-model}, 
the observed data distribution is the distribution of $(Y,Z)$. 
Because $(Y,Z)$ is jointly Gaussian with zero means, the distribution of $(Y,Z)$ can be uniquely determined by either
\begin{align}\label{eq:joint_ZY}
    \var(Z) & = \Sigma_{ZZ}, 
    \nonumber
    \\
    \cov(Y, Z) & = \beta \Sigma_{ZZ} + \Sigma_{UZ}, \nonumber
    \\
    \var(Y) & = \beta^2 \Sigma_{ZZ} + \Sigma_{UU} + \beta \Sigma_{ZU} + \beta \Sigma_{UZ} + \sigma^2_{\epsilon} I_n, 
\end{align}
or 
\begin{align}\label{eq:condvar}
    \var(Z) & = \Sigma_{ZZ},
    \nonumber
    \\
    \var(Y\mid Z) & = 
    \Sigma_{UU} - \Sigma_{UZ} \Sigma_{ZZ}^{-1}\Sigma_{ZU} + \sigma^2_\epsilon  I_n, \nonumber
    \\
    \cov(Y,Z) \var(Z)^{-1}
    & = \beta I_n + \Sigma_{UZ}\Sigma_{ZZ}^{-1}.
\end{align}
The expressions in \eqref{eq:joint_ZY} and \eqref{eq:condvar} follow by some algebra. 
The quantities in \eqref{eq:joint_ZY} directly characterize the joint distribution of $(Y,Z)$, while the quantities in \eqref{eq:condvar} characterize the marginal distribution of $Z$ and the conditional distribution of $Y$ given $Z$. 
Each of the two forms in \eqref{eq:joint_ZY} and \eqref{eq:condvar} for characterizing the observed data distribution can be preferred under different model assumptions.

\section{Technical details for \S \ref{sec:car-papadogeorgu}}\label{sec:proofs-discrete-car}

\subsection{Proof of Theorem \ref{thm:car-identifiability}}
\label{sec:car-proofs}

To prove Theorem \ref{thm:car-identifiability}, we first define some notations. As in \S\ref{sec:car-papadogeorgu}, define $W, D, W_{[b]}$s and $D_{[b]}$s. 
Define further $\Lambda_{[b]}$ and $\Gamma_{[b]}$ via the spectral decomposition  of $D_{[b]}^{-1/2} W_{[b]} D_{[b]}^{-1/2}$: 
\begin{align*}%
    D_{[b]}^{-1/2} W_{[b]} D_{[b]}^{-1/2} = \Gamma_{[b]} \Lambda_{[b]} \Gamma_{[b]}^\top,
\end{align*}
where $\Lambda_{[b]}$ is a diagonal matrix and $\Gamma_{[b]}$ is an orthogonal matrix. 
Without loss of generality, we assume that the locations are ordered according to the connected components to which they belong, such that both $W$ and $D$ are block-diagonal matrices formed by the $W_{[b]}$s and $D_{[b]}$s. 
Let $\Lambda$ and $\Gamma$ be the block-diagonal matrices formed by the $\Lambda_{[b]}$s and $\Gamma_{[b]}$s. We can then verify that  
$D^{-1/2} W D^{-1/2} = \Gamma \Lambda \Gamma^\top$ is a spectral decomposition $D^{-1/2} W D^{-1/2}$, 
where $\Lambda$ is a diagonal matrix and $\Gamma$ is an orthogonal matrix.

We need the following three lemmas to prove Theorem \ref{thm:car-identifiability}.  

\begin{lemma}\label{lem:cond_Y_Z}
Define $W, D, W_{[b]}$s and $D_{[b]}$s
as in \S\ref{sec:car-papadogeorgu}, 
and recall the definition of $\Lambda_{[b]}$s,
$\Gamma_{[b]}$s, $\Lambda$ and $\Gamma$ at the beginning of \S\ref{sec:car-proofs}. 
Let $(\tau_U, \varphi_U, \tau_{\epsilon})$ be a set of parameters such that 
$\tau_U>0$, $\tau_{\epsilon}>0, \varphi_U \ne 0$ and $I-\varphi_U \Lambda$ is positive definite, 
and $(\tilde{\tau}_U, \tilde{\varphi}_U, \tilde{\tau}_{\epsilon})$ be another set of parameters such that $\tilde{\tau}_U>0$, $\tilde{\tau}_{\epsilon}>0$ and $I-\tilde{\varphi}_U \Lambda$ is positive definite.
Suppose that 
\begin{align}\label{eq:cond_Y_Z_eq}
    \tau_U^{-1} (I-\varphi_U \Lambda)^{-1} 
    +\tau_{\epsilon}^{-1} \Gamma^\top D \Gamma = 
    \tilde{\tau}_U^{-1} (I-\tilde{\varphi}_U \Lambda)^{-1} 
    +\tilde{\tau}_{\epsilon}^{-1} \Gamma^\top D \Gamma. 
\end{align}
or equivalently, for all $1\le b\le B$, 
\begin{align}\label{eq:cond_Y_Z_eq_k}
    \tau_U^{-1} (I-\varphi_U \Lambda_{[b]})^{-1} 
    +\tau_{\epsilon}^{-1} \Gamma_{[b]}^\top D_{[b]} \Gamma_{[b]} = 
    \tilde{\tau}_U^{-1} (I-\tilde{\varphi}_U \Lambda_{[b]})^{-1} 
    +\tilde{\tau}_{\epsilon}^{-1} \Gamma_{[b]}^\top D_{[b]} \Gamma_{[b]}. 
\end{align}
\begin{enumerate}[label=(\roman*)]
    \item If $D_{[b]}$ contains distinct diagonal elements for some $1\le b \le B$, 
    then $(\tau_U, \varphi_U, \tau_{\epsilon}) = (\tilde{\tau}_U, \tilde{\varphi}_U, \tilde{\tau}_{\epsilon})$.
    \item If $D_{[b]}$ is a scaled identity matrix for all $1\le b \le B$, and 
    $\Lambda_{[b]}$ has at least three distinct diagonal elements for some $1\le b \le B$,
    then $(\tau_U, \varphi_U, \tau_{\epsilon}) = (\tilde{\tau}_U, \tilde{\varphi}_U, \tilde{\tau}_{\epsilon})$. 
\end{enumerate}
\end{lemma}

\begin{proof}[Proof of Lemma \ref{lem:cond_Y_Z}(i)]
We first consider the case where $D_{[b]}$ contains distinct diagonal elements, for some $1\le b\le B$. 
For descriptive convenience, 
we use the notation $[b]$ to denote the set of indices that correspond to locations in $\mathcal{M}_b$. 

First, we prove that $\Gamma_{[b]}^\top D_{[b]} \Gamma_{[b]}$ must contain some nonzero off-diagonal entry.
We prove this by contradiction. 
Assume that $\Gamma_{[b]}^\top D_{[b]} \Gamma_{[b]}$ is a diagonal matrix. By definition, $\Gamma_{[b]}^\top D_{[b]}^{-1/2} W_{[b]} D_{[b]}^{-1/2} \Gamma_{[b]} =  \Lambda_{[b]}$ is also a diagonal matrix. 
Thus, $D_{[b]}$ and $D_{[b]}^{-1/2} W_{[b]} D_{[b]}^{-1/2}$ can be 
simultaneously diagonalized. 
This implies that they commute with each other, i.e., 
$D_{[b]} D_{[b]}^{-1/2} W_{[b]} D_{[b]}^{-1/2} = D_{[b]}^{-1/2} W_{[b]} D_{[b]}^{-1/2} D_{[b]}$, which is  equivalent to that $D_{[b]} W_{[b]} = W_{[b]} D_{[b]}$. 
Consequently, for any $i, j \in [b]$, 
$D_{ii}W_{ij} = W_{ij} D_{jj}$. 
If $W_{ij}\ne 0$, then $D_{ii}=D_{jj}$. 
By definition, all units in $[b]$ are connected. Thus, for any $i,j\in [b]$, there is a sequence $s_{m_1},s_{m_2},\ldots,s_{m_l}$ such that $m_1 = i, m_l = j,$ and  $W_{m_r,m_{r+1}}>0$ for every $1\le r \le l-1$, which implies that $D_{ii}= D_{m_1 m_1} = \cdots = D_{jj}$. Therefore, $D_{ii}$ must be identical for all $i\in [b]$. 
This contradicts the assumption that $D_{[b]}$ contains distinct diagonal elements. 

Second, we prove that $\tau_{\epsilon} = \tilde{\tau}_{\epsilon}$. 
From the discussion before, 
$\Gamma_{[b]}^\top D_{[b]} \Gamma_{[b]}$ contains nonzero off-diagonal elements. 
Thus, there must exist $i, j \in [b]$ with $i\ne j$, such that $[\Gamma^\top D \Gamma]_{ij} \ne 0$.
Note that 
both $\tau_U^{-1} (I-\varphi_U \Lambda)^{-1}$ and $\tilde{\tau}_U^{-1} (I-\tilde{\varphi}_U \Lambda)^{-1}$ are diagonal matrices. 
From \eqref{eq:cond_Y_Z_eq}, we immediately have 
$\tau_{\epsilon}^{-1} [\Gamma^\top D \Gamma]_{ij} = \tilde{\tau}_{\epsilon}^{-1} [\Gamma^\top D \Gamma]_{ij}$. 
Because $[\Gamma^\top D \Gamma]_{ij} \ne 0$, this further implies that $\tau_{\epsilon} = \tilde{\tau}_{\epsilon}$.

Third, we prove that $\tau_U = \tilde{\tau}_U$ and $\varphi_U = \tilde{\varphi}_U$. 
Because $\tau_{\epsilon} = \tilde{\tau}_{\epsilon}$, 
from \eqref{eq:cond_Y_Z_eq}, 
\begin{align}\label{eq:proof_car_tau_phi_U}
    & \tau_U (I-\varphi_U \Lambda) = \tilde{\tau}_U (I-\tilde{\varphi}_U \Lambda) 
    \nonumber
    \\
    \Longrightarrow \ & 
    (\tau_U - \tilde{\tau}_U) I = ( \tau_U \varphi_U - \tilde{\tau}_U \tilde{\varphi}_U ) \Lambda
    \ \Longrightarrow \ 
    (\tau_U - \tilde{\tau}_U) I = ( \tau_U \varphi_U - \tilde{\tau}_U \tilde{\varphi}_U ) \Gamma \Lambda \Gamma^\top 
    \nonumber
    \\
    \Longrightarrow \ & 
    (\tau_U - \tilde{\tau}_U) I = ( \tau_U \varphi_U - \tilde{\tau}_U \tilde{\varphi}_U ) D^{-1/2} W D^{-1/2}
    \ \Longrightarrow \  
    (\tau_U - \tilde{\tau}_U) D = ( \tau_U \varphi_U - \tilde{\tau}_U \tilde{\varphi}_U ) W.
\end{align}
The diagonal elements of $D$ are positive and all of the diagonal elements of $W$ are zero. From \eqref{eq:proof_car_tau_phi_U}, we can then derive that $\tau_U-\tilde{\tau}_U =0$. 
Moreover, $W$ must have a nonzero off-diagonal element; otherwise all entries of $D$ are zero, violating our assumption that all diagonal elements of $D$ are positive in \S\ref{sec:car-papadogeorgu}. 
Because $D$ is diagonal, 
from \eqref{eq:proof_car_tau_phi_U}, then $\tau_U \varphi_U - \tilde{\tau}_U \tilde{\varphi}_U = 0$. 
Because $\tilde{\tau}_U = \tau_U>0$,  $\varphi_U = \tilde{\varphi}_U$.

From the above, $(\tau_U, \varphi_U, \tau_{\epsilon}) = (\tilde{\tau}_U, \tilde{\varphi}_U, \tilde{\tau}_{\epsilon})$, i.e., Lemma \ref{lem:cond_Y_Z}(i) holds. 
\end{proof}

\begin{proof}[Proof of Lemma \ref{lem:cond_Y_Z}(ii)]
We then consider the case where $D_{[b]}$ is a scaled identity matrix for all $1\le b \le B$. 
Let $D_{[b]} = d_{[b]} I$ with $d_{[b]}>0$ for $1\le b \le B$. Then $ \Gamma_{[b]}^\top D_{[b]} \Gamma_{[b]} = d_{[b]}I$ for all $1\le b \le B$.
Suppose that for some $1\le m \le B$, 
$\Lambda_{[m]}$ contains at least three distinct diagonal elements. 
We introduce $A$ to denote the quantity in \eqref{eq:cond_Y_Z_eq}, with $A_{[b]}$ being the quantity in \eqref{eq:cond_Y_Z_eq_k} for $1\le b \le B$.

First, we show that $\varphi_{U}' \ne 0$. We prove this by contradiction.  
If $\varphi_{U}' = 0$, 
from \eqref{eq:cond_Y_Z_eq_k}, 
$$A_{[m]} = \tilde{\tau}_U^{-1} (I-\tilde{\varphi}_U \Lambda_{[m]})^{-1} 
    +\tilde{\tau}_{\epsilon}^{-1} \Gamma_{[m]}^\top D_{[m]} \Gamma_{[m]} = \tilde{\tau}_U^{-1} I 
    +\tilde{\tau}_{\epsilon}^{-1} d_{[m]} I = (\tilde{\tau}_U^{-1}+\tilde{\tau}_{\epsilon}^{-1} d_{[m]})I$$ 
will be a scaled identity matrix. 
However, because $\varphi_U \ne 0$ and $\Lambda_{[m]}$ has three distinct diagonal elements, 
from \eqref{eq:cond_Y_Z_eq_k}, 
$$
A_{[b]} = \tau_U^{-1} (I-\varphi_U \Lambda_{[b]})^{-1} 
    +\tau_{\epsilon}^{-1} \Gamma_{[b]}^\top D_{[b]} \Gamma_{[b]}
= \tau_U^{-1} (I-\varphi_U \Lambda_{[b]})^{-1} + \tau_{\epsilon}^{-1} d_{[b]} I
$$
must have distinct diagonal elements, leading to a contradiction. 

Second, we try to solve $\tau_{\epsilon}$ from $A_{[m]}, \Lambda_{[m]}$ and $d_{[m]}$. 
For any $i\in [m]$, from \eqref{eq:cond_Y_Z_eq_k}, we have 
\begin{align*}
    A_{ii} = \tau_U^{-1} (1- \varphi_U \lambda_i)^{-1} + \tau_\epsilon^{-1} d_{[m]}
    \ \Longleftrightarrow \ 
    \tau_U( A_{ii} - \tau_\epsilon^{-1} d_{[m]} ) = (1- \varphi_U \lambda_i)^{-1}, 
\end{align*}
where $\lambda_i$ is the diagonal element of $\Lambda$ corresponding to location $i$. 
For any $i,j\in [m]$, we then have 
\begin{align*}
    & \frac{A_{ii} - \tau_\epsilon^{-1} d_{[m]}}{A_{jj} - \tau_\epsilon^{-1} d_{[m]}}
    = 
    \frac{1- \varphi_U \lambda_j}{1- \varphi_U \lambda_i}
    \\
    \Longleftrightarrow \  & 
    A_{ii} - A_{jj} = \varphi_U  \left\{ \lambda_iA_{ii}- \lambda_j A_{jj} - 
    (\lambda_i -\lambda_j) \tau_\epsilon^{-1} d_{[m]} 
    \right\}
\end{align*}
Using our assumption that $\varphi_U \neq 0$, for any $i,j, s, t\in [m]$ such that $\lambda_i \ne \lambda_j$ and $\lambda_s \ne \lambda_t$, $A_{ii}\ne A_{jj}$, $A_{ss}\ne A_{tt}$, and thus
\begin{align}\label{eq:solve_tau_eps}
    \frac{A_{ii} - A_{jj}}{A_{ss} - A_{tt}} = 
    \frac{\lambda_iA_{ii}- \lambda_j A_{jj} - 
    (\lambda_i -\lambda_j) \tau_\epsilon^{-1} d_{[m]} 
    }{
    \lambda_sA_{ss}- \lambda_t A_{tt} - 
    (\lambda_s -\lambda_t) \tau_\epsilon^{-1} d_{[m]} 
    }.
\end{align}
Note that $d_{[m]}$ is positive. 
From \eqref{eq:solve_tau_eps}, $\tau_{\epsilon}$ can be uniquely determined by $A_{[m]}, \Lambda_{[m]}$ and $d_{[m]}$ as long as 
\begin{align*}
    (A_{ii}-A_{jj})(\lambda_s-\lambda_t) \ne (A_{ss} - A_{tt}) (\lambda_i -\lambda_j) 
    \Longleftrightarrow
    \frac{A_{ii}-A_{jj}}{\lambda_i -\lambda_j} \ne \frac{A_{ss} - A_{tt}}{\lambda_s-\lambda_t}. 
\end{align*}

Third, we prove that $\tau_{\epsilon} = \tilde{\tau}_{\epsilon}$. 
We prove this by contradiction. 
Assume that $\tau_{\epsilon} \ne \tilde{\tau}_{\epsilon}$. 
Let $i,j,t$ be the three indices in $[m]$ such that $\lambda_i, \lambda_j$ and $\lambda_t$ are mutually different. 
From the discussion before, 
\begin{align*}
    & \frac{A_{ii}-A_{jj}}{\lambda_i -\lambda_j} = \frac{A_{ii} - A_{tt}}{\lambda_i-\lambda_t}
    = 
    \frac{A_{tt}-A_{jj}}{\lambda_t -\lambda_j} \\
    \Longrightarrow \ & 
    \frac{(1- \varphi_U \lambda_i)^{-1} - (1- \varphi_U \lambda_j)^{-1}}{\lambda_i -\lambda_j}
    = 
    \frac{(1- \varphi_U \lambda_i)^{-1} - (1- \varphi_U \lambda_t)^{-1}}{\lambda_i -\lambda_t}
    = 
    \frac{(1- \varphi_U \lambda_t)^{-1} - (1- \varphi_U \lambda_j)^{-1}}{\lambda_t -\lambda_j}.
\end{align*}
Denote the common value of the above three expressions as $b_1$. Then:
\begin{align*} &b_1(\lambda_i - \lambda_j) = (1-\varphi_U \lambda_i)^{-1} - (1-\varphi_U \lambda_j)^{-1}\\
\implies & b_1 \lambda_i - (1-\lambda_U \lambda_i)^{-1} = b_1\lambda_j - (1-\lambda_U \lambda_j)^{-1}.
\end{align*} 
By the same logic, we have 
\begin{align*}
    b_1 \lambda_i - (1-\lambda_U \lambda_i)^{-1} = b_1\lambda_j - (1-\lambda_U \lambda_j)^{-1}
    = b_1\lambda_t - (1- \varphi_u \lambda_t)^{-1}. 
\end{align*}
Denote the common value of the above three expressions as $-b_0$.
Then, for any $\lambda \in \{\lambda_i, \lambda_j, \lambda_t\}$, we have 
\begin{align}\label{eq:equal_tau_epsilon_lambda}
    (1- \varphi_U \lambda)^{-1} = b_0 + b_1 \lambda
    \ \Longrightarrow \ 
    (b_0 + b_1 \lambda)  (1- \varphi_U \lambda) = 1. 
\end{align}
However, \eqref{eq:equal_tau_epsilon_lambda}, as an quadratic equation for $\lambda$, has at most two distinct roots, leading to a contradiction. 
Therefore, $\tau_{\epsilon} = \tilde{\tau}_{\epsilon}$. 

Fourth, we prove that $\tau_U = \tilde{\tau}_U$ and $\varphi_U = \tilde{\varphi}_U$. 
The proof then follows by the same logic as the third step in the proof of Lemma \ref{lem:cond_Y_Z}(i). 

From the above, $(\tau_U, \varphi_U, \tau_{\epsilon}) = (\tilde{\tau}_U, \tilde{\varphi}_U, \tilde{\tau}_{\epsilon})$, i.e., Lemma \ref{lem:cond_Y_Z}(ii) holds. 
\end{proof}

\begin{lemma}\label{lem:mean_Y_Z}
Define $W, D, W_{[b]}$s and $D_{[b]}$s
as in \S\ref{sec:car-papadogeorgu}, 
and recall the definition of $\Lambda_{[b]}$s,
$\Gamma_{[b]}$s, $\Lambda$ and $\Gamma$ at the beginning of \S\ref{sec:car-proofs}.
Consider any given $\tau_U>0$ and $\varphi_U \ne 0$ such that $I-\varphi_U \Lambda$ is positive definite. 
Let $(\beta, \rho, \tau_Z)$ and $(\tilde{\beta}, \tilde{\rho}, \tilde{\tau}_Z)$ be two sets of parameters with $\tau_Z$ and $\tilde{\tau}_Z$ being positive. 
If 
\begin{align}\label{eq:mean_Y_Z}
    \beta I + \rho \sqrt{\tau_Z/\tau_U} (I-\varphi_U \Lambda)^{-1} 
    = 
    \tilde{\beta} I + \tilde{\rho} \sqrt{\tilde{\tau}_Z/\tau_U} (I-\varphi_U \Lambda)^{-1}, 
\end{align}
then $\beta = \tilde{\beta}$ and $\rho \sqrt{\tau_Z} = \tilde{\rho} \sqrt{\tilde{\tau}_Z}$. 
\end{lemma}

\begin{proof}[Proof of Lemma \ref{lem:mean_Y_Z}]
From \eqref{eq:mean_Y_Z}, we have 
\begin{align*}
    & (\beta - \tilde{\beta}) I 
    = 
    \Big(\tilde{\rho} \sqrt{\tilde{\tau}_Z/\tau_U}- \rho \sqrt{\tau_Z/\tau_U} \Big) (I-\varphi_U \Lambda)^{-1}
    \\
    \Longrightarrow \ & 
    (\beta - \tilde{\beta}) (I-\varphi_U \Lambda)
    = \big( \tilde{\rho} \sqrt{\tilde{\tau}_Z/\tau_U}- \rho \sqrt{\tau_Z/\tau_U} \big) I 
    \\
    \Longrightarrow \ & 
    \big\{ (\beta - \tilde{\beta}) - \big(\tilde{\rho} \sqrt{\tilde{\tau}_Z/\tau_U}- \rho \sqrt{\tau_Z/\tau_U} \big) \big\} I
    = \varphi_U (\beta - \tilde{\beta}) \Lambda
    \\
    \Longrightarrow \ & 
    \big\{ (\beta - \tilde{\beta}) - \big(\tilde{\rho} \sqrt{\tilde{\tau}_Z/\tau_U}- \rho \sqrt{\tau_Z/\tau_U} \big) \big\} I
    = \varphi_U (\beta - \tilde{\beta}) \Gamma \Lambda \Gamma^\top
    \\
    \Longrightarrow \ & 
    \big\{ (\beta - \tilde{\beta}) - \big(\tilde{\rho} \sqrt{\tilde{\tau}_Z/\tau_U}- \rho \sqrt{\tau_Z/\tau_U} \big) \big\} I
    = \varphi_U (\beta - \tilde{\beta}) D^{-1/2} W D^{-1/2}
    \\
    \Longrightarrow \ & 
    \big\{ (\beta - \tilde{\beta}) - \big(\tilde{\rho} \sqrt{\tilde{\tau}_Z/\tau_U}- \rho \sqrt{\tau_Z/\tau_U} \big) \big\} D
    = \varphi_U (\beta - \tilde{\beta}) W.
\end{align*}
The matrix
$D$ must have positive diagonal elements, and consequently $W$ must have nonzero off-diagonal elements. 
These imply that $(\beta - \tilde{\beta}) - ( \rho \sqrt{\tau_Z/\tau_U} - \tilde{\rho} \sqrt{\tilde{\tau}_Z/\tau_U} ) = \varphi_U (\beta - \tilde{\beta}) = 0$.
Because $\varphi_U \ne 0$ and $\tau_U>0$,
we have 
$\beta = \tilde{\beta}$ and $\rho \sqrt{\tau_Z} = \tilde{\rho} \sqrt{\tilde{\tau}_Z}$. 
Therefore, Lemma \ref{lem:mean_Y_Z} holds. 
\end{proof}

\begin{lemma}\label{lem:prec_Z}
Define $W, D, W_{[b]}$s and $D_{[b]}$s
as in \S\ref{sec:car-papadogeorgu}, 
and recall the definition of $\Lambda_{[b]}$s,
$\Gamma_{[b]}$s, $\Lambda$ and $\Gamma$ at the beginning of \S\ref{sec:car-proofs}. 
Consider any given $\varphi_U$ such that $I - \varphi_U \Lambda$ is positive definite. 
Let $(\tau_Z, \varphi_Z, \rho)$ and $(\tilde{\tau}_Z, \tilde{\varphi}_Z, \tilde{\rho})$ be two sets of parameters such that $\tau_Z>0, \tilde{\tau}_Z>0$ and $\rho \sqrt{\tau_Z} = \tilde{\rho} \sqrt{\tilde{\tau}_Z}$. 
If 
\begin{align}\label{eq:prec_Z_lemma}
     \tau_Z \big\{ I - \varphi_Z \Lambda - \rho^2 (I - \varphi_U \Lambda)^{-1} \big\}
     = 
    \tilde{\tau}_Z \big\{ I - \tilde{\varphi}_Z \Lambda - \tilde{\rho}^2 (I - \varphi_U \Lambda)^{-1} \big\}, 
\end{align}
then $\tau_Z=\tilde{\tau}_Z, \varphi_Z=\tilde{\varphi}_Z$, and $\rho = \tilde{\rho}$. 
\end{lemma}

\begin{proof}[Proof of Lemma \ref{lem:prec_Z}]
Note that 
\begin{align*}
    \tr(\Lambda) & = \tr(\Gamma \Lambda \Gamma^\top) = \tr(D^{-1/2} W D^{-1/2}) = \tr( W D^{-1}) = 0,
\end{align*}
where the last equality holds because all the diagonal elements of $W$ are zero and $D$ is a diagonal matrix. 
From \eqref{eq:prec_Z_lemma}, we have 
\begin{align*}
    & \tau_Z \cdot \tr\big( I - \varphi_Z \Lambda - \rho^2 (I - \varphi_U \Lambda)^{-1} \big)
     = 
    \tilde{\tau}_Z \cdot \tr\big( I - \tilde{\varphi}_Z \Lambda - \tilde{\rho}^2 (I - \varphi_U \Lambda)^{-1} \big)
    \\
    \Longrightarrow \ 
    & 
    n \tau_Z - (\rho \sqrt{\tau_Z})^2 \cdot \tr\big( (I - \varphi_U \Lambda)^{-1} \big)
    = 
    n \tilde{\tau}_Z - \left(\tilde{\rho} \sqrt{\tilde{\tau}_Z}\right)^2 \cdot \tr\big( (I - \varphi_U \Lambda)^{-1} \big).
\end{align*}
Because $\rho \sqrt{\tau_Z} = \tilde{\rho} \sqrt{\tilde{\tau}_Z}$, $\tau_Z=\tilde{\tau}_Z$, which further implies that $\rho = \tilde{\rho}$. 
From \eqref{eq:prec_Z_lemma}, 
this then implies that $\varphi_Z \Lambda = \tilde{\varphi}_Z \Lambda$. 
Consequently, 
\begin{align*}
     (\varphi_Z - \tilde{\varphi}_Z) \Lambda = 0
    & \Longrightarrow (\varphi_Z - \tilde{\varphi}_Z) \Gamma \Lambda \Gamma^\top = 0
    \Longrightarrow (\varphi_Z - \tilde{\varphi}_Z) D^{-1/2} W D^{-1/2} = 0
    \\
    & \Longrightarrow (\varphi_Z - \tilde{\varphi}_Z) W = 0 
    \Longrightarrow \varphi_Z = \tilde{\varphi}_Z, 
\end{align*}
where the last step holds because $W$ has nonzero off-diagonal elements. 
Therefore, Lemma \ref{lem:prec_Z} holds. 
\end{proof}

\begin{proof}[Proof of Theorem \ref{thm:car-identifiability}]

Recall the definition of $W, D, W_{[b]}$s and $D_{[b]}$s in \S\ref{sec:car-papadogeorgu}, 
and the definition of $\Lambda_{[b]}$s,
$\Gamma_{[b]}$s, $\Lambda$ and $\Gamma$ at the beginning of \S\ref{sec:car-proofs}. 
We further introduce 
$\tau_\epsilon$ to denote $\sigma_\epsilon^{-2}$.
By definition, we have
$D^{-1/2} W D^{-1/2} = \Gamma \Lambda \Gamma^\top$. 
Consequently, 
\begin{align}\label{eq:DW_simp_U}
    \tau_U (D-\varphi_U W)
    & = 
    D^{1/2} \cdot \tau_U (I- \varphi_U D^{-1/2} W D^{-1/2}) \cdot D^{1/2}
    \nonumber
    \\
    & = 
    D^{1/2} \cdot \tau_U (I- \varphi_U \Gamma \Lambda \Gamma^\top) \cdot D^{1/2}
    \nonumber
    \\
    & = 
    D^{1/2} \Gamma \cdot \tau_U (I- \varphi_U \Lambda ) \cdot \Gamma^\top D^{1/2}, 
\end{align}
and by the same logic, 
\begin{align}\label{eq:DW_simp_Z}
    \tau_Z (D-\varphi_Z W)
    & = 
    D^{1/2} \Gamma \cdot \tau_Z (I- \varphi_Z \Lambda ) \cdot \Gamma^\top D^{1/2}. 
\end{align}

First, we consider the joint covariance matrix in \eqref{eq:Sigma_SP}. 
From \eqref{eq:Sigma_SP} and the block matrix inversion formula, we have 
\begin{align*}
\begin{pmatrix}
        \Sigma_{UU} & \Sigma_{UZ}\\
        \Sigma_{ZU} & \Sigma_{ZZ}
\end{pmatrix}
& = 
\begin{pmatrix} 
        \tau_U(D-\varphi_U W) & -\rho\sqrt{\tau_U\tau_Z} D \\
        -\rho\sqrt{\tau_U\tau_Z} D & \tau_Z(D-\varphi_Z W)
    \end{pmatrix}^{-1}\\
& = 
\begin{pmatrix} 
* \ & * \\
* \ & 
[ \tau_Z(D-\varphi_Z W) -  \rho^2\tau_Z\tau_U D \{ \tau_U(D-\varphi_U W) \}^{-1} D ]^{-1}
\end{pmatrix}^{-1},
\end{align*} 
and  
\begin{align*}
    & \quad \ 
    \begin{pmatrix} 
        \tau_U(D-\varphi_U W) & -\rho\sqrt{\tau_U\tau_Z} D \\
        -\rho\sqrt{\tau_U\tau_Z} D & \tau_Z(D-\varphi_Z W)
    \end{pmatrix}
    = 
    \begin{pmatrix}
        \Sigma_{UU} & \Sigma_{UZ}\\
        \Sigma_{ZU} & \Sigma_{ZZ}
    \end{pmatrix}^{-1}
    \\
    & = 
     \begin{pmatrix}
        (\Sigma_{UU} - \Sigma_{UZ} \Sigma_{ZZ}^{-1} \Sigma_{ZU} )^{-1} & - (\Sigma_{UU} - \Sigma_{UZ} \Sigma_{ZZ}^{-1} \Sigma_{ZU} )^{-1} \Sigma_{UZ} \Sigma_{ZZ}^{-1} \\
        *
         & 
         * 
    \end{pmatrix}.
\end{align*} 
These imply that 
\begin{align}
\label{eq:cond_dist_YZ_1}
    \Sigma_{ZZ}^{-1} 
    & = \tau_Z(D-\varphi_Z W) -  \rho^2\tau_Z\tau_U D \{ \tau_U(D-\varphi_U W) \}^{-1} D, 
    \\
\label{eq:cond_dist_YZ_2}
    \Sigma_{UU} - \Sigma_{UZ} \Sigma_{ZZ}^{-1} \Sigma_{ZU}
    & = \tau_U^{-1}(D-\varphi_U W)^{-1},
    \\
\label{eq:cond_dist_YZ_3}
    \Sigma_{UZ} \Sigma_{ZZ}^{-1} 
    & = 
    \tau_U^{-1}(D-\varphi_U W)^{-1} \rho\sqrt{\tau_U\tau_Z} D. 
\end{align}

Second, we simplify the expressions in  \eqref{eq:condvar} that 
determine the distribution of the observed data $(Y,Z)$. 
The precision matrix for $Z$ is
\begin{align*}
     \var(Z)^{-1} & = \Sigma_{ZZ}^{-1}
     = \tau_Z(D-\varphi_Z W) - \rho^2\tau_Z\tau_U D \{ \tau_U(D-\varphi_U W) \}^{-1} D
     \\
    & = 
    D^{1/2} \Gamma \cdot \tau_Z (I- \varphi_Z \Lambda ) \cdot \Gamma^\top D^{1/2}
    - \rho^2\tau_Z
    D^{1/2}\Gamma 
     (I- \varphi_U \Lambda )^{-1} \Gamma^\top D^{1/2}\\
    & = D^{1/2} \Gamma \cdot \tau_Z 
     \big\{ I - \varphi_Z \Lambda - \rho^2 (I - \varphi_U \Lambda)^{-1} \big\} \cdot \Gamma^\top D^{1/2},
\end{align*}
where the second equality follows from \eqref{eq:cond_dist_YZ_1}, and 
the third equality follows from \eqref{eq:DW_simp_U} and \eqref{eq:DW_simp_Z}.  
The conditional variance of $Y$ given $Z$ is 
\begin{align*}
    \var (Y\mid Z) &= \Sigma_{UU} - \Sigma_{UZ} \Sigma_{ZZ}^{-1}\Sigma_{ZU} + \tau_\epsilon^{-1}  I
    = \tau_U^{-1}(D-\varphi_U W)^{-1} + \tau_\epsilon^{-1} I \\
    & = 
    D^{-1/2} \Gamma 
    \tau_U^{-1} (I- \varphi_U \Lambda )^{-1} \Gamma^\top D^{-1/2} + \tau_\epsilon^{-1} I
    \\
    &= D^{-1/2}\Gamma \cdot
    \big\{ \tau_U^{-1} (I-\varphi_U \Lambda)^{-1} 
    +\tau_{\epsilon}^{-1} \Gamma^\top D \Gamma
    \big\} \cdot
    \Gamma^{\top} D^{-1/2}, 
\end{align*}
where the second equality follows from \eqref{eq:cond_dist_YZ_2} and the third equality follows from \eqref{eq:DW_simp_U}. 
The coefficient from the conditional mean of $Y$ given $Z$ is 
\begin{align*}
    \cov(Y,Z) \var(Z)^{-1} &
     = \beta I + \Sigma_{UZ}\Sigma_{ZZ}^{-1}
    = \beta I + \rho\sqrt{\tau_Z/\tau_U} (D - \varphi_U W)^{-1} D\\
    & = 
    \beta I + \rho\sqrt{\tau_Z/\tau_U}
    D^{-1/2} \Gamma (I-\varphi_U \Lambda)^{-1} \Gamma^\top D^{-1/2} D
    \\
    &=D^{-1/2} \Gamma
    \big\{ \beta I +\rho \sqrt{\tau_Z/\tau_U} (I-\varphi_U \Lambda)^{-1} \big\} 
    \Gamma^{\top} D^{1/2}, 
\end{align*}
where the second equality follows from \eqref{eq:cond_dist_YZ_3}.

Third, we consider the identifiability of the model parameters. 
From the discussion before and noting that $D$ and $\Gamma$ are known from the observed proximity matrix $W$, the observed data distribution can  identify 
\begin{equation*}%
    \beta I + \rho \sqrt{\tau_Z/\tau_U} (I-\varphi_U \Lambda)^{-1}, 
    \ \ 
    \tau_U^{-1} (I-\varphi_U \Lambda)^{-1} 
    +\tau_{\epsilon}^{-1} \Gamma^\top D \Gamma, 
    \ \ 
    \tau_Z 
     \big\{ I - \varphi_Z \Lambda - \rho^2 (I - \varphi_U \Lambda)^{-1} \big\}.
\end{equation*}
Let $(\beta, \tau_Z, \varphi_Z, \tau_U, \varphi_U, \tau_{\epsilon}, \rho)$ denote the true data generating parameters, and consider any parameters $(\tilde{\beta}, \tilde{\tau}_Z, \tilde{\varphi}_Z, \tilde{\tau}_U, \tilde{\varphi}_U, \tilde{\tau}_{\epsilon}, \tilde{\rho})$ leading 
to 
the same observed data distribution. Then 
\begin{align}
    \beta I + \rho \sqrt{\tau_Z/\tau_U} (I-\varphi_U \Lambda)^{-1} & = 
    \tilde{\beta} I + \tilde{\rho} \sqrt{\tilde{\tau}_Z/\tilde{\tau}_U} (I-\tilde{\varphi}_U \Lambda)^{-1}, 
    \label{eq:iden_cond_mean_Y_Z}
    \\ 
    \tau_U^{-1} (I-\varphi_U \Lambda)^{-1} 
    +\tau_{\epsilon}^{-1} \Gamma^\top D \Gamma
    & = 
    \tilde{\tau}_U^{-1} (I-\tilde{\varphi}_U \Lambda)^{-1} 
    +\tilde{\tau}_{\epsilon}^{-1} \Gamma^\top D \Gamma, 
    \label{eq:iden_var_mean_Y_Z}
    \\ 
    \tau_Z \big\{ I - \varphi_Z \Lambda - \rho^2 (I - \varphi_U \Lambda)^{-1} \big\}
    & = 
    \tilde{\tau}_Z \big\{ I - \tilde{\varphi}_Z \Lambda - \tilde{\rho}^2 (I - \tilde{\varphi}_U \Lambda)^{-1} \big\}.
    \label{eq:prec_Z}
\end{align}
Below we prove that these two sets of parameters must be identical. 

First, we prove that, under the conditions in Theorem \ref{thm:car-identifiability}, one of the following two conditions must hold:  
\begin{enumerate}[label=(\roman*)]
    \item $D_{[b]}$ contains distinct diagonal elements for some $1\le b \le B$, 
    \item $D_{[b]}$ is a scaled identity matrix for all $1\le b \le B$, 
    and 
    $\Lambda_{[b]}$ has at least three distinct diagonal elements for some $1\le b \le B$.
\end{enumerate}
Suppose that (i) fails. In this case, 
$D_{[b]}$ is a scaled identity matrix for all $1\le b \le B$. 
Moreover, by the conditions in Theorem \ref{thm:car-identifiability}, $W_{[b]}$ has at least three distinct eigenvalues for some $1\le b \le B$. 
Let $D_{[b]} = d_{[b]} I$ for all $1 \le b \le B$. Then
$\Gamma_{[b]} \Lambda_{[b]} \Gamma_{[b]}^\top = D_{[b]}^{-1/2} W_{[b]} D_{[b]}^{-1/2} = d_{[b]}^{-1} W_{[b]}$ for all $b$. 
This implies that, for all $b$, the diagonal elements of $d_{[b]}\Lambda_{[b]}$ are the eigenvalues of $W_{[b]}$. 
Because $W_{[b]}$ has at least three distinct eigenvalues for some $1\le b\le B$, we can know that $\Lambda_{[b]}$ has at least three distinct diagonal elements for some $1\le b \le B$. Therefore, when (i) fails, (ii) must hold.

Second, we prove $(\tau_U, \varphi_U, \tau_{\epsilon}) = (\tilde{\tau}_U, \tilde{\varphi}_U, \tilde{\tau}_{\epsilon})$. 
This follows immediately from \eqref{eq:iden_var_mean_Y_Z} and Lemma \ref{lem:cond_Y_Z}. 

Third, we prove $\beta = \tilde{\beta}$ and $\rho \sqrt{\tau_Z} = \tilde{\rho} \sqrt{\tilde{\tau}_Z}$. This follows immediately from \eqref{eq:iden_cond_mean_Y_Z} and Lemma \ref{lem:mean_Y_Z}.

Fourth, we prove $(\tau_Z, \varphi_Z, \rho)=(\tilde{\tau}_Z, \tilde{\varphi}_Z, \tilde{\rho})$. This follows immediately from \eqref{eq:prec_Z} and Lemma \ref{lem:prec_Z}. 

From the above, Theorem \ref{thm:car-identifiability} holds.  
\end{proof}

\subsection{Proof of Corollary \ref{cor:car-identifiability-binary}}

To prove Corollary \ref{cor:car-identifiability-binary}, 
we need the following lemma.

\begin{lemma}\label{lemma:three-eigenvalues}
Consider any integer $n\ge 4$, 
and any $n\times n$ symmetric matrix $W$, whose diagonal elements are all zero and off-diagonal elements are either $0$ or $1$.  
Let $D$ be an $n\times n$ diagonal matrix, whose $i$th diagonal element is the sum of the $i$th column of $W$, for $1\le i \le n$. 
If 
\begin{enumerate}[label={(\roman*)}, topsep=1ex,itemsep=-0.3ex,partopsep=1ex,parsep=1ex]
    \item for any $1\le i\ne j \le n$, there exist $(i_1, i_2, \ldots, i_l)$ such that $i_1 = i, i_l=j$, and $W_{i_t i_{t+1}} = 1$ for all $1\le t\le l-1$, 
    i.e., $i$ and $j$ are connected either directly or indirectly, 
    \item $D$ is a scaled identity matrix, in the sense that  $D = m I_n$ for $2 \le  m \le n-2$,
\end{enumerate}
then 
$D^{-1/2} W D^{-1/2}$ has at least three distinct eigenvalues. 
\end{lemma}

\begin{proof}[Proof of Lemma \ref{lemma:three-eigenvalues}]
Assume that conditions (i) and (ii) in Lemma \ref{lemma:three-eigenvalues} hold. Then  $D^{-1/2} W D^{-1/2} = m^{-1} W$. 
Thus, it suffices to show that $W$ has at least three distinct eigenvalues.

First, we prove $m$ is an eigenvalue of $W$ and its multiplicity is $1$. 
Let $u = (1, 1, \ldots, 1)^\top$. 
Since $D = mI_n$, by definition,  
$
Wu = m u, 
$
implying that $m$ is an eigenvalue of $W$. We now show that the multiplicity of $m$ is $1$. 
Suppose $v = (v_1,\ldots,v_n)^\top$ is an eigenvector of $W$ associated with the eigenvalue $m$, i.e., $
Wv = m v. 
$
Then, for every $1\le i\le n$, 
$m v_i = \sum_{j \in \partial_i} v_j$,  
where $\partial_i \equiv \{j: W_{ij} = 1, 1\le j \le n\}$. 
Because $D = mI_n$, then for any $1\le i\le n$, $\partial_i$ contains exactly $m$ elements. 
Let $v_{\max} = \max_{1\le i \le n} v_i$, and $a$ be an index such that $v_a = v_{\max}$. 
Since $m v_a = \sum_{j \in \partial_a} v_j \le m v_{\max} = m v_a$,  $v_j = v_{\max}$ for all $j \in \partial_a$. 
From condition (i) in Lemma \ref{lemma:three-eigenvalues}, for any $1\le b \le n$ and $b\ne a$, 
there exists a sequence $\{i_1,\ldots,i_l\}$ such that $i_1=a, i_l=b$, and $W_{i_t, i_{t+1}} =1$ for $1\le t\le l-1$. 
By iteratively applying the previous argument, we have 
$v_{\max} = v_a = v_{i_{1}} = v_{i_{2}} = \ldots = v_{i_{l}} = v_b$. 
Therefore, $v = (v_{\max}, \ldots, v_{\max})^\top = v_{\max} u$. 
Consequently, the multiplicity of the eigenvalue $m$ must be $1$.

Second, we prove by contradiction that $W$ has at least three distinct eigenvalues. 
Suppose that $W$ has only two distinct eigenvalues.
From the discussion before, 
one of the eigenvalues must be $m$ with multiplicity $1$. 
We denote the other eigenvalue by $\lambda$, which must have multiplicity $n-1$. 
Note that the trace of $W$ is zero, and it is also the same as the sum of all eigenvalues of $W$. 
Thus, 
$(n-1)\lambda + m = \tr (W) = 0$, 
implying that $\lambda = -m/(n-1)$. 
Consequently, the determinant of $W$ equals $\det(W) = m \lambda^{n-1} = m \{-(m/(n-1)\}^{n-1} = (-1)^{n-1} m^n/(n-1)^{n-1}$. 
Because all elements of $W$ are $0$ or $1$,  $\det(W)$ must be an integer. Therefore, 
$m^n/(n-1)^{n-1}$ must be an integer. 
Because $m<n-1$, there must exist a prime number $p$ such that 
the prime factorizations of $n-1$ and $m$ contain, respectively, $r_2$ and $r_1$ powers of $p$, for some integers $r_2 > r_1 \ge 0$. 
Because $m^n/(n-1)^{n-1}$ is an integer, 
$p^{nr_1}/p^{(n-1)r_2}$ must also be an integer, which implies that $nr_1 \ge (n-1)r_2$. 
This further implies that $(n-1)/n \le r_1/r_2 \le (r_2-1)/r_2$. Consequently, $n\le r_2$. 
However, 
because $r_2$ is a positive integer, we also have 
$n-1\ge p^{r_2} \ge 2^{r_2} \ge r_2$, leading to a contradiction. 
Therefore, $W$ must have at least three distinct eigenvalues. 

From the above, Lemma \ref{lemma:three-eigenvalues} holds. 
\end{proof}

\begin{proof}[Proof of Corollary \ref{cor:car-identifiability-binary}]
From the proof of Theorem \ref{thm:car-identifiability},
it suffices to prove that one of the following two conditions must hold: 
\begin{enumerate}[label=(\roman*)]
    \item $D_{[b]}$ contains distinct diagonal elements for some $1\le b \le B$, 
    \item $D_{[b]}$ is a scaled identity matrix for all $1\le b \le B$, 
    and 
    $\Lambda_{[b]}$ has at least three distinct diagonal elements for some $1\le b \le B$.
\end{enumerate} 

Suppose that condition (i) fails. In the following, we prove that condition (ii) must hold. When condition (i) fails, $D_{[b]}$ is a scaled identity matrix for all $1\le b \le B$.
From the conditions in Corollary \ref{cor:car-identifiability-binary}, 
there exists some $b\in \{1,2,\ldots,B\}$ such that the locations in the connected component $\mathcal{M}_b$ are not fully connected. 
Let $m$ be the positive integer such that $D_{[b]} = m I$, and $n_{[b]}$ be the size of $\mathcal{M}_b$.

We first show that 
$\mathcal{M}_b$ must contain at least four locations. 
First, $\mathcal{M}_b$ has at least two locations, because all the diagonal elements of $D$ are positive as assumed in \S \ref{sec:car-papadogeorgu}. 
Second, $\mathcal{M}_b$ cannot have exactly two locations; otherwise, it will be fully connected. 
Third, $\mathcal{M}_b$ cannot have exactly three locations. 
This is because the only way for the three locations to be indirectly connected without being fully connected is that one location is directly connected to the other two, while those two are not directly connected to each other, under which $D_{[b]}$ cannot be a scaled identity. 
These then imply that $\mathcal{M}_b$ must contain at least four locations. 

Second, we prove that $2 \le  m \le n_{[b]}-2$. 
By definition, $1\le m \le n_{[b]}-1$. 
Because locations in $\mathcal{M}_b$ are not fully connected, we must have $m \ne n_{[b]}-1$. 
Because there are two locations, say, $i$ and $j$, that are indirectly connected, each of the locations on the path between $i$ and $j$ must have at least two neighbors. Thus, $m\ne 1$. 
These then imply that $2 \le  m \le n_{[b]}-2$.

From Lemma \ref{lemma:three-eigenvalues} and the discussion before, 
$D_{[b]}^{-1/2} W_{[b]} D_{[b]}^{-1/2}$ must have at least three distinct eigenvalues. Consequently, $\Lambda_{[b]}$ has at least three distinct diagonal elements. Thus, condition (ii) must hold. 

From the above, we derive Corollary \ref{cor:car-identifiability-binary}. 
\end{proof}

\subsection{A remark on checking whether $\varphi_U=0$ based on the observed data distribution}

We first prove that $\Lambda$ contains at least two distinct diagonal elements. We prove this by contradiction. 
Suppose that $\Lambda = \tilde{\lambda} I$ for some $\tilde{\lambda}$. Then, by definition, 
$D^{-1/2} W D^{-1/2} = \Gamma \Lambda \Gamma^\top = \tilde{\lambda} \Gamma \Gamma^\top = \tilde{\lambda} I$. 
This implies that 
$W = \tilde{\lambda} D$.
Because $D$ is a diagonal matrix with positive diagonal elements, and $W$ is a matrix with zero diagonal elements and some nonzero off-diagonal elements, $W = \tilde{\lambda} D$ cannot be hold for $\tilde{\lambda}$, leading to a contradiction. 

We then prove that $\cov(Y,Z) \var(Z)^{-1}$ is a scaled identity matrix if and only if $\varphi_U=0$ or $\rho = 0$. 
From the proof of Theorem \ref{thm:car-identifiability}, 
\begin{align*}
    \cov(Y,Z) \var(Z)^{-1} &
     =D^{-1/2} \Gamma
    \big\{ \beta I +\rho \sqrt{\tau_Z/\tau_U} (I-\varphi_U \Lambda)^{-1} \big\} 
    \Gamma^{\top} D^{1/2}. 
\end{align*}
We can know that $\cov(Y,Z) \var(Z)^{-1}$ is a scaled identity matrix if and only if 
\begin{align*}
    & D^{-1/2} \Gamma
    \big\{ \beta I +\rho \sqrt{\tau_Z/\tau_U} (I-\varphi_U \Lambda)^{-1} \big\} 
    \Gamma^{\top} D^{1/2}
    = aI \text{ for some } a
    \\
    \Longleftrightarrow \quad  
    & 
    \beta I +\rho \sqrt{\tau_Z/\tau_U} (I-\varphi_U \Lambda)^{-1} 
    = aI \text{ for some } a
    \\
    \Longleftrightarrow \quad  
    & 
    \beta I +\rho \sqrt{\tau_Z/\tau_U} (I-\varphi_U \Lambda)^{-1} 
    = aI \text{ for some } a
    \\
    \Longleftrightarrow \quad  
    & 
    I-\varphi_U \Lambda
    = aI \text{ for some } a, \text{ or } \rho = 0\\
    \Longleftrightarrow \quad  
    & 
    \varphi_U = 0 \text{ or } \rho = 0. 
\end{align*}

\subsection{A remark on the necessity of $\varphi_U\neq 0$}\label{sec:nece_phiU_car}
Let $(\beta, \tau_Z, \varphi_Z, \tau_U, \varphi_U, \tau_{\epsilon}, \rho)$ be the true data-generating parameters, where we again use $\tau_{\epsilon}$ to denote $\sigma_{\epsilon}^{-2}$ for convenience. 
Below we show that when $\varphi_U = 0$, the model parameters are not identifiable. 
From the proof of Theorem \ref{thm:car-identifiability}, 
to prove nonidentifiability, 
it suffices to find parameters 
$(\tilde{\beta}, \tilde{\tau}_Z, \tilde{\varphi}_Z, \tilde{\tau}_U, \tilde{\varphi}_U, \tilde{\tau}_{\epsilon}, \tilde{\rho})$
that are different from the true data-generating parameters and yield the same values as in \eqref{eq:iden_cond_mean_Y_Z}--\eqref{eq:prec_Z}. 
Let $\tilde{\varphi}_U=\varphi_U=0$, $\tilde{\tau}_U = \tau_U$ and $\tilde{\tau}_\epsilon = \tau_\epsilon$. Then \eqref{eq:iden_cond_mean_Y_Z}--\eqref{eq:prec_Z} hold if and only if 
\begin{align*}
    \beta I + \rho \sqrt{\tau_Z/\tau_U} I & = 
    \tilde{\beta} I + \tilde{\rho} \sqrt{\tilde{\tau}_Z/\tau_U} I, 
    \\ 
    \tau_Z \big( I - \varphi_Z \Lambda - \rho^2 I \big)
    & = 
    \tilde{\tau}_Z \big( I - \tilde{\varphi}_Z \Lambda - \tilde{\rho}^2 I \big).
\end{align*}
These hold obviously when 
\begin{align}\label{eq:non_iden_car_phi_0}
\beta + \rho \sqrt{\tau_Z/\tau_U} & = 
    \tilde{\beta} + \tilde{\rho} \sqrt{\tilde{\tau}_Z/\tau_U}, 
    \quad 
\tau_Z (1-\rho^2) = \tilde{\tau}_Z (1-\tilde{\rho}^2), 
\quad 
\tau_Z \varphi_Z = \tilde{\tau}_Z \tilde{\varphi}_Z
\end{align}
For any $\delta > 1 - \rho^2 \ge 0$ and $\zeta \in \{-1,1\}$, define 
\begin{align*}
    \tilde{\tau}_Z = \delta \tau _Z, 
    \quad 
    \tilde{\varphi}_Z = \delta^{-1} \varphi_Z,
    \quad 
    \tilde{\rho} = \zeta \sqrt{\frac{\delta - (1-\rho^2)}{\delta}}, 
    \quad 
    \tilde{\beta} =  \beta + (\rho - \sqrt{\delta}\tilde{\rho})\sqrt{\frac{\tau_Z}{\tau_U}} .
\end{align*}
Under the above definition, we can verify that \eqref{eq:non_iden_car_phi_0} must hold. 
Thus, the parameters $(\tilde{\beta}, \tilde{\tau}_Z, \tilde{\varphi}_Z, \tilde{\tau}_U,$ $\tilde{\varphi}_U, \tilde{\tau}_{\epsilon}, \tilde{\rho})$ lead to the same observed data distribution as the true data-generating parameters. 
Moreover, $\tilde{\beta} = \beta$ if and only if 
$\sqrt{\delta}\tilde{\rho} = \rho$, which is further equivalent to
\begin{align*}
    \sqrt{\delta}\zeta \sqrt{\frac{\delta - (1-\rho^2)}{\delta}} = \rho 
    \Longleftrightarrow 
    \zeta \sqrt{\delta - (1-\rho^2)} = \rho. 
\end{align*}
By appropriately choosing $\delta$ and $\zeta$, 
the above equation will fail, under which 
$\tilde{\beta} \ne \beta$.

\subsection{A remark on a fully connected neighborhood structure}

\begin{proposition}\label{prop:car-fullyconnected-nonid}
Consider the same model setup as in Theorem \ref{thm:car-identifiability}. 
If 
$\rho \ne 0$, and
$W$ corresponds to a fully connected neighborhood structure, i.e., the diagonal elements of $W$ are all $0$ and the off-diagonal elements of $W$ are all $1$, 
then the treatment effect $\beta$ a
is 
not identifiable.
\end{proposition}

\begin{proof}[Proof of Proposition \ref{prop:car-fullyconnected-nonid}]
Suppose that $(\beta, \varphi_U, \tau_U, \varphi_Z, \tau_Z, \rho, \tau_\epsilon)$ denote the true data generating parameter, where we again use $\tau_{\epsilon}$ to denote $\sigma_{\epsilon}^{-2}$ for convenience. 
When $\varphi_U = 0$, Proposition \ref{prop:car-fullyconnected-nonid} follows immediately from Section \ref{sec:nece_phiU_car}. 
Below we consider only the case where $\varphi_U \ne 0$. 
From the proof of Theorem \ref{thm:car-identifiability}, it suffices to find other parameters $(\tilde\beta, \tilde\varphi_U, \tilde\tau_U, \tilde\varphi_Z, \tilde\tau_Z, \tilde\rho, \tilde\tau_\epsilon)$ such that equations \eqref{eq:iden_cond_mean_Y_Z}--\eqref{eq:prec_Z} hold and $\tilde\beta \ne \beta$. 

First, we simplify \eqref{eq:iden_cond_mean_Y_Z}--\eqref{eq:prec_Z} using the condition that $W$ represents a connected neighborhood structure. 
By definition, we have $W = 1_{n} 1_{n}^\top - I_n$ and $D = (n-1)I_n$. 
Because the eigenvalues of $1_{n} 1_{n}^\top$ are either $n$ or $0$, 
the eigenvalues of $W$ must be either $n-1$ or $-1$. 
This further implies that the eigenvalue of $D^{-1/2}WD^{-1/2} = (n-1)^{-1} W$ are either $1$ or $-(n-1)^{-1}$. 
We can then verify that equations \eqref{eq:iden_cond_mean_Y_Z}--\eqref{eq:prec_Z}  are equivalent to 

\begin{align}
\beta + \rho\sqrt{\tau_Z/\tau_U}(1-\varphi_U)^{-1} &= \tilde\beta + \tilde\rho\sqrt{\tilde\tau_Z/\tilde\tau_U}(1-\tilde\varphi_U)^{-1}\label{eq:fullyconnected-exp-product-1} \\
\beta + \rho\sqrt{\tau_Z/\tau_U}\left(1+\frac{\varphi_U}{n-1}\right)^{-1} &= \tilde\beta + \tilde\rho\sqrt{\tilde\tau_Z/\tilde\tau_U}\left(1+\frac{\tilde\varphi_U}{n-1}\right)^{-1}\label{eq:fullyconnected-exp-product-n-1} 
\end{align}
\begin{align}
\tau_U^{-1}(1-\varphi_U)^{-1} + \tau_\epsilon^{-1} (n-1) &= \tilde\tau_U^{-1}(1-\tilde\varphi_U)^{-1} + \tilde\tau_\epsilon^{-1}(n-1)\label{eq:fullyconnected-condvar-1} \\
\tau_U^{-1}\left(1+\frac{\varphi_U}{n-1}\right)^{-1}  + \tau_\epsilon^{-1} (n-1) &= \tilde\tau_U^{-1}\left(1+\frac{\tilde\varphi_U}{n-1}\right)^{-1}  + \tilde\tau_\epsilon^{-1}(n-1)\label{eq:fullyconnected-condvar-n-1} 
\end{align}
\begin{align}
\tau_Z\left\{1 - \varphi_Z - \rho^2(1-\varphi_U)^{-1}\right\} &= \tilde\tau_Z\left\{1 - \tilde\varphi_Z - \tilde\rho^2(1-\tilde\varphi_U)^{-1}\right\}\label{eq:fullyconnected-precision-1} \\
\tau_Z\left\{1 + \frac{\varphi_Z}{n-1} - \rho^2\left(1+\frac{\varphi_U}{n-1}\right)^{-1}\right\} &= \tilde\tau_Z\left\{1 + \frac{\tilde\varphi_Z}{n-1} - \tilde\rho^2\left(1+\frac{\tilde\varphi_U}{n-1}\right)^{-1}\right\}\label{eq:fullyconnected-precision-n-1} 
\end{align}

Second, we try to construct parameters $(\tilde\beta, \tilde\varphi_U, \tilde\tau_U, \tilde\varphi_Z, \tilde\tau_Z, \tilde\rho, \tilde\tau_\epsilon)$ such that equations \eqref{eq:fullyconnected-exp-product-1}--\eqref{eq:fullyconnected-precision-n-1} hold. 
Let $b$ be a scalar whose value will be determined later. Define 
\begin{align}
r(b) &= \frac{\left(1 + \frac{\varphi_U}{n-1} \right)^{-1} - \left(1-\varphi_U\right)^{-1}}{\left(1 + \frac{b\varphi_U}{n-1} \right)^{-1} - \left(1-b\varphi_U\right)^{-1}} 
,
\label{eq:car-full-r}
\\
\nonumber
B_1(b) & = \tau_Z + \frac{\tau_Z\varphi_Z}{n-1} - \rho^2\tau_Z\left( 1 + \frac{\varphi_U}{n-1}\right)^{-1} + r\rho^2 \tau_Z \left( 1 + \frac{b\varphi_U}{n-1}\right)^{-1},
\\
\nonumber
B_2(b) & = \tau_Z - \tau_Z \varphi_Z - \rho^2\tau_Z(1-\varphi_U)^{-1} + r\rho^2\tau_Z(1-b\varphi_U)^{-1}.
\end{align}
For notational convenience, we will simply write $r(b), B_1(b)$ and $B_2(b)$ as  $r, B_1$ and $B_2$. 
Define further $(\tilde\beta, \tilde\varphi_U, \tilde\tau_U, \tilde\varphi_Z, \tilde\tau_Z, \tilde\rho, \tilde\tau_\epsilon)$ as follows:
\begin{align}\label{eq:tilde_car_non}
    \tilde\varphi_U & = b\varphi_U, 
    \qquad 
    \tilde\tau_U^{-1} = r \tau_U^{-1}
    \qquad 
    \tilde\tau_\epsilon^{-1} = (n-1)^{-1}\tau_U^{-1}\{(1-\varphi_U)^{-1} - r(1-b\varphi_U)^{-1}\} + \tau_\epsilon^{-1}, 
    \nonumber
    \\
    \tilde\beta & = \beta + \rho\sqrt{\frac{\tau_Z}{\tau_U}}(1-\varphi_U)^{-1} - r\rho\sqrt{\frac{\tau_Z}{\tau_U}}(1-b\varphi_U)^{-1}, 
    \qquad 
    \tilde\varphi_Z = \frac{B_1 - B_2}{(n-1)^{-1}B_2 + B_1},
    \nonumber
    \\
    \tilde\tau_Z &= 
    \frac{B_2}{1-\tilde\varphi_Z}
    = 
    \frac{(n-1)B_1 + B_2}{n}, \qquad \tilde\rho = \rho\sqrt{r \tau_Z/\tilde\tau_Z}. 
\end{align}

Third, we show that equations \eqref{eq:fullyconnected-exp-product-1}--\eqref{eq:fullyconnected-precision-n-1} hold for $(\tilde\beta, \tilde\varphi_U, \tilde\tau_U, \tilde\varphi_Z, \tilde\tau_Z, \tilde\rho, \tilde\tau_\epsilon)$ defined in \eqref{eq:tilde_car_non}, under the assumption that the quantities in \eqref{eq:tilde_car_non} are well-defined. 
We begin with \eqref{eq:fullyconnected-condvar-1} and \eqref{eq:fullyconnected-condvar-n-1}. 
By the definition in \eqref{eq:car-full-r} and \eqref{eq:tilde_car_non}, we have 
\begin{align*}
\tilde\tau_U^{-1}(1-\tilde\varphi_U)^{-1} + \tilde\tau_\epsilon^{-1}(n-1) &= r \tau_U^{-1}(1-b\varphi_U)^{-1} + \tau_U^{-1}\{(1-\varphi_U)^{-1} - r(1-b\varphi_U)^{-1}\} + \tau_\epsilon^{-1}(n-1)
\\
&= \tau_U^{-1}(1-\varphi_U)^{-1} + \tau_\epsilon^{-1}(n-1)
\end{align*} and 
\begin{align*}
& \quad \ \tilde\tau_U^{-1}\left(1+\frac{\tilde\varphi_U}{n-1}\right)^{-1} + \tilde\tau_\epsilon^{-1}(n-1)
\\ 
&= r \tau_U^{-1}\left(1+\frac{b\varphi_U}{n-1}\right)^{-1} + \tau_U^{-1}\{(1-\varphi_U)^{-1} - r(1-b\varphi_U)^{-1}\} + \tau_\epsilon^{-1}(n-1) \\
&= r\tau_U^{-1}\left\{\left(1+\frac{b\varphi_U}{n-1}\right)^{-1} - (1-b\varphi_U)^{-1}\right\} + \tau_U^{-1}(1-\varphi_U)^{-1} + \tau_\epsilon^{-1}(n-1) \\
&= \tau_U^{-1}\left\{\left(1 + \frac{\varphi_U}{n-1} \right)^{-1} - \left(1-\varphi_U\right)^{-1}\right\}+ \tau_U^{-1}(1-\varphi_U)^{-1} + \tau_\epsilon^{-1}(n-1) \\
&= \tau_U^{-1}\left(1+\frac{\varphi_U}{n-1}\right)^{-1} + \tau_\epsilon^{-1} (n-1),
\end{align*} 
which imply that \eqref{eq:fullyconnected-condvar-1} and \eqref{eq:fullyconnected-condvar-n-1} hold. We next consider \eqref{eq:fullyconnected-exp-product-1} and \eqref{eq:fullyconnected-exp-product-n-1}. 
By definition, we have  \begin{equation}\nonumber%
\tilde\rho\sqrt{\tilde\tau_Z/\tilde\tau_U}=
 \rho\sqrt{r \tau_Z/\tilde\tau_Z} \sqrt{\tilde\tau_Z/\tilde\tau_U}
=
\rho\sqrt{r \tau_Z/\tilde\tau_U}
= 
r\rho\sqrt{\tau_Z/\tau_U}.
\end{equation} 
Consequently, 
\begin{align*}
&\tilde\beta + \tilde\rho\sqrt{\tilde\tau_Z/\tilde\tau_U}(1-\tilde\varphi_U)^{-1}\\ &=  \beta + \rho\sqrt{\frac{\tau_Z}{\tau_U}}(1-\varphi_U)^{-1} - r\rho\sqrt{\frac{\tau_Z}{\tau_U}}(1-b\varphi_U)^{-1} + r\rho\sqrt{\tau_Z/\tau_U}(1-b\varphi_U)^{-1} \\
&= \beta + \rho\sqrt{\tau_Z/\tau_U}(1-\varphi_U)^{-1},
\end{align*} 
and 
\begin{align*}
& \quad \ \tilde\beta + \tilde\rho\sqrt{\tilde\tau_Z/\tilde\tau_U}\left(1+\frac{\tilde\varphi_U}{n-1}\right)^{-1} 
\\
&= \beta + \rho\sqrt{\frac{\tau_Z}{\tau_U}}(1-\varphi_U)^{-1} - r\rho\sqrt{\frac{\tau_Z}{\tau_U}}(1-b\varphi_U)^{-1} + r\rho\sqrt{\frac{\tau_Z}{\tau_U}}\left(1+ \frac{b\varphi_U}{n-1} \right)^{-1} \\
&= \beta + \rho\sqrt{\frac{\tau_Z}{\tau_U}}(1-\varphi_U)^{-1} + r\rho \sqrt{\frac{\tau_Z}{\tau_U}}\left\{\left(1+ \frac{b\varphi_U}{n-1} \right)^{-1}-(1-b\varphi_U)^{-1} \right\} \\
&= \beta + \rho\sqrt{\frac{\tau_Z}{\tau_U}}(1-\varphi_U)^{-1} + \rho\sqrt{\frac{\tau_Z}{\tau_U}}\left\{\left(1+ \frac{\varphi_U}{n-1} \right)^{-1}-(1-\varphi_U)^{-1}\right\}\\
&= \beta +\rho\sqrt{\frac{\tau_Z}{\tau_U}}\left(1+ \frac{\varphi_U}{n-1} \right)^{-1},
\end{align*} 
which imply that \eqref{eq:fullyconnected-exp-product-1} and \eqref{eq:fullyconnected-exp-product-n-1} hold. 
We finally consider \eqref{eq:fullyconnected-precision-1}
and \eqref{eq:fullyconnected-precision-n-1}. 
By definition, 
\begin{align*}
\tilde\tau_Z\{1-\tilde\varphi_Z - \tilde\rho^2(1-\tilde\tau_U)^{-1}\} &= \frac{B_2}{1-\tilde\varphi_Z}(1-\tilde\varphi_Z) - \tilde\tau_Z\tilde\rho^2(1-b\varphi_U)^{-1}
\\ &= B_2 - r \tau_Z \rho^2 (1-b\varphi_U)^{-1} \\
&= \tau_Z - \tau_Z \varphi_Z - \rho^2\tau_Z(1-\varphi_U)^{-1} 
\\ &= \tau_Z \{ 1 - \varphi_Z - \rho^2(1-\varphi_U)^{-1} \}, 
\end{align*} 
and 
\begin{align*}
& \quad \ \tilde\tau_Z\left\{1+ \frac{\tilde\varphi_Z}{n-1} - \tilde\rho^2 \left( 1+ \frac{\tilde\varphi_U}{n-1}\right)^{-1}\right\} 
\\
&= \tilde\tau_Z\left(1+\frac{\tilde\varphi_Z}{n-1} \right) - \tilde\tau_Z\tilde\rho^2\left(1+\frac{b\varphi_U}{n-1}\right)^{-1} 
\\ &= 
\frac{B_2}{1 -\tilde\varphi_Z }\left(1+\frac{\tilde\varphi_Z}{n-1} \right) - \tilde\tau_Z\tilde\rho^2\left(1+\frac{b\varphi_U}{n-1}\right)^{-1}
\\
&= 
B_2 \cdot \frac{(n-1)^{-1}B_2+B_1}{\{(n-1)^{-1} + 1\}B_2}\left\{ 1+\frac{B_1-B_2}{B_2 + (n-1)B_1}\right\} - r\rho^2\tau_Z\left(1+\frac{b\varphi_U}{n-1}\right)^{-1}
\\
&= 
B_2 \cdot \frac{B_2+(n-1)B_1}{nB_2}
\frac{n B_1}{B_2 + (n-1)B_1} - r\rho^2\tau_Z \left(1+\frac{b\varphi_U}{n-1}\right)^{-1}
\\ &= B_1 - r\rho^2\tau_Z \left(1+\frac{b\varphi_U}{n-1}\right)^{-1}
\\
& = \tau_Z\left\{1+\frac{\varphi_Z}{n-1} - \rho^2 \left(1+ \frac{\varphi_U}{n-1}\right)^{-1} \right\}, 
\end{align*} 
which imply that \eqref{eq:fullyconnected-precision-1}
and \eqref{eq:fullyconnected-precision-n-1} hold.

Finally, we show that there exists $b$ such that the quantities in \eqref{eq:tilde_car_non} are well-defined and $\tilde\beta \ne \beta$. 
Note that $(\tilde\beta, \tilde\varphi_U, \tilde\tau_U, \tilde\varphi_Z, \tilde\tau_Z, \tilde\rho, \tilde\tau_\epsilon)$ are the same as $(\beta, \varphi_U, \tau_U, \varphi_Z, \tau_Z, \rho, \tau_\epsilon)$ when $b=1$. 
Thus, the parameters $(\tilde\beta, \tilde\varphi_U, \tilde\tau_U, \tilde\varphi_Z, \tilde\tau_Z, \tilde\rho, \tilde\tau_\epsilon)$ are well-defined, continuous in $b$ when viewed as functions of $b$, and lead to a valid model, as long as $b$ is in a sufficiently small neighborhood of $1$. 
Note that 
\begin{align*}
    \tilde\beta - \beta & = 
    \rho\sqrt{\frac{\tau_Z}{\tau_U}}(1-\varphi_U)^{-1} - r\rho\sqrt{\frac{\tau_Z}{\tau_U}}(1-b\varphi_U)^{-1}
    \\ &= 
    \rho\sqrt{\frac{\tau_Z}{\tau_U}} 
    \left\{ (1-\varphi_U)^{-1}  - r (1-b\varphi_U)^{-1}\right\}\\
    & = 
    \rho\sqrt{\frac{\tau_Z}{\tau_U}} 
    \left\{ (1-\varphi_U)^{-1}  - \frac{\left(1 + \frac{\varphi_U}{n-1} \right)^{-1} - \left(1-\varphi_U\right)^{-1}}{\left(1 + \frac{b\varphi_U}{n-1} \right)^{-1} - \left(1-b\varphi_U\right)^{-1}}  (1-b\varphi_U)^{-1}\right\}\\
    & = \rho\sqrt{\frac{\tau_Z}{\tau_U}} \cdot
    \frac{
    (1-\varphi_U)^{-1}\left(1 + \frac{b\varphi_U}{n-1} \right)^{-1} - \left(1 + \frac{\varphi_U}{n-1} \right)^{-1} (1-b\varphi_U)^{-1}
    }{\left(1 + \frac{b\varphi_U}{n-1} \right)^{-1} - \left(1-b\varphi_U\right)^{-1}}\\
    & = 
    \rho\sqrt{\frac{\tau_Z}{\tau_U}} \cdot
    \frac{
    (1-\varphi_U)^{-1} \cdot
    \frac{1-b\varphi_U}{ 1 + \frac{b\varphi_U}{n-1} }
     - \left(1 + \frac{\varphi_U}{n-1} \right)^{-1} 
    }{\frac{1-b\varphi_U}{ 1 + \frac{b\varphi_U}{n-1} } - 1}.
\end{align*}
Because $\varphi_U \ne 0$ and $\rho \ne 0$, we can verify that $ \tilde\beta - \beta$ is a strictly monotone function of $\frac{1-b\varphi_U}{ 1 + \frac{b\varphi_U}{n-1} }$, which itself is also a strictly monotone function of $b$. 
Thus, $\tilde\beta - \beta$ is strictly monotone in $b$ for $b$ in a sufficiently small neighborhood around $1$. 
Thus, $\tilde\beta - \beta$ must be nonzero for $b\ne 1$.

From the above, we derive Proposition \ref{prop:car-fullyconnected-nonid}. 
\end{proof}

\section{Technical details for the non-parsimonious model in \S \ref{sec:leroux_car}}\label{sec:proofs-lcar-nonpars}

\subsection{Proof of Theorem \ref{thm:lcar-flex-identifiability}}

To prove Theorem \ref{thm:lcar-flex-identifiability}, we need the following three lemmas. 

\begin{lemma}\label{lem:lcar-varz}
Consider any $\sigma^2 >0, \lambda < 1$ and $(\omega_1,\ldots,\omega_n) \in \mathbb{R}^n$ with $\omega_i \ge 0$ for all $1 \le i \le n$, where $n\ge 2$ is an integer and $(\omega_1,\ldots,\omega_n)$ contains at least two distinct elements. If $\tilde{\sigma}^2 >0$ and $\tilde{\lambda} < 1$ satisfy
\begin{equation}\label{eq:lcar-var}
\frac{\sigma^2}{1-\lambda + \lambda \omega_i} = \frac{\tilde{\sigma}^2}{1-\tilde{\lambda} + \tilde{\lambda}\omega_i}, \qquad i=1,\ldots, n
\end{equation} 
then $(\tilde{\sigma}^2, \tilde{\lambda}) = (\sigma^2, \lambda)$.  
\end{lemma}
\begin{proof}[Proof of Lemma \ref{lem:lcar-varz}]
Suppose that $\tilde{\sigma}^2 >0$ and $\tilde{\lambda} < 1$ satisfy \eqref{eq:lcar-var}. 
From the condition in Lemma \ref{lem:lcar-varz}, there exist $1 \le j \neq k \le n$ such that $\omega_j \neq \omega_k$. From \eqref{eq:lcar-var}, we have 
\begin{align*}
& \frac{1-\lambda +\lambda \omega_j}{\sigma^2} - \frac{1-\lambda + \lambda \omega_k}{\sigma^2} = \frac{1-\tilde{\lambda} + \tilde{\lambda}\omega_j}{\tilde{\sigma}^2} - \frac{1-\tilde\lambda + \tilde\lambda \omega_k}{\tilde{\sigma}^2}\\
\implies 
& 
\frac{\lambda}{\sigma^2}(\omega_j-\omega_k) 
=
\frac{\tilde\lambda}{\tilde{\sigma}^2}(\omega_j-\omega_k). \end{align*}
Since $\omega_j-\omega_k \neq 0$, this implies that $\lambda/\sigma^2 = \tilde\lambda / \tilde{\sigma}^2$.
For any $1\le i \le n$, from \eqref{eq:lcar-var}, we further have 
\begin{align*}
    \frac{1-\lambda +\lambda \omega_i}{\sigma^2}  = \frac{1-\tilde{\lambda} + \tilde{\lambda}\omega_i}{\tilde{\sigma}^2} 
    \implies
    \frac{1}{\sigma^2} = \frac{1}{\tilde{\sigma}^2}
    \implies 
    \sigma^2 = \tilde{\sigma}^2. 
\end{align*}
Consequently, we must have $\lambda = \tilde\lambda$. 
From the above, we derive Lemma \ref{lem:lcar-varz}. 
\end{proof}

\begin{lemma}\label{lem:lcar-flex-expectation}
Consider any $\sigma^2_Z, \sigma^2_U > 0$, $\lambda_Z, \lambda_{UZ} < 1, \rho \in [-1, 1], \beta \in \mathbb{R}$ and $(\omega_1,\ldots,\omega_n) \in \mathbb{R}^n$ with $\omega_i \ge 0$ for all $1 \le i \le n$, where $n \ge 3$ is an integer and $(\omega_1,\ldots,\omega_n)$ contains at least three distinct elements. Suppose that $\tilde{\sigma}_U^2 > 0, \tilde{\lambda}_{UZ} < 1, \tilde\rho \in [-1, 1]$ and $\tilde{\beta}\in\mathbb{R}$ satisfy
\begin{equation}\label{eq:lcar-flex-expectation}
\beta + \frac{\rho\sigma_U}{\sigma_Z} \cdot \frac{1-\lambda_Z + \lambda_Z \omega_i}{1- \lambda_{UZ} + \lambda_{UZ} \omega_i} = \tilde\beta + \frac{\tilde\rho \tilde\sigma_U}{\sigma_Z} \cdot \frac{1-\lambda_Z + \lambda_Z \omega_i}{1- \tilde\lambda_{UZ} + \tilde\lambda_{UZ} \omega_i}.
\end{equation}
(i) If $\lambda_Z \neq \lambda_{UZ}$ and $\rho \neq 0$, then  $(\tilde\beta, \tilde\lambda_{UZ}, \tilde\rho \tilde\sigma_U) = (\beta, \lambda_{UZ}, \rho\sigma_U)$. \\
(ii) If $\lambda_Z = \lambda_{UZ}$ or $\rho=0$, then  $\tilde\lambda_{UZ} = \lambda_Z$ or $\tilde{\rho}=0$. 
\end{lemma}

\begin{proof}[Proof of Lemma \ref{lem:lcar-flex-expectation}(i)]
Consider the case where $\lambda_Z \neq \lambda_{UZ}$ and $\rho \neq 0$.
From \eqref{eq:lcar-flex-expectation}, for any $1\le i, j \le n$, 
\[\frac{\rho \sigma_U}{\sigma_Z}\left( \frac{1-\lambda_Z + \lambda_Z \omega_i}{1- \lambda_{UZ} + \lambda_{UZ} \omega_i} -  \frac{1-\lambda_Z + \lambda_Z \omega_j}{1- \lambda_{UZ} + \lambda_{UZ} \omega_j} \right) = \frac{\tilde\rho \tilde\sigma_U}{\sigma_Z}\left( \frac{1-\lambda_Z + \lambda_Z \omega_i}{1- \tilde\lambda_{UZ} + \tilde\lambda_{UZ} \omega_i} -  \frac{1-\lambda_Z + \lambda_Z \omega_j}{1- \tilde\lambda_{UZ} + \tilde\lambda_{UZ} \omega_j} \right).\]
Since $\rho \neq 0$, we can define $c_1 = \tilde{\rho}\tilde{\sigma}_U / (\rho \sigma_U)$. Thus, from the above, for any $1 \le i,j \le n$, 
\begin{align*}&\left( \frac{1-\lambda_Z + \lambda_Z \omega_i}{1- \lambda_{UZ} + \lambda_{UZ} \omega_i} -  \frac{1-\lambda_Z + \lambda_Z \omega_j}{1- \lambda_{UZ} + \lambda_{UZ} \omega_j} \right) = c_1\left( \frac{1-\lambda_Z + \lambda_Z \omega_i}{1- \tilde\lambda_{UZ} + \tilde\lambda_{UZ} \omega_i} -  \frac{1-\lambda_Z + \lambda_Z \omega_j}{1- \tilde\lambda_{UZ} + \tilde\lambda_{UZ} \omega_j} \right)\\
\iff & \frac{1-\lambda_Z + \lambda_Z \omega_i}{1- \lambda_{UZ} + \lambda_{UZ} \omega_i} - c_1  \frac{1-\lambda_Z + \lambda_Z \omega_i}{1- \tilde\lambda_{UZ} + \tilde\lambda_{UZ} \omega_i} = \frac{1-\lambda_Z + \lambda_Z \omega_j}{1- \lambda_{UZ} + \lambda_{UZ} \omega_j}  -  c_1\frac{1-\lambda_Z + \lambda_Z \omega_j}{1- \tilde\lambda_{UZ} + \tilde\lambda_{UZ} \omega_j} .
\end{align*}
Consequently, there exists a constant $c_2$ such that for any $\omega \in \{\omega_1, \omega_2,\ldots, \omega_n\}$,
\begin{align}
&\frac{1-\lambda_Z + \lambda_Z \omega}{1- \lambda_{UZ} + \lambda_{UZ} \omega} - c_1  \frac{1-\lambda_Z + \lambda_Z \omega}{1- \tilde\lambda_{UZ} + \tilde\lambda_{UZ} \omega} = c_2 \label{eq:lcar-exp-const1} \\
\implies &(1-\lambda_Z+\lambda_Z \omega)(1-\tilde{\lambda}_{UZ} + \tilde{\lambda}_{UZ}\omega) - c_1(1-\lambda_Z+\lambda_Z \omega)(1-{\lambda}_{UZ} + {\lambda}_{UZ}\omega) \nonumber\\
& -c_2(1-\lambda_{UZ}+\lambda_{UZ} \omega)(1-\tilde{\lambda}_{UZ} + \tilde{\lambda}_{UZ}\omega) =0. \label{eq:lcar-exp-const2}
\end{align}
From the condition in Lemma \ref{lem:lcar-flex-expectation}, the left-hand side (LHS) of \eqref{eq:lcar-exp-const2} has at least three distinct roots. Since it is quadratic in $\omega$, it must be zero for all $\omega \in \mathbb{R}$. Thus, the LHS  of \eqref{eq:lcar-exp-const2} must be zero for $\omega=1$, which implies that 
\begin{equation}\label{eq:lcar-flex-root1}
1-c_1-c_2 =0 \implies 1-c_1=c_2.
\end{equation} Similarly, the LHS of \eqref{eq:lcar-exp-const2} must be zero at $\omega=0$, which implies that 
\begin{equation}\label{eq:lcar-flex-root2}
(1-\lambda_Z)\{(1-\tilde{\lambda}_{UZ}) - c_1(1-\lambda_{UZ})\} - c_2(1-\lambda_{UZ})(1-\tilde{\lambda}_{UZ}) =0.
\end{equation} 
Additionally, 
the coefficient of $\omega^2$ in the LHS of \eqref{eq:lcar-exp-const2} must also be zero, which implies that 
\begin{equation}\label{eq:lcar-flex-root3}
\lambda_Z\tilde{\lambda}_{UZ} - c_1 \lambda_Z \lambda_{UZ} - c_2\lambda_{UZ} \tilde\lambda_{UZ} =0 \implies \lambda_Z(\tilde{\lambda}_{UZ} - c_1\lambda_{UZ}) = c_2\lambda_{UZ} \tilde\lambda_{UZ}.
\end{equation}
Applying these together, 
we then have 
\begin{align}
0 &= \lambda_Z(1-\lambda_Z)\{(1-\tilde{\lambda}_{UZ}) - c_1(1-\lambda_{UZ})\} - c_2\lambda_Z(1-\lambda_{UZ})(1-\tilde{\lambda}_{UZ})\nonumber\\
&= (1-\lambda_Z)\{\lambda_Z(1-c_1) -\lambda_Z(\tilde\lambda_{UZ} - c_1\lambda_{UZ})\} - c_2\lambda_Z(1-\lambda_{UZ})(1-\tilde{\lambda}_{UZ})\nonumber\\
&=(1-\lambda_Z)(\lambda_Zc_2 - c_2\lambda_{UZ}
\tilde{\lambda}_{UZ}) - c_2\lambda_Z(1-\lambda_{UZ})(1-\tilde\lambda_{UZ}) \nonumber\\
&= c_2 \{(1-\lambda_Z)\lambda_Z - (1-\lambda_Z)\lambda_{UZ}\tilde\lambda_{UZ} - \lambda_Z(1-\lambda_{UZ}) + \lambda_Z(1-\lambda_{UZ})\tilde\lambda_{UZ}\} \nonumber\\
&= c_2\{\lambda_Z(\lambda_{UZ}-\lambda_Z) - (\lambda_{UZ} - \lambda_Z)\tilde{\lambda}_{UZ}\} \nonumber\\
&= c_2(\lambda_{UZ} - \lambda_Z)(\lambda_Z - \tilde{\lambda}_{UZ}), \label{eq:lcar-flex-root-fin}
\end{align} 
where the first equality follows by multiplying both sides of \eqref{eq:lcar-flex-root2} by $\lambda_Z$, 
the third equality follows from \eqref{eq:lcar-flex-root1} and \eqref{eq:lcar-flex-root3}, and the remaining equalities follow from some algebra. 
From \eqref{eq:lcar-flex-root-fin} and the condition that $\lambda_Z \neq \lambda_{UZ}$, we must have either $c_2=0$ or $\tilde\lambda_{UZ} = \lambda_Z$. 

We then prove that $\tilde\lambda_{UZ}\ne \lambda_Z$ by contradiction. 
Suppose that $\tilde\lambda_{UZ} = \lambda_Z$. From \eqref{eq:lcar-exp-const1}, we can see that for every $\omega \in \{\omega_1,\ldots,\omega_n\}$, 
\[\frac{1-\lambda_Z+\lambda_Z\omega}{1-\lambda_{UZ} + \lambda_{UZ} \omega} = c_1+c_2.\]
Because $\lambda_Z \neq \lambda_{UZ}$, we can verify that $(1-\lambda_Z+\lambda_Z\omega)/(1-\lambda_{UZ} + \lambda_{UZ} \omega)$ is strictly monotone in $\omega$. Because $\{\omega_1,\ldots,\omega_n\}$ contains at least three distinct values, this leads to a contradiction.

From the above, we must have $c_2=0$. From \eqref{eq:lcar-flex-root1}, we then have $c_1=1$. 
By the definition of $c_1$, this implies that $\tilde\rho \tilde\sigma_U = \rho\sigma_U$. 
In addition, from \eqref{eq:lcar-exp-const1}, for any $1\le i \le n$, 
\begin{align*}
&\frac{1-\lambda_Z+\lambda_Z \omega_i}{ 1-\lambda_{UZ} + \lambda_{UZ} \omega_i} = \frac{1-\lambda_Z+\lambda_Z \omega_i}{ 1-\tilde\lambda_{UZ} + \tilde\lambda_{UZ} \omega_i}\\
\implies & 1-\lambda_{UZ} + \lambda_{UZ} \omega_i  =  1-\tilde\lambda_{UZ} + \tilde\lambda_{UZ} \omega_i \\
\implies & (\lambda_{UZ} - \tilde\lambda_{UZ})(\omega_i-1) =0.
\end{align*}
Since $\{\omega_1,\ldots,\omega_n\}$ contains at least three distinct values, we must have $\lambda_{UZ} = \tilde\lambda_{UZ}$. 
From \eqref{eq:lcar-flex-expectation}, we then have $\tilde \beta = \beta$. 
Therefore, we derive Lemma \ref{lem:lcar-flex-expectation}(i).
\end{proof}

\begin{proof}[Proof of Lemma \ref{lem:lcar-flex-expectation}(ii)]
Consider the case where $\rho=0$ or $\lambda_Z = \lambda_{UZ}$. In this case, the LHS of \eqref{eq:lcar-flex-expectation} is equal to either $\beta$ or $\beta + \rho\sigma_U/\sigma_Z$, and is thus constant for all $\omega_i$. 

We now prove $\tilde\rho =0$ or $\tilde\lambda_{UZ} = \lambda_Z$ by contradiction. 
Suppose that $\tilde\lambda_{UZ} \neq \lambda_Z$ and $\tilde\rho \neq 0$. We can verify that the right-hand side (RHS) of \eqref{eq:lcar-flex-expectation} is strictly monotone in $\omega_i$. 
Because $\{\omega_1,\ldots,\omega_n\}$ contains at least three distinct values, this leads to a contradiction. 

Therefore, we must have either $\tilde\rho =0$ or $\tilde\lambda_{UZ} = \lambda_Z$, i.e., Lemma \ref{lem:lcar-flex-expectation}(ii) holds. 
\end{proof}

\begin{lemma}\label{lem:lcar-expression-linear-independence}
Consider any integer $1 \le K \le 3$, $a_1, \ldots, a_K, b \in \mathbb{R}$ and $\lambda_1,\ldots,\lambda_K \in  (0,1)$. If $\lambda_1,\ldots,\lambda_K$ are all distinct, and 
\begin{equation}\label{eq:lem-lcar-expression-linind}
\sum_{k=1}^K \frac{a_k}{1-\lambda_k+\lambda_k \omega} + b =0
\end{equation} 
for at least $K+1$ distinct nonnegative values of $\omega$, then $a_1=\cdots = a_K=b=0$.
\end{lemma}
\begin{proof}[Proof of Lemma \ref{lem:lcar-expression-linear-independence}]
First, we consider the case where $K=1$. From \eqref{eq:lem-lcar-expression-linind}, for at least two distinct nonnegative values of $\omega$,
\[\frac{a_1}{1-\lambda_1 + \lambda_1 \omega} + b =0 \implies a_1 + b(1-\lambda_1 + \lambda_1\omega)=0. \]
Since $a_1 + b(1-\lambda_1 + \lambda_1\omega)$ is linear in $\omega$, this implies that $b\lambda_1=0$. Since $\lambda_1\neq 0$, this further implies that $b=0$. Consequently, $a_1=0$.
Therefore, Lemma \ref{lem:lcar-expression-linear-independence} holds when $K=1$.
\\

Second, we consider the case where $K=2$. For at least 3 distinct nonnegative values of $\omega$, we have 
\begin{equation}\label{eq:lcar-expression-k2}
a_1(1-\lambda_2+\lambda_2\omega)+a_2(1-\lambda_1+\lambda_1\omega) + b(1-\lambda_1+\lambda_1\omega)(1-\lambda_2+\lambda_2 \omega)=0.
\end{equation} 
Because the expression on the LHS of \eqref{eq:lcar-expression-k2} is quadratic in $\omega$, it must be identically zero for all $\omega \in \mathbb{R}$. 
By letting $\omega$ equal $0$ and $1$, we have 
\begin{align}
    \label{eq:lcar-expression-root0}
a_1(1-\lambda_2)+a_2(1-\lambda_1)+b(1-\lambda_1)(1-\lambda_2) & =0,
\\
\label{eq:lcar-expression-root1}
a_1+a_2+b & = 0.
\end{align}
Because the coefficient of $\omega^2$ in the LHS of \eqref{eq:lcar-expression-k2} must be zero, we also have 
\begin{equation}\label{eq:lcar-expression-root-quadratic-coef}
b \lambda_1\lambda_2 =0.
\end{equation}
Note that $\lambda_1 \neq \lambda_2$ and both of them are nonzero. 
From \eqref{eq:lcar-expression-root-quadratic-coef}, $b$ must be zero. Multiplying \eqref{eq:lcar-expression-root1} by $\lambda_2-1$ and adding it to \eqref{eq:lcar-expression-root0}, we obtain
\[a_2(1-\lambda_1) -a_2(1-\lambda_2) =0 \implies a_2(\lambda_2-\lambda_1)=0. \] This implies $a_2=0$. Consequently $a_1=0$. Therefore, Lemma \ref{lem:lcar-expression-linear-independence} holds when $K=2$. \\

Finally, we consider the case where $K=3$. In this case, there are at least 4 distinct nonnegative values of $\omega$ such that: 
\begin{align}
a_1(1-\lambda_2+\lambda_2\omega)(1-\lambda_3+\lambda_3\omega) + a_2(1-\lambda_1+\lambda_1\omega)(1-\lambda_3+\lambda_3\omega) & \nonumber\\  + a_3 (1-\lambda_1+\lambda_1 \omega)(1-\lambda_2+\lambda_2\omega)+ b \prod_{k=1}^3 (1-\lambda_k +\lambda_k \omega) & =0 \label{eq:lcar-expression-k3}. 
\end{align} 
Because the LHS of \eqref{eq:lcar-expression-k3} is cubic in $\omega$ and has at most three real roots, 
it must be identically zero for all $\omega\in \mathbb{R}$. 
Thus, the coefficient of $\omega^3$ must be zero, implying that 
$b\lambda_1\lambda_2\lambda_3 =0.$
Because $\lambda_1, \lambda_2, \lambda_3$ are all nonzero, we then have $b=0$. 
From \eqref{eq:lcar-expression-k3}, 
we then have, for at least 4 distinct nonnegative values of $\omega$, 
\begin{align*}
    a_1(1-\lambda_2+\lambda_2\omega)(1-\lambda_3+\lambda_3\omega) + a_2(1-\lambda_1+\lambda_1\omega)(1-\lambda_3+\lambda_3\omega) &  \\
    + a_3 (1-\lambda_1+\lambda_1 \omega)(1-\lambda_2+\lambda_2\omega) & =0. 
\end{align*}
Because the LHS of the above equation is quadratic in $\omega$ and has at most two real roots, it must be identically zero for all $\omega\in \mathbb{R}$. 
Again, 
because the LHS of the above equation must be zero when $\omega$ equals 0 and 1, and the coefficient of $\omega^2$ is also zero, we can derive that 
\begin{align}
    \label{eq:lcar-k3-root0}
a_1(1-\lambda_2)(1-\lambda_3) + a_2(1-\lambda_1)(1-\lambda_3) + a_3 (1-\lambda_1)(1-\lambda_2) & =0,
\\
\label{eq:lcar-k3-root1}
a_1+a_2+a_3 & = 0,
\\
\label{eq:lcar-k3-coef}
a_1\lambda_2\lambda_3 + a_2\lambda_1\lambda_3+a_3\lambda_1\lambda_2 & =0.
\end{align}
Adding \eqref{eq:lcar-k3-root1}, \eqref{eq:lcar-k3-coef} and subtracting \eqref{eq:lcar-k3-root0}, we obtain:
\begin{equation}\label{eq:lcar-k3-interm}
a_1(\lambda_2+\lambda_3) + a_2(\lambda_1+\lambda_3) + a_3(\lambda_1+\lambda_2)=0.
\end{equation}
Consider \eqref{eq:lcar-k3-root1},  \eqref{eq:lcar-k3-coef},  \eqref{eq:lcar-k3-interm} as equations for variables $a_1, a_2, a_3$, and let $M$ denote the corresponding coefficient matrix. 
By some algebra, the determinant of $M$ has the following equivalent forms:
\begin{align*}
|M|&=\left| \begin{matrix} 1 & 1 & 1 \\
\lambda_2\lambda_3 & \lambda_1\lambda_3 & \lambda_1\lambda_2\\
\lambda_2+\lambda_3 & \lambda_1+\lambda_3 & \lambda_1+\lambda_2 \end{matrix}\right|
= (\lambda_1-\lambda_2)(\lambda_2-\lambda_3)(\lambda_3-\lambda_1),
\end{align*} 
which must be nonzero since $\lambda_1, \lambda_2$ and $\lambda_3$ are all distinct. 
Consequently, we must have $a_1=a_2=a_3=0$. 
Therefore, Lemma \ref{lem:lcar-expression-linear-independence} holds when $K=3$.
\end{proof}

\begin{proof}[Proof of Theorem \ref{thm:lcar-flex-identifiability}]
By definition, 
\begin{align*}
    \Sigma_{UU}' & = \sigma_U^2 \{ (1-\lambda_U) I_n + \lambda_U \Omega \}^{-1}, 
    \qquad \quad 
    \Sigma_{ZZ}' = \sigma_Z^2 \{ (1-\lambda_Z) I_n + \lambda_Z \Omega \}^{-1}, \\
    \Sigma_{UZ}' & = \rho \sigma_U \sigma_Z \{ (1-\lambda_{UZ} )I_n + \lambda_{UZ} \Omega \}^{-1}, 
\end{align*}
where $\Omega$ is a diagonal matrix with diagonal elements $\omega_1, \ldots, \omega_n$, and $\omega_1, \ldots, \omega_n$ are the eigenvalues of $D-W$. 
Because $D-W$ is diagonally dominant, the eigenvalues $\omega_1, \ldots, \omega_n$ are all nonnegative. 
By the same logic as \eqref{eq:condvar}, the observed data distribution is determined by 
\begin{align}\label{eq:lcar_obs}
    \var(Z') & = \Sigma_{ZZ}'
    = \sigma_Z^2 \{ (1-\lambda_Z) I_n + \lambda_Z \Omega \}^{-1}, 
    \nonumber\\
    \cov(Y', Z') \var(Z')^{-1} & = \beta I_n + \Sigma_{UZ}' \Sigma'^{-1}_{ZZ}
    \nonumber
    \\
    & =\beta I_n + \rho \frac{\sigma_U}{ \sigma_Z} \{ (1-\lambda_{UZ} )I_n + \lambda_{UZ} \Omega \}^{-1} \{ (1-\lambda_Z) I_n + \lambda_Z \Omega \},
    \nonumber
    \\
    \var(Y' \mid Z') & = \Sigma_{UU}' - \Sigma_{UZ}' \Sigma'^{-1}_{ZZ} \Sigma'_{ZU} + \sigma_{\epsilon}^2 I_n
    \nonumber
    \\
    & = 
    \sigma_U^2 \{ (1-\lambda_U) I_n + \lambda_U \Omega \}^{-1} + \sigma_{\epsilon}^2 I_n
    \nonumber
    \\
    & \quad \ 
    - 
    \rho^2 \sigma_U^2 \{ (1-\lambda_{UZ} )I_n + \lambda_{UZ} \Omega \}^{-2}
     \{ (1-\lambda_Z) I_n + \lambda_Z \Omega \} 
\end{align}
Let $(\sigma_Z^2, \lambda_Z, \sigma_U^2, \lambda_U, \rho, \lambda_{UZ}, \sigma_\epsilon^2, \beta)$ denote the true data generating parameters, and suppose that $(\tilde\sigma_Z^2, \tilde\lambda_Z, $ $\tilde\sigma_U^2, \tilde\lambda_U, \tilde\rho, \tilde\lambda_{UZ}, \tilde\sigma_\epsilon^2, \tilde\beta)$ lead to the same observed data distribution.

From $\var(Z')$ in \eqref{eq:lcar_obs}, we have, for $1\le i \le n$, 
\begin{equation}\label{eq:lcar-flex-varz}
\frac{\sigma_Z^2}{1-\lambda_Z + \lambda_Z \omega_i} = \frac{\tilde\sigma_Z^2}{1-\tilde\lambda_Z + \tilde\lambda_Z \omega_i}.
\end{equation}
In either case (i) or (ii) in Theorem \ref{thm:lcar-flex-identifiability}, 
$\{ \omega_1, \ldots, \omega_n \}$ contains at least three distinct eigenvalues. 
From Lemma \ref{lem:lcar-varz}, we must have $(\tilde\sigma_Z^2, \tilde{\lambda}_Z) = (\sigma_Z^2, \lambda_Z)$. 
From $\cov(Y', Z') \var(Z')^{-1}$ and $\var(Y'|Z')$ in \eqref{eq:lcar_obs}, we then have, for $1\le i \le n$, 
\begin{equation}\label{eq:lcar-flex-expectation-2}
\beta + \frac{\rho \sigma_U}{\sigma_Z} \cdot \frac{1-\lambda_Z +\lambda_Z\omega_i}{1-\lambda_{UZ} + \lambda_{UZ} \omega_i} = \tilde\beta + \frac{\tilde\rho \tilde\sigma_U}{\sigma_Z} \cdot \frac{1-\lambda_Z +\lambda_Z\omega_i}{1-\tilde\lambda_{UZ} + \tilde\lambda_{UZ} \omega_i},
\end{equation}
and 
\begin{align}\label{eq:lcar-flex-condvar}
& \quad 
\frac{\sigma_U^2}{1-\lambda_U + \lambda_U \omega_i} + \sigma_\epsilon^2 - (\rho\sigma_U)^2 \frac{1-\lambda_Z+\lambda_Z \omega_i}{(1-\lambda_{UZ} +\lambda_{UZ}\omega_i)^2} 
\nonumber
\\
& = \frac{\tilde\sigma_U^2}{1-\tilde\lambda_U + \tilde\lambda_U \omega_i} + \tilde\sigma_\epsilon^2 
- (\tilde\rho\tilde\sigma_U)^2 \frac{1-\lambda_Z+\lambda_Z \omega_i}{(1-\tilde\lambda_{UZ} +\tilde\lambda_{UZ}\omega_i)^2}.
\end{align}

First, we consider the case in (i) where $\lambda_{UZ}\neq\lambda_Z$ and $\rho \neq 0$. By applying Lemma \ref{lem:lcar-flex-expectation} to \eqref{eq:lcar-flex-expectation-2}, we have $(\tilde\beta, \tilde\lambda_{UZ}, \tilde\rho\tilde\sigma_U) = (\beta, \lambda_{UZ}, \rho\sigma_U)$. 
From the discussion before, $(\beta, \lambda_{UZ}, \sigma_Z^2, \lambda_Z)$ are identifiable. 
We now consider the case where we further have $\lambda_U\neq 0$. 
Due to the fact that $( \tilde\lambda_{UZ}, \tilde\rho\tilde\sigma_U) = ( \lambda_{UZ}, \rho\sigma_U)$, \eqref{eq:lcar-flex-condvar} reduces to 
\begin{equation}\label{eq:lcar-flex-condvar-interm}
\frac{\sigma_U^2}{1-\lambda_U+\lambda_U\omega_i} +
\frac{-\tilde\sigma_U^2}{1-\tilde\lambda_U+\tilde\lambda_U\omega_i} 
+ 
(\sigma_\epsilon^2 - \tilde\sigma_\epsilon^2) = 0
\end{equation} 
Note that $\lambda_U\neq 0$ and, by the conditions in Theorem \ref{thm:lcar-flex-identifiability}, $\{\omega_1,\ldots,\omega_n\}$ contains at least 3 distinct nonnegative values. If $\tilde\lambda_U \neq \lambda_U$, from \eqref{eq:lcar-flex-condvar-interm} and applying Lemma \ref{lem:lcar-expression-linear-independence} with $K=1$ or $K=2$ depending on whether $\tilde\lambda_U=0$, we can derive that $\sigma_U^2=0$, leading to a contradiction. Thus, $\tilde\lambda_U = \lambda_U$. Applying Lemma \ref{lem:lcar-expression-linear-independence} with $K=1$, we can then derive that $\tilde\sigma_U^2 = \sigma_U^2$ and $\tilde\sigma_\epsilon^2 = \sigma_\epsilon^2$. 
Because $\tilde\rho\tilde\sigma_U = \rho\sigma_U$, this further implies that $\tilde\rho = \rho$. 
Consequently, all the parameters $(\sigma_Z^2, \lambda_Z, \sigma_U^2, \lambda_U, \rho, \lambda_{UZ}, \sigma_\epsilon^2, \beta)$ can be identified from the observed data. 

We then consider the case in (ii) where $\rho =0$. Note that  $\sigma_Z$ and $\lambda_Z$ have already been identified before. By applying Lemma \ref{lem:lcar-flex-expectation} to \eqref{eq:lcar-flex-expectation-2}, we have $\tilde\rho =0$ or $\tilde\lambda_{UZ} = \lambda_Z$.
Below we prove $\tilde\rho =0$. 
It suffices to consider
the case where $\tilde\lambda_{UZ} =\lambda_Z$. From \eqref{eq:lcar-flex-condvar}, this implies that 
\begin{equation}\label{eq:lcar-flex-condvar-interm2}
\frac{\sigma_U^2}{1-\lambda_U+\lambda_U\omega_i} + \sigma_\epsilon^2 = \frac{\tilde\sigma_U^2}{1-\tilde\lambda_U+\tilde\lambda_U\omega_i} + \tilde\sigma_\epsilon^2 - \frac{(\tilde\rho\tilde\sigma_U)^2}{1-\lambda_Z +\lambda_Z\omega_i}.
\end{equation}
From the conditions in Theorem \ref{thm:lcar-flex-identifiability}(ii), $\lambda_U \ne \lambda_Z$ are both nonzero, and $\{\omega_1,\ldots,\omega_n)$ contains at least 4 distinct nonnegative values. 
If $\tilde\lambda_U \neq \lambda_U$, from \eqref{eq:lcar-flex-condvar-interm2} and applying Lemma \ref{lem:lcar-expression-linear-independence} with $K=2$ or $K=3$ depending on whether $\tilde\lambda_U\in\{\lambda_Z, 0\}$, we can derive that $\sigma_U^2=0$, leading to a contradiction. Thus, $\tilde\lambda_U = \lambda_U$. Applying Lemma \ref{lem:lcar-expression-linear-independence} again with $K=2$,  $\tilde\sigma_U^2 = \sigma_U^2, \tilde\sigma_\epsilon^2 = \sigma_\epsilon^2$, and $(\tilde\rho \tilde\sigma_U)^2=0$, which further implies that $\tilde\rho =0$. Therefore, we must have  $\tilde\rho=0 = \rho$. 
From \eqref{eq:lcar-flex-expectation-2}, we then have that $\tilde\beta = \beta$. From \eqref{eq:lcar-flex-condvar}, we also have
\[\frac{\sigma_U^2}{1-\lambda_U+\lambda_U \omega_i} + \sigma_\epsilon^2 = \frac{\tilde\sigma_U^2}{1-\tilde\lambda_U + \tilde\lambda_U \omega_i} + \tilde\sigma_\epsilon^2. \]
If $\tilde\lambda_U\neq \lambda_U $, from Lemma \ref{lem:lcar-expression-linear-independence} with $K=1$ or $K=2$ depending on whether $\tilde\lambda_U=0$, we can show that $\sigma_U^2=0$, leading to a contradiction. Thus, $\tilde\lambda_U=\lambda_U$.  Applying Lemma \ref{lem:lcar-expression-linear-independence} again with $K=1$, we can then show $\tilde\sigma_U^2 = \sigma_U^2$ and $\tilde\sigma_\epsilon^2 = \sigma_\epsilon^2$.
Consequently, the  parameters $(\sigma_Z^2, \lambda_Z, \sigma_U^2, \lambda_U, \rho, \sigma_\epsilon^2, \beta)$ can be identified from the observed data. 

From the above, Theorem \ref{thm:lcar-flex-identifiability} holds. 
\end{proof}

\subsection{A remark on the necessity of $\lambda_{UZ} \neq \lambda_Z$ when $\rho\neq 0$} 

\begin{proposition}\label{prop:nonid-lambdaUZ-lambdaZ}
Consider $n$ spatial locations with a proximity matrix $W$, and assume the data generating process in \eqref{eq:simp-model}, \eqref{eq:leroux_car} and \eqref{eq:lcar-crosscov-flex}. 
If $\rho\neq 0$ and $\lambda_{UZ}=\lambda_Z$, then $\beta$ is not identifiable from the observed data.
\end{proposition}

\begin{proof}[Proof of Proposition \ref{prop:nonid-lambdaUZ-lambdaZ}]
Let $(\sigma_Z^2, \lambda_Z, \sigma_U^2, \lambda_U, \rho, \lambda_{UZ}, \sigma_\epsilon^2, \beta)$ denote the true data generating parameters. 
Below we construct alternative parameters $(\tilde\sigma_Z^2, \tilde\lambda_Z, \tilde\sigma_U^2, \tilde\lambda_U, \tilde\rho, \tilde\lambda_{UZ}, \tilde\sigma_\epsilon^2, \tilde\beta)$ such that they lead to the same observed data distribution and $\tilde\beta \ne \beta$. 
Specifically, let $(\tilde\sigma_Z^2, \tilde\lambda_Z, \tilde\sigma_U^2, \tilde\lambda_U, \tilde\lambda_{UZ}, \tilde{\sigma}_\epsilon^2) = (\sigma_Z^2, \lambda_Z, \sigma_U^2, \lambda_U, \lambda_{UZ}, \sigma_\epsilon^2)$, $\tilde\rho= -\rho$,  and $\tilde\beta = \beta + 2\rho \sigma_U/\sigma_Z$. 
We can then verify that \eqref{eq:lcar-flex-varz}--\eqref{eq:lcar-flex-condvar} hold. 
In addition, because $\rho \ne 0$, $\tilde \beta$ must be different from $\beta$. 
Therefore, we can then derive Proposition \ref{prop:nonid-lambdaUZ-lambdaZ}. 
\end{proof}

\subsection{A remark on the necessity of $\lambda_U\neq \lambda_Z, \lambda_Z \neq 0$ and $\lambda_U \neq 0$ when $\rho = 0$} 

\begin{proposition}\label{prop:nonid-lcar_rho0}
Consider $n$ spatial locations with a proximity matrix $W$, and assume the data generating process in \eqref{eq:simp-model}, \eqref{eq:leroux_car} and \eqref{eq:lcar-crosscov-flex}. 
If $\rho = 0$, and $(\lambda_U, \lambda_Z, 0)$ are not mutually distinct, then $\beta$ is not identifiable from the observed data.
\end{proposition}

\begin{proof}[Proof of Proposition \ref{prop:nonid-lcar_rho0}]
    Let $(\sigma_Z^2, \lambda_Z, \sigma_U^2, \lambda_U, \rho, \lambda_{UZ}, \sigma_\epsilon^2, \beta)$ denote the true data generating parameters. 
Below we construct alternative parameters $(\tilde\sigma_Z^2, \tilde\lambda_Z, \tilde\sigma_U^2, \tilde\lambda_U, \tilde\rho, \tilde\lambda_{UZ}, \tilde\sigma_\epsilon^2, \tilde\beta)$ such that they lead to the same observed data distribution and $\tilde\beta \ne \beta$. 
Let $(\tilde\sigma_Z^2, \tilde{\lambda}_Z) = (\sigma_Z^2, \lambda_Z)$, and $\tilde\lambda_{UZ} = \lambda_Z$. 
Then \eqref{eq:lcar-flex-varz} holds obviously, 
and \eqref{eq:lcar-flex-expectation-2} and \eqref{eq:lcar-flex-condvar} become equivalent to 
\begin{align}\label{eq:lcar-flex-expectation-non-rho0}
\beta & = \tilde\beta + \frac{\tilde\rho \tilde\sigma_U}{\sigma_Z}, 
\nonumber
\\
\frac{\sigma_U^2}{1-\lambda_U + \lambda_U \omega_i}  
& = \frac{\tilde\sigma_U^2}{1-\tilde\lambda_U + \tilde\lambda_U \omega_i} + \tilde\sigma_\epsilon^2 - \sigma_\epsilon^2
- (\tilde\rho\tilde\sigma_U)^2 \frac{1}{1-\lambda_Z+\lambda_Z \omega_i}.
\end{align}
Let $\tilde\beta = \beta - \tilde\rho \tilde\sigma_U/\sigma_Z$, where the values of $\tilde\rho$ and $\tilde\sigma_U$ will be specified later. 
Because $(\lambda_U, \lambda_Z, 0)$ are not mutually distinct, we must have $\lambda_U = \lambda_Z, \lambda_Z = 0$ or $\lambda_U = 0$. 
Below we consider these three cases, separately. 

First, consider the case where $\lambda_U = \lambda_Z$. In this case, 
\eqref{eq:lcar-flex-expectation-non-rho0} is equivalent to 
\begin{align*}
    \frac{\sigma_U^2 + (\tilde\rho\tilde\sigma_U)^2}{1-\lambda_U + \lambda_U \omega_i}  
& = \frac{\tilde\sigma_U^2}{1-\tilde\lambda_U + \tilde\lambda_U \omega_i} + \tilde\sigma_\epsilon^2 - \sigma_\epsilon^2. 
\end{align*}
Let $\tilde\lambda_U = \lambda_U$ and $\tilde\sigma_\epsilon^2 = \sigma_\epsilon^2$. Then the above equation will hold as long as $(1-\tilde\rho^2) \tilde\sigma_U^2 = \sigma^2_U$. We can then choose $\tilde{\rho}$ and $\tilde\sigma_U^2$ such that $\tilde\beta = \beta - \tilde\rho \tilde\sigma_U/\sigma_Z \ne \beta$.

Second, consider the case where $\lambda_Z=0$. In this case, 
\eqref{eq:lcar-flex-expectation-non-rho0} is equivalent to 
\begin{align*}
    \frac{\sigma_U^2}{1-\lambda_U + \lambda_U \omega_i}  
& = \frac{\tilde\sigma_U^2}{1-\tilde\lambda_U + \tilde\lambda_U \omega_i} + \tilde\sigma_\epsilon^2 - \sigma_\epsilon^2
- (\tilde\rho\tilde\sigma_U)^2 
\end{align*}
Let $\tilde\lambda_U = \lambda_U$ and $\tilde\sigma_U^2 = \sigma_U^2$. 
Then the above equation will hold as long as $\tilde\sigma_\epsilon^2 
- (\tilde\rho \sigma_U)^2 = \sigma_\epsilon^2$. 
We can then choose $\tilde\rho$ and $\tilde\sigma_\epsilon^2$ such that $\tilde\beta = \beta - \tilde\rho \tilde\sigma_U/\sigma_Z \ne \beta$. 

Third, consider the case where $\lambda_U = 0$.  In this case, 
\eqref{eq:lcar-flex-expectation-non-rho0} is equivalent to 
\begin{align*}
\sigma_U^2  
& = \frac{\tilde\sigma_U^2}{1-\tilde\lambda_U + \tilde\lambda_U \omega_i} + \tilde\sigma_\epsilon^2 - \sigma_\epsilon^2
- (\tilde\rho\tilde\sigma_U)^2 \frac{1}{1-\lambda_Z+\lambda_Z \omega_i}.
\end{align*}
Let $\tilde\lambda_U = \lambda_Z$,  $\tilde{\rho} = 1$ or $-1$, 
$\tilde\sigma_\epsilon^2 = \sigma_U^2 + \sigma_\epsilon^2$, 
and $\tilde\sigma_U^2 > 0$ be any positive number. Then the above equation will hold. We can then choose $\tilde{\rho}$ and $\tilde\sigma_U^2$ such that $\tilde\beta = \beta - \tilde\rho \tilde\sigma_U/\sigma_Z \ne \beta$. 

From the above, Proposition \ref{prop:nonid-lcar_rho0} holds. 
\end{proof}

\section{Technical details for the parsimonious model in \S \ref{sec:leroux_car}}\label{sec:proofs-lcar-pars}

\subsection{Proof of Theorem \ref{thm:lcar-pars-identifiability}}
To prove Theorem \ref{thm:lcar-pars-identifiability}, we need the following four lemmas. 

\begin{lemma}\label{lem:lcar-pars-strictly-increasing}
Consider any constants $a>0$ and $c\ge b \ge 0$. The function:
\[
g(r) = \frac{\frac{r}{\sqrt{a+r}}+b}{\sqrt{a+r}+c}
= \frac{r+b\sqrt{a+r}}{a+r+c\sqrt{a+r}}
\]
is strictly increasing in $r$ for $r\in[0,\infty)$.
\end{lemma}
\begin{proof}[Proof of Lemma \ref{lem:lcar-pars-strictly-increasing}]

By definition and some algebra, 

\begin{align*}
g'(r)\cdot (a+r+c\sqrt{a+r})^2 &= (1+\frac{b}{2\sqrt{a+r}})(a+r+c\sqrt{a+r}) - (r+b\sqrt{a+r})(1+\frac{c}{2\sqrt{a+r}}) \\
&=a+r+c\sqrt{a+r} + \frac{b(a+r)}{2\sqrt{a+r}} + \frac{bc}{2} - \left( r+b\sqrt{a+r} + \frac{rc}{2\sqrt{a+r}} + \frac{bc}{2} \right) \\
&= a+(c-b)\sqrt{a+r} + \frac{ab-(c-b)r}{2\sqrt{a+r}}\\
&= a+\frac{ab}{2\sqrt{a+r}} + (c-b)\left( \sqrt{a+r}-\frac{r}{2\sqrt{a+r}} \right) \\
&= a + \frac{ab}{2\sqrt{a+r}} + \frac{c-b}{2\sqrt{a+r}}(2a+r) >0.
\end{align*} 
We can then derive Lemma \ref{lem:lcar-pars-strictly-increasing}. 
\end{proof}

\begin{lemma}\label{lem:lcar-pars-strictly-increasing2}
Consider any $r_1,r_2,r_3 \in [0,1)$. If $r_2>r_1$, then 
\[\frac{\sqrt{1+r_2\omega} + \sqrt{1+r_3\omega}}{\sqrt{1+r_1\omega} + \sqrt{1+r_3\omega}}\]
is strictly increasing in $\omega$ for $\omega>0$.
\end{lemma}

\begin{proof}[Proof of Lemma \ref{lem:lcar-pars-strictly-increasing2}]
Since logarithm is strictly increasing, it suffices to show that
\[f(\omega) = \log(\sqrt{1+r_2\omega} + \sqrt{1+r_3 \omega}) - \log(\sqrt{1+r_1\omega} + \sqrt{1+r_3\omega})\] is strictly increasing in $\omega \in (0,\infty)$. We now compute its derivative: 
\begin{align}
\frac{\partial f(\omega)}{\partial \omega} &=\frac{\frac{r_2}{2\sqrt{1+r_2\omega}}+\frac{r_3}{2\sqrt{1+r_3\omega}}}{\sqrt{1+r_2\omega} + \sqrt{1+r_3 \omega}} - \frac{\frac{r_1}{2\sqrt{1+r_1\omega}}+\frac{r_3}{2\sqrt{1+r_3\omega}}}{\sqrt{1+r_1\omega} + \sqrt{1+r_3 \omega}} \nonumber\\ 
&= \frac{1}{2\omega} \left(\frac{\frac{r_2}{\sqrt{\omega^{-1}+r_2}} + \frac{r_3}{\sqrt{\omega^{-1}+r_3}}}{\sqrt{\omega^{-1}+r_2} + \sqrt{\omega^{-1}+r_3}} - \frac{\frac{r_1}{\sqrt{\omega^{-1}+r_1}} + \frac{r_3}{\sqrt{\omega^{-1}+r_3}}}{\sqrt{\omega^{-1}+r_1} + \sqrt{\omega^{-1}+r_3}} \right). \label{eq:lcar-pars-increasing-derivative}
\end{align}
Now we apply Lemma \ref{lem:lcar-pars-strictly-increasing} with $a=\omega^{-1}, b= r_3/\sqrt{\omega^{-1}+r_3}, c = \sqrt{\omega^{-1}+r_3}$. Note that $r_2 > r_1$, and  $c-b = (\omega^{-1}+r_3-r_3)/\sqrt{\omega^{-1}+r_3} = \omega^{-1}/\sqrt{\omega^{-1}+r_3} > 0$. 
From Lemma \ref{lem:lcar-pars-strictly-increasing}, \eqref{eq:lcar-pars-increasing-derivative} must be positive. 
Therefore, Lemma \ref{lem:lcar-pars-strictly-increasing2} holds. 
\end{proof}

\begin{lemma}\label{lem:lcar-pars-linind}
Consider any $a_1, a_2, a_3\in\mathbb{R}$ and $\lambda_1,\lambda_2,\lambda_3 \in [0,1)$. If $\lambda_1, \lambda_2,\lambda_3$ are all distinct, and for $\omega=0$ and at least two distinct positive values of $\omega$, 
\begin{equation}\label{eq:lcar-pars-linind}
\frac{a_1}{\sqrt{1-\lambda_1+\lambda_1\omega}} + \frac{a_2}{\sqrt{1-\lambda_2+\lambda_2\omega}} + \frac{a_3}{\sqrt{1-\lambda_3+\lambda_3\omega}} = 0,
\end{equation}
then $a_1=a_2=a_3=0$.
\end{lemma}

\begin{proof}[Proof of Lemma \ref{lem:lcar-pars-linind}]
Without loss of generality, we assume $\lambda_1 < \lambda_2$. Let $b_k = a_k/\sqrt{1-\lambda_k}$ and $r_k = \lambda_k/(1-\lambda_k)$ for $k=1,2,3$. Since $\lambda_1,\lambda_2,\lambda_3$ are all distinct, $r_1,r_2$ and $r_3$ must also be all distinct. Because \eqref{eq:lcar-pars-linind} holds for $\omega=0$, we have $b_1+b_2+b_3=0$, and consequently $b_3 = -(b_1+b_2)$. For any positive $\omega$ such that \eqref{eq:lcar-pars-linind} holds, we then have 
\begin{align}
&\frac{b_1}{\sqrt{1+r_1\omega}} + \frac{b_2}{\sqrt{1+r_2\omega}} + \frac{-(b_1+b_2)}{\sqrt{1+r_3\omega}} =0 \nonumber\\
\iff & b_1\left(\frac{1}{\sqrt{1+r_1\omega}} - \frac{1}{\sqrt{1+r_3\omega}} \right) + b_2\left( \frac{1}{\sqrt{1+r_2\omega}} - \frac{1}{\sqrt{1+r_3\omega}}\right) =0. \label{eq:lcar-pars-interm1} 
\end{align}
Note that
\begin{align*}
\frac{1}{\sqrt{1+r_1\omega}} - \frac{1}{\sqrt{1+r_3\omega}} &= \frac{\sqrt{1+r_3\omega} - \sqrt{1+r_1\omega}}{\sqrt{1+r_1\omega}\sqrt{1+r_3\omega}} \\
&= \frac{(1+r_3\omega) - (1+r_1\omega)}{\sqrt{1+r_1\omega}\sqrt{1+r_3\omega}(\sqrt{1+r_1\omega}+\sqrt{1+r_3\omega})}\\
&= \frac{(r_3-r_1)\omega}{\sqrt{1+r_1\omega}\sqrt{1+r_3\omega}(\sqrt{1+r_1\omega}+\sqrt{1+r_3\omega})}
\end{align*} 
and by analogous calculation,
\[\frac{1}{\sqrt{1+r_2\omega}} - \frac{1}{\sqrt{1+r_3\omega}} = \frac{(r_3-r_2)\omega}{\sqrt{1+r_2\omega}\sqrt{1+r_3\omega}(\sqrt{1+r_2\omega}+\sqrt{1+r_3\omega})}. \]
From \eqref{eq:lcar-pars-interm1}, for any positive $\omega$ such that \eqref{eq:lcar-pars-linind} holds, we have
\begin{align}\label{eq:b1_b2_re}
&\frac{(r_3-r_1)\omega b_1}{\sqrt{1+r_1\omega}\sqrt{1+r_3\omega}(\sqrt{1+r_1\omega}+\sqrt{1+r_3\omega})} + \frac{(r_3-r_2)\omega b_2}{\sqrt{1+r_2\omega}\sqrt{1+r_3\omega}(\sqrt{1+r_2\omega}+\sqrt{1+r_3\omega})} =0 
\nonumber
\\
\implies & \frac{b_1(r_3-r_1)}{\sqrt{1+r_1\omega}(\sqrt{1+r_1\omega} + \sqrt{1+r_3\omega})} = \frac{-b_2(r_3-r_2)}{\sqrt{1+r_2\omega}(\sqrt{1+r_2\omega} + \sqrt{1+r_3\omega})}. 
\end{align}

Now suppose that $b_1\neq 0$. Then $b_2$ is also nonzero, and consequently  
\begin{equation}\label{eq:lcar-pars-interm2}
\frac{\sqrt{1+r_2\omega}(\sqrt{1+r_2\omega}+\sqrt{1+r_3\omega})}{\sqrt{1+r_1\omega}(\sqrt{1+r_1\omega}+\sqrt{1+r_3\omega})} = \frac{-b_2(r_3-r_2)}{b_1(r_3-r_1)}
\end{equation} is a constant for all positive $\omega$ such that \eqref{eq:lcar-pars-linind}  holds. Recall that $\lambda_2>\lambda_1$ by our assumption, which immediately implies that $r_2 > r_1$ by definition. We can then verify that  $(1+r_2\omega)/(1+r_1\omega)$ is strictly increasing in  $\omega \in (0, \infty)$.
From Lemma \ref{lem:lcar-pars-strictly-increasing2}, 
\[ \frac{\sqrt{1+r_2\omega} + \sqrt{1+r_3\omega}}{\sqrt{1+r_1\omega} + \sqrt{1+r_3\omega}} \]
is also strictly increasing in $\omega\in (0,\infty)$. Thus the LHS of \eqref{eq:lcar-pars-interm2} is strictly increasing in $\omega\in(0,\infty)$, leading to a contradiction that it is constant for at least two distinct positive values of $\omega$. 

From the above, we must have $b_1=0$, which further implies that $b_2 = 0$ from \eqref{eq:b1_b2_re} and consequently $b_3=0$. By definition,  $a_1=a_2=a_3=0$. Therefore, Lemma \ref{lem:lcar-pars-linind} holds. 
\end{proof}

\begin{proof}[Proof of Theorem \ref{thm:lcar-pars-identifiability}] 
By definition, 
\begin{align*}
    \Sigma_{UU}' & = \sigma_U^2 \{ (1-\lambda_U) I_n + \lambda_U \Omega \}^{-1}, 
    \qquad \quad 
    \Sigma_{ZZ}' = \sigma_Z^2 \{ (1-\lambda_Z) I_n + \lambda_Z \Omega \}^{-1}, \\
    \Sigma_{UZ}' & = \rho\sigma_U\sigma_Z \left[\{(1-\lambda_Z)I_n+\lambda_Z \Omega\}^{1/2}\{(1-\lambda_U)I_n + \lambda_U\Omega\}^{1/2}\right]^{-1}, 
\end{align*}
where $\Omega$ is a diagonal matrix with diagonal elements $\omega_1, \ldots, \omega_n$, and $\omega_1, \ldots, \omega_n$ are the eigenvalues of $D-W$. 
Because $D-W$ is diagonally dominant, the eigenvalues $\omega_1, \ldots, \omega_n$ are all nonnegative. 
By the same logic as \eqref{eq:condvar}, the observed data distribution is determined by 
\begin{align}\label{eq:lcar_obs-pars}
    \var(Z') & = \Sigma_{ZZ}'
    = \sigma_Z^2 \{ (1-\lambda_Z) I_n + \lambda_Z \Omega \}^{-1}, 
    \nonumber\\
    \cov(Y', Z') \var(Z')^{-1} & = \beta I_n + \Sigma_{UZ}' \Sigma'^{-1}_{ZZ}
    \nonumber
    \\
    & =\beta I_n + \rho \frac{\sigma_U}{ \sigma_Z} \{ (1-\lambda_{U} )I_n + \lambda_{U} \Omega \}^{-1/2} \{ (1-\lambda_Z) I_n + \lambda_Z \Omega \}^{1/2},
    \nonumber
    \\
    \var(Y' \mid Z') & = \Sigma_{UU}' - \Sigma_{UZ}' \Sigma'^{-1}_{ZZ} \Sigma'_{ZU} + \sigma_{\epsilon}^2 I_n
    \nonumber
    \\
    & = 
    (1-\rho^2) \sigma_U^2 \{ (1-\lambda_{U} )I_n + \lambda_{U} \Omega \}^{-1} 
    + \sigma_{\epsilon}^2 I_n. 
\end{align}
Let $(\sigma_Z^2, \lambda_Z, \sigma_U^2, \lambda_U, \rho, \beta, \sigma_\epsilon^2)$ be the true data generating parameters, 
and $(\tilde\sigma_Z^2, \tilde\lambda_Z, \tilde\sigma_U^2, \allowbreak\tilde\lambda_U, \tilde\rho, \tilde\beta, \tilde\sigma_\epsilon^2)$ be parameters leading to the same observable distribution of $(Y, Z)$. 
From the expression of $\var(Z')$ in \eqref{eq:lcar_obs-pars}, we have, for $1\le i \le n$, 
\begin{equation}\nonumber%
\frac{\sigma_Z^2}{1-\lambda_Z+\lambda_Z\omega_i} = \frac{\tilde\sigma_Z^2}{1-\tilde\lambda_Z+\tilde\lambda_Z\omega_i}. 
\end{equation} From Lemma \ref{lem:lcar-varz}, because there are at least two distinct $\omega_i$s, we have $(\tilde\sigma_Z^2, \tilde\lambda_Z) = (\sigma_Z^2, \lambda_Z)$. From the expression of  $\cov(Y', Z') \var(Z')^{-1}$ and $\var(Y'|Z')$ in \eqref{eq:lcar_obs-pars}, we then have
\begin{align}
\label{eq:lcar-pars-expectation}
\beta + \rho \cdot \frac{\sigma_U}{\sigma_Z} \sqrt{\frac{1-\lambda_Z + \lambda_Z \omega_i}{1-\lambda_U + \lambda_U \omega_i}} & = \tilde\beta + \tilde\rho \cdot \frac{\tilde\sigma_U}{\sigma_Z} \sqrt{\frac{1-\lambda_Z + \lambda_Z \omega_i}{1-\tilde\lambda_U + \tilde\lambda_U \omega_i}}
\\
\label{eq:lcar-pars-condvar}
\frac{(1-\rho^2)\sigma_U^2}{1-\lambda_U + \lambda_U\omega_i} + \sigma_\epsilon^2 & = \frac{(1-\tilde\rho^2)\tilde\sigma_U^2}{1-\tilde\lambda_U + \tilde\lambda_U\omega_i} + \tilde\sigma_\epsilon^2.
\end{align}

We first consider case (i) where
$\lambda_U\neq \lambda_Z$ and $\rho\neq 0$. From the discussion before, $(\tilde\sigma_Z^2, \tilde\lambda_Z) = (\sigma_Z^2, \lambda_Z)$. From \eqref{eq:lcar-pars-expectation}, we have 
\[\frac{\beta-\tilde\beta}{\sqrt{1-\lambda_Z+\lambda_Z\omega_i}} + \frac{\rho\sigma_U/\sigma_Z}{\sqrt{1-\lambda_U + \lambda_U \omega_i}} - \frac{\tilde\rho \tilde\sigma_U/\sigma_Z}{\sqrt{1-\tilde\lambda_U + \tilde\lambda_U \omega_i}} = 0.\]
Let $v = (1,\ldots,1)^\top \in \mathbb{R}^n$. We can verify that that $(D-W)v=0$. Thus, $0$ is an eigenvalue of $D-W$, i.e., one of the $\omega_i$s is zero. Because $\{\omega_1, \ldots, \omega_n\}$ contains contains at least three distinct values, from Lemma \ref{lem:lcar-pars-linind}, if $\lambda_Z, \lambda_U, \tilde\lambda_U$ are all distinct, then we must have $\rho\sigma_U/\sigma_Z=0$ and consequently $\rho =0$, leading to a contradiction. Thus, $\lambda_Z, \lambda_U, \tilde\lambda_U$ cannot be all distinct. Because $\lambda_Z\neq \lambda_U$ in this case, we have  either $\tilde\lambda_U = \lambda_Z$ or $\tilde\lambda_U = \lambda_U$. If $\tilde\lambda_U=\lambda_Z$, then the RHS of \eqref{eq:lcar-pars-expectation} takes a constant value for all $\omega_i$s, while the LHS is strictly monotone in $\omega_i$s, leading to a contradiction. Thus, we must have $\tilde\lambda_U=\lambda_U$. 
From \eqref{eq:lcar-pars-expectation}, for all $1\le i \le n$, we then have 
\begin{equation}\label{eq:lcar-pars-beta-interm2}
\beta-\tilde\beta = \frac{\tilde\rho\tilde\sigma_U - \rho\sigma_U}{\sigma_Z}\cdot \sqrt{\frac{1-\lambda_Z + \lambda_Z \omega_i}{1-\lambda_U + \lambda_U \omega_i}}.
\end{equation} 
Recall that $\lambda_Z\neq\lambda_U$. If $\tilde\rho\tilde\sigma_U \neq \rho\sigma_U$, then the RHS of \eqref{eq:lcar-pars-beta-interm2} is strictly monotone in $\omega_i$, leading to a contradiction. Thus, we must have $\tilde\rho \sigma_U = \rho\sigma_U$ and $\tilde\beta=\beta$.
In sum, in this case, we have 
$(\tilde\sigma_Z^2, \tilde\lambda_Z, \tilde\lambda_U, \tilde\beta, \tilde\rho \tilde\sigma_U) = (\sigma_Z^2, \lambda_Z, \lambda_U, \beta, \rho \sigma_U)$.

We consider then the case where 
$\lambda_U\neq \lambda_Z, \lambda_U \neq 0$ and $\rho^2\neq 1$. 
From the conditions in Theorem \ref{thm:lcar-pars-identifiability}, 
$\{\omega_1, \ldots, \omega_n\}$ contains at least three distinct values. From \eqref{eq:lcar-pars-condvar} and using Lemma \ref{lem:lcar-expression-linear-independence} with $K=1$ or $K=2$ (depending on whether $\tilde\lambda_U=0$), if $\tilde\lambda_U \neq \lambda_U$, then $(1-\rho^2)\sigma_U^2 =0$, leading to a contradiction. Thus, we have $\tilde\lambda_U = \lambda_U \neq 0$. Using Lemma \ref{lem:lcar-expression-linear-independence} again with $K=1$, we then have $(1-\tilde\rho^2)\tilde\sigma_U^2 = (1-\rho^2)\sigma_U^2$ and $\tilde\sigma_\epsilon^2 = \sigma_\epsilon^2$. In addition, \eqref{eq:lcar-pars-expectation} becomes equivalent to 
\begin{equation}\label{eq:lcar-pars-exp-interm}
(\beta-\tilde\beta) + \frac{\rho \sigma_U - \tilde\rho \tilde\sigma_U}{\sigma_Z}\sqrt{\frac{1-\lambda_Z + \lambda_Z \omega_i}{1-\lambda_U+\lambda_U\omega_i}} =0.
\end{equation}
Because $\lambda_U\neq \lambda_Z$, $(1-\lambda_Z+\lambda_Z\omega_i) /(1-\lambda_U+\lambda_U \omega_i)$ is strictly in monotone in $\omega_i$. From \eqref{eq:lcar-pars-exp-interm}, we must have $\tilde\rho\tilde\sigma_U = \rho\sigma_U$ and $\beta = \tilde\beta$. Because $(1-\tilde\rho^2)\tilde\sigma_U^2 = (1-\rho^2)\sigma_U^2$, we further have $\tilde\sigma_U^2 = \sigma_U^2$ and $\tilde\rho = \rho$. 
Thus, in this case, 
we have 
$(\tilde\sigma_Z^2, \tilde\lambda_Z, \tilde\sigma_U^2, \tilde\lambda_U, \tilde\rho, \tilde\beta, \tilde\sigma_\epsilon^2) = (\sigma_Z^2, \lambda_Z, \sigma_U^2, \lambda_U, \rho, \beta, \sigma_\epsilon^2)$. 

Finally, we consider the case where $\lambda_U\neq \lambda_Z$, $\lambda_U\neq 0$ and $\rho\neq 0$. 
Following the logic of case (i), 
we have $(\tilde\sigma_Z^2, \tilde\lambda_Z, \tilde\lambda_U, \tilde\beta, \tilde\rho \tilde\sigma_U) = (\sigma_Z^2, \lambda_Z, \lambda_U, \beta, \rho \sigma_U)$. 
From \eqref{eq:lcar-pars-condvar}, we then have, for all $1\le i\le n$,  
\[ \frac{(1-\rho^2)\sigma_U^2 - (1-\tilde\rho^2)\tilde\sigma_U^2}{1-\lambda_U + \lambda_U \omega_i} +\sigma_\epsilon^2 - \tilde\sigma_\epsilon^2  =0.\]
Because $\lambda_U\neq 0$, we can apply Lemma \ref{lem:lcar-expression-linear-independence} with $K=1$, and derive that $(1-\tilde\rho^2)\tilde\sigma_U^2 = (1-\rho^2)\sigma_U^2$ and $\tilde\sigma_\epsilon^2 = \sigma_\epsilon^2$. Since $\tilde\rho\sigma_U = \rho\sigma_U$, we further have 
$\tilde\sigma_U^2 = \sigma_U^2$ and $\tilde\rho =\rho$.
Thus, in this case, we have 
$(\tilde\sigma_Z^2, \tilde\lambda_Z, \tilde\sigma_U^2, \tilde\lambda_U, \tilde\rho, \tilde\beta, \tilde\sigma_\epsilon^2) = (\sigma_Z^2, \lambda_Z, \sigma_U^2, \lambda_U, \rho, \beta, \sigma_\epsilon^2)$.

Finally, we prove Theorem \ref{thm:lcar-pars-identifiability}. 
Case (i) in Theorem \ref{thm:lcar-pars-identifiability} follows from the first part above, 
and case (ii) in Theorem \ref{thm:lcar-pars-identifiability} follows from the second and third parts. 
Therefore, we derive Theorem \ref{thm:lcar-pars-identifiability}. 
\end{proof}

\subsection{A remark on the necessity of $\lambda_{U} \neq \lambda_Z$} 

\begin{proposition}\label{prop:lcar-pars-lambdaU-lambdaZ-equal}
Consider $n$ spatial locations with a proximity matrix $W$, and assume that the data generating process follows \eqref{eq:simp-model}, \eqref{eq:leroux_car} and \eqref{eq:lcar-crosscov-pars}. Suppose that $\lambda_U=\lambda_Z$. 
Then the treatment effect $\beta$ is not identifiable.
\end{proposition}
\begin{proof}[Proof of Proposition \ref{prop:lcar-pars-lambdaU-lambdaZ-equal}]

Let $(\sigma_Z^2, \lambda_Z, \sigma_U^2, \lambda_U, \rho, \beta, \sigma_\epsilon^2)$ be the true data generating parameters, 
where $\lambda_U=\lambda_Z$. 
From the proof of Theorem \ref{thm:lcar-pars-identifiability}, it suffices to find another set of parameters $(\tilde\sigma_Z^2, \tilde\lambda_Z, \tilde\sigma_U^2, \tilde\lambda_U, \tilde\rho, \tilde\beta, \tilde\sigma_\epsilon^2)$ such that 
$\tilde\beta \ne \beta$ and, for $1\le i \le n$, 
\begin{align}\label{eq:parsi_non_iden}
    \frac{\sigma_Z^2}{1-\lambda_Z+\lambda_Z\omega_i} & = \frac{\tilde\sigma_Z^2}{1-\tilde\lambda_Z+\tilde\lambda_Z\omega_i}, 
    \nonumber
    \\
\beta + \rho \cdot \frac{\sigma_U}{\sigma_Z} \sqrt{\frac{1-\lambda_Z + \lambda_Z \omega_i}{1-\lambda_U + \lambda_U \omega_i}} & = \tilde\beta + \tilde\rho \cdot \frac{\tilde\sigma_U}{\tilde \sigma_Z} \sqrt{\frac{1-\tilde \lambda_Z + \tilde \lambda_Z \omega_i}{1-\tilde\lambda_U + \tilde\lambda_U \omega_i}}, 
\nonumber
\\
\frac{(1-\rho^2)\sigma_U^2}{1-\lambda_U + \lambda_U\omega_i} + \sigma_\epsilon^2 & = \frac{(1-\tilde\rho^2)\tilde\sigma_U^2}{1-\tilde\lambda_U + \tilde\lambda_U\omega_i} + \tilde\sigma_\epsilon^2.
\end{align}
Let $(\tilde{\lambda}_Z, \tilde\sigma_Z, \tilde\sigma_\epsilon^2, \tilde\lambda_U) = (\lambda_Z, \sigma_Z, \sigma_\epsilon^2, \lambda_U)$, 
with the values of $(\tilde\sigma_U^2, \tilde\rho, \tilde\beta)$ to be specified later. 
Then the first equation in \eqref{eq:parsi_non_iden} holds obviously, and the latter two equations hold as long as  
\begin{align}\label{eq:parsi_non_iden2}
    \beta + \rho\sigma_U/\sigma_Z = \tilde\beta + \tilde\rho\tilde\sigma_U /\sigma_Z,
    \qquad 
    (1-\rho^2)\sigma_U^2 = (1-\tilde\rho^2)\tilde\sigma_U^2.
\end{align}
Note that we have three parameter values to specify, while they only need to satisfy two constraints from \eqref{eq:parsi_non_iden2}. 
We can verify that there exists $(\tilde\sigma_U^2, \tilde\rho, \tilde\beta)$ such that \eqref{eq:parsi_non_iden2} holds and $\tilde\beta \ne \beta$. 
Therefore, Proposition \ref{prop:lcar-pars-lambdaU-lambdaZ-equal} holds. 
\end{proof}

\subsection{A remark on the necessity of 
$\rho \ne 0$ or $\lambda_U \ne 0$
when $\lambda_{U} \neq \lambda_Z$} 

\begin{proposition}\label{prop:lcar-pars-lambdaU-lambdaZ-unequal-rho-lambdaU-zero}
Consider $n$ spatial locations with a proximity matrix $W$, and assume that the data generating process follows \eqref{eq:simp-model}, \eqref{eq:leroux_car} and \eqref{eq:lcar-crosscov-pars}. Suppose that $\lambda_U\ne \lambda_Z$ and $\rho = \lambda_U = 0$. 
Then the treatment effect $\beta$ is not identifiable.
\end{proposition}

\begin{proof}[Proof of Proposition \ref{prop:lcar-pars-lambdaU-lambdaZ-unequal-rho-lambdaU-zero}]

Let $(\sigma_Z^2, \lambda_Z, \sigma_U^2, \lambda_U, \rho, \beta, \sigma_\epsilon^2)$ be the true data generating parameters, 
where $\lambda_U\ne \lambda_Z$ and $\rho = \lambda_U = 0$.  
From the proof of Theorem \ref{thm:lcar-pars-identifiability}, it suffices to find another set of parameters $(\tilde\sigma_Z^2, \tilde\lambda_Z, \tilde\sigma_U^2, \tilde\lambda_U, \tilde\rho, \tilde\beta, \tilde\sigma_\epsilon^2)$ such that 
$\tilde\beta \ne \beta$ and, for $1\le i \le n$, 
the three equations in \eqref{eq:parsi_non_iden} hold. 
Let $(\tilde{\lambda}_Z, \tilde\sigma_Z) = (\lambda_Z, \sigma_Z)$, 
$\tilde \lambda_U = \lambda_Z$, $\tilde \rho$ equal $1$ or $-1$,
$\tilde \sigma_\epsilon = (1-\rho^2) \sigma^2_U + \sigma^2_\epsilon$
$\tilde{\beta} = \beta - \tilde{\rho} \tilde{\sigma}_U/\sigma_Z$, with the value of $\tilde\sigma_U^2$ to be specified later.
We can verify that the three equations in \eqref{eq:parsi_non_iden} hold. 
Moreover, $\tilde{\beta}$ is different from $\beta$. 
Therefore, Proposition \ref{prop:lcar-pars-lambdaU-lambdaZ-unequal-rho-lambdaU-zero} holds. 
\end{proof}

\section{Proof for Linear Model of Coregionalization}\label{sec:proofs-lmcoreg}

\subsection{Proof of Theorem \ref{thm:lmcoreg-nonid}} 

\begin{proof}[Proof of Theorem \ref{thm:lmcoreg-nonid}]
We first characterize the observable distribution of $(Y,Z)$. 
From the model assumptions in  \eqref{eq:simp-model} and \eqref{eq:LMC}, we have  
\begin{align*}
Y &= \beta\left(\sum_{t=1}^T a_t L_t\right) + \left(\sum_{t=1}^T b_t L_t\right)+ \epsilon = \sum_{t=1}^T (\beta a_t + b_t)L_t + \epsilon, \\
Z & = \sum_{t=1}^T a_t L_t. 
\end{align*}
By some algebra,  $(Y,Z)$ jointly follows a multivariate Gaussian distribution:
\begin{align*}
\begin{pmatrix}
Y \\
Z\end{pmatrix} &= \mathcal{N} \left(\begin{pmatrix} 0 \\0\end{pmatrix}, \begin{pmatrix} \Sigma_{YY} & \Sigma_{YZ} \\ \Sigma_{ZY} & \Sigma_{ZZ} \end{pmatrix}\right), 
\end{align*}
and their covariance matrices have the following forms:
\begin{align}\label{eq:lmcoreg}
\Sigma_{YY} &= \sigma_\epsilon^2 I_n + \sum_{t=1}^T (\beta a_t+b_t)^2 \Sigma_t, 
\quad 
 \Sigma_{YZ} = \sum_{t=1}^T (\beta a_t + b_t)a_t \Sigma_t,
\quad 
\Sigma_{ZZ} = \sum_{t=1}^T a_t^2 \Sigma_t.
\end{align}

We then prove the nonidentifiability of the treatment effect $\beta$. 
Let $(a_1,\ldots,a_T, b_1,\ldots,b_T, \beta, \phi_1,\ldots,\phi_T,$ $\sigma_\epsilon^2)$ be the true data generating parameters.
Below we find an alternative set of parameters  $(\tilde{a}_1,\ldots,\tilde{a}_T, \tilde{b}_1,\ldots,\tilde{b}_T, \allowbreak \tilde{\beta}, \allowbreak \tilde{\phi}_1,\ldots,\tilde{\phi}_T, \allowbreak \tilde{\sigma}_\epsilon^2)$ with $\tilde\beta \neq \beta$ 
that leads  to the same observable distribution of $(Y,Z)$. 
Specifically, for any $\delta \in \mathbb{R}$, we define 
\begin{align*}
    \tilde{a}_t = a_t, \quad \tilde{\phi}_t = \phi_t, \quad 
    \tilde\sigma^2_\epsilon = \sigma^2_\epsilon, \quad 
    \tilde\beta = \beta - \delta, \quad 
    \tilde{b}_t = b_t + \delta a_t. 
\end{align*}
We can then verify that 
\begin{align*}
\tilde\beta \tilde a_t + \tilde{b}_t &= (\beta-\delta)a_t + b_t + \delta a_t = \beta a_t + b_t, 
\end{align*} 
and thus all the equations in \eqref{eq:lmcoreg} hold.
Moreover, once we choose a nonzero $\delta$, $\tilde{\beta}$ will be different from $\beta$.

From the above, Theorem \ref{thm:lmcoreg-nonid} holds. 
\end{proof}

\section{Proofs for Parametric Bivariate Covariance Models}\label{sec:proofs-continuous-bivariate}

\subsection{Proof of Proposition \ref{prop:K-linind}}

\begin{proposition}\label{prop:K-linind}
 The exponential, Gaussian, and powered exponential families of covariance functions are $K$-linearly independent with respect to any $K$ distinct elements in $\mathbb{R}_{+}$. 
\end{proposition}
To prove Proposition \ref{prop:K-linind}, we need the following lemma. 
\begin{lemma}\label{lem:k-linind-generalized}
Let $f: \mathbb{R}_+\to \mathbb{R}_+$ be a strictly monotonic and differentiable function. 
Consider the family of function of form $f(w)^{\psi}$ for  $\psi \in \mathbb{R}$, i.e., 
$\mathcal{F} = \{ f(\cdot)^{\psi}: \psi \in \mathbb{R} \}$.
Then $\mathcal{F}$ is $K$-linearly independent with respect to any collection of $K$ distinct values in $\mathbb{R}_+$, for any positive integer $K$.  
\end{lemma}
\begin{proof}[Proof of Lemma \ref{lem:k-linind-generalized}]
We prove this lemma by induction.

Consider first the case where $K=1$. 
For any given $w \in \mathbb{R}_+$ and $\psi \in \mathbb{R}$, suppose that $a f(w)^{\psi} = 0$ for some $a\in \mathbb{R}$. 
Because $f(w)>0$ by definition, we must have $a = 0$. 
Thus, $\mathcal{F}$ is $1$-linearly independent with respect to any collection of $1$ value in $\mathbb{R}_+$.

Now assume 
that $\mathcal{F}$ is $(n-1)$-linearly independent with respect to any collection of $(n-1)$ distinct values in $\mathbb{R}_+$, for some $n\ge 2$. 
Consider any $w_1, w_2, \ldots, w_n\in \mathbb{R}_+$ that are mutually distinct, and any nonzero $\psi_1, \psi_2, \ldots, \psi_n$ that are mutually distinct. 
Without loss of generality, we assume $\psi_1< \psi_2 < \ldots < \psi_n$
and 
$w_1 < w_2 < \ldots < w_n$. 
Suppose that, for some $a_1, a_2, \ldots, a_n \in \mathbb{R}$, 
\begin{align}\label{eq:linear_depend_n}
a_1 f(w)^{\psi_1} + a_2 f(w)^{\psi_2} + 
\ldots + 
a_n f(w)^{\psi_n} = 0 
\quad \text{for } 
w \in \{ w_1, w_2, \ldots, w_n \}. 
\end{align}
This then implies that, for $w \in \{ w_1, w_2, \ldots, w_n \}$, 
\begin{align*}
0 & = f(w)^{-\psi_1} \big\{ a_1 f(w)^{\psi_1} + 
\ldots + 
a_n f(w)^{\psi_n} \big\}
= a_1 + a_2 f(w)^{\psi_2-\psi_1} + 
\ldots + 
a_n f(w)^{\psi_n-\psi_1}\\
& = a_1 + a_2 f(w)^{b_2} + 
\ldots + 
a_n f(w)^{b_n}, 
\end{align*}
where $b_k = \psi_k-\psi_1 > 0$ for $2\le k \le n$. 
By Rolle’s Theorem, there exist $\tilde{w}_1 < \tilde{w}_2 < \ldots \tilde{w}_{n-1}$ such that, 
$w_k < \tilde{w}_k < w_{k+1}$ for $1\le k \le n-1$, and, 
for $w \in \{\tilde{w}_1, \tilde{w}_2, \ldots, \tilde{w}_{n-1}\}$
\begin{align*}
0 
& = a_2b_2 f(w)^{b_2-1} f'(w) + 
\ldots + 
a_n b_n f(w)^{b_n-1} f'(w) \\
&= 
f'(w) \big\{ 
a_2 b_2 f(w)^{b_2-1} + 
\ldots + 
a_n b_n f(w)^{b_n-1} 
\big\}, 
\end{align*}
where $f'$ denotes the derivative of $f$. 
Because $f$ is strictly monotone, this then implies that, for $w \in \{\tilde{w}_1, \tilde{w}_2, \ldots, \tilde{w}_{n-1}\}$, 
\begin{align*}
0 
& =
a_2 b_2 f(w)^{b_2-1} + 
\ldots + 
a_n b_n f(w)^{b_n-1}.  
\end{align*}
By construction, $\tilde{w}_1, \tilde{w}_2, \ldots, \tilde{w}_{n-1}\in \mathbb{R}_+$ are mutually independent, and $b_2-1 < b_3-1< \ldots < b_n-1$ are mutually independent. 
Because $\mathcal{F}$ is $(n-1)$-linearly independent with respect to any collection of $(n-1)$ distinct values in $\mathbb{R}_+$, we must have 
$a_2b_2 = \ldots = a_nb_n = 0$. 
Because, by construction, $b_k$ is positive for $2\le k \le n$, we then have $a_2 = \ldots = a_n = 0$. 
From \eqref{eq:linear_depend_n}, we then have $a_1 = 0$. 
Therefore, $\mathcal{F}$ is also $n$-linearly independent with respect to any collection of $n$ distinct values in $\mathbb{R}_+$. 

From the above, Lemma \ref{lem:k-linind-generalized} holds. 
\end{proof}

\begin{proof}[Proof of Proposition \ref{prop:K-linind}]

The powered exponential family of covariance functions is $\mathcal{F} = \{C(\psi, w) = e^{-(w/\psi)^c}: \psi > 0\}$, where $c$ is some fixed positive real number. 
Note that 
$
C(\psi, w)  = e^{-(w/\psi)^c} = (e^{-w^c})^{1/\psi^c} = f(w)^{\psi}, 
$
where
$f(w) = e^{-w^c}$ is strictly monotone in $w\in \mathbb{R}_+$, 
and $\psi = 1/\psi^c$ is also strictly monotone in $\psi\in \mathbb{R}_+$. 
From Lemma \ref{lem:k-linind-generalized}, 
we can know that $\mathcal{F}$ is $K$-linearly independent with respect to any $K$ distinct positive elements.
Note that both exponential and Gaussian families of covariance functions are special cases of the powered exponential family. 
We can then derive Proposition \ref{prop:K-linind}. 
\end{proof}

\subsection{Proof of Theorem \ref{thm:flexible-bivariate-id}}

Before we give the proof of this Theorem, we expound on the remark given in the Theorem. We will also need an additional lemma which we prove.

\begin{proposition}\label{prop:bivariate-scaled-identity}
Consider $n$ spatial locations with a distance matrix $W$, 
and assume that the data generating process follows \eqref{eq:simp-model} and \eqref{eq:bivariate-cov-framework}. If $\cov(Y,Z)\var(Z)^{-1}$ is not a scaled identity matrix, then $\rho \neq 0$ and $\psi_{UZ} \neq \psi_Z$.
\end{proposition}
\begin{proof}[Proof of Proposition \ref{prop:bivariate-scaled-identity}]
We prove the contrapositive - if either $\rho=0$ or $\psi_{UZ} = \psi_Z$, then our parameter must be a scaled identity. For notation's sake write $C_{Z}$ for the matrix with $ij$th entry given by $C(\psi_{Z}, W_{ij})$ and $C_{UZ}$ for the matrix with $ij$th entry given by $C(\psi_{UZ}, W_{ij})$. Then \[ \cov(Y,Z)\var(Z)^{-1} = \beta I_n + \Sigma_{UZ}\Sigma_{ZZ}^{-1} = \beta I_n + \frac{\rho \sigma_U}{\sigma_Z} C_{UZ}C_Z^{-1}.\]
It is easy to see then that if $\rho=0$ that $\cov(Y,Z)\var(Z)^{-1}=\beta I_n$, and similarly if $\psi_{UZ} = \psi_Z$ then $C_{UZ} = C_Z^{-1}$ and $\cov(Y,Z)\var(Z)^{-1} = (\beta + \rho\sigma_U\sigma_Z^{-1})I_n$, again a scaled identity.
\end{proof}

\begin{lemma}\label{lemma:linear_indep_less_K}
    If $\mathcal{F} = \{ f(\psi, \cdot): \psi \in \psi \}$ is $K$-linearly independent with respect to $\mathcal{S}$, 
    then it is also $J$-linearly independent with respect to $\mathcal{S}$ for all $1\le J\le K$.  
\end{lemma}

\begin{proof}[Proof of Lemma \ref{lemma:linear_indep_less_K}]
It suffices to show that if $\mathcal{F}$ is $K$-linearly independent with respect to $S$, then it must also be $K-1$-linearly independent with respect to $\mathcal{S}$. 
Consider any mutually distinct $\psi_1, \ldots, \psi_{K-1}\in \psi$. 
Suppose that, for some $a_1, \ldots, a_{K-1} \in \mathbb{R}$, we have 
\begin{align*}
    a_1 f(\psi_1, w) + \ldots + a_{K-1} f(\psi_{K-1},ws) = 0 \quad \text{for all } w \in \mathcal{S}. 
\end{align*}
Let $\psi_K \in \psi$ by any value that is distinct from $\psi_1, \ldots, \psi_{K-1}$, and $a_K = 0$. We then have 
\begin{align*}
    a_1 f(\psi_1, w) + \ldots + a_{K-1} f(\psi_{K-1}, w) + a_{K} f(\psi_{K}, w) = 0 \quad \text{for all } w \in \mathcal{S}. 
\end{align*}
Because $\mathcal{F}$ is $K$-linearly independent with respect to $\mathcal{S}$, we must have $a_1 = \ldots = a_{K-1} = 0$. 
Therefore, $\mathcal{F}$ is also $(K-1)$-linearly independent with respect to $\mathcal{S}$.
\end{proof}

\begin{proof}[Proof of Theorem \ref{thm:flexible-bivariate-id}]
First note that if the function family is $3$-linearly independent with respect to the off-diagonal entries of $W$, then it is $3$-linearly independent with respect to $W$ and $2$-linearly independent with respect to the off-diagonal values of $W$. Now, from \eqref{eq:joint_ZY}, the observed data distribution is determined by: 
\begin{align*}
[\var(Z)]_{ij} &= \sigma_Z^2 C(\psi_Z, W_{ij}),\\
[\var(Y)]_{ij} &= \beta^2 \sigma_Z^2 C(\psi_{Z}, W_{ij}) + 2\beta\rho \sigma_U\sigma_ZC(\psi_{UZ}, W_{ij}) +  \sigma_U^2 C(\psi_{U}, W_{ij}) + \sigma_\epsilon^2(I_n)_{ij} , \\
[\cov (Y,Z)_{ij}] &= \beta \sigma_Z^2 C(\psi_{Z}, W_{ij}) + \rho\sigma_U\sigma_Z C(\psi_{UZ}, W_{ij}),
\end{align*} 
where $1\le i, j \le n.$
Consider any alternative set of parameters $(\tilde\beta, \tilde\sigma_Z, \tilde\sigma_U, \tilde\rho, \tilde\psi_Z, \tilde\psi_U, \tilde\psi_{UZ}, \tilde\sigma_\epsilon)$ that yield the same observed data distribution. 
Below we will prove that $(\tilde\beta, \tilde\sigma_Z, \tilde\sigma_U, \tilde\rho, \tilde\psi_Z, \tilde\psi_U, \tilde\psi_{UZ}, \tilde\sigma_\epsilon^2) = (\beta, \sigma_Z, \sigma_U, \rho, \psi_Z, \psi_U, \psi_{UZ}, \sigma_\epsilon^2)$. 

First, we prove $\psi_Z = \tilde \psi_Z$ and $\sigma_Z^2 = \tilde\sigma_Z^2$. 
From the expression of $\var(Z)$, we have
\[\sigma_Z^2 C(\psi_{Z}, W_{ij}) = \tilde\sigma_Z^2 C(\tilde{\psi}_{Z}, W_{ij})
\quad 
\text{for all $1\le i,j \le n$}.
\]   
From Lemma \ref{lemma:linear_indep_less_K} and the condition in the theorem, 
if $\psi_Z \neq \tilde\psi_Z$, then we must have $\sigma_Z^2 = \tilde\sigma_Z^2=0$, which leads to a contradiction. 
Thus, we must have $\psi_Z = \tilde \psi_Z$ and consequently $\sigma_Z^2 = \tilde\sigma_Z^2$.

Second, we prove $\psi_{UZ} = \tilde\psi_{UZ}$,  $\beta = \tilde\beta$ and $\rho\sigma_U = \tilde\rho\tilde\sigma_U$. 
From the expression of $\cov(Y,Z)$ and the first part, we have 
\[\beta\sigma_Z^2 C(\psi_{Z}, W_{ij}) + \rho \sigma_U\sigma_Z C(\psi_{UZ}, W_{ij}) = \tilde\beta\sigma_Z^2 C(\psi_{Z}, W_{ij}) + \tilde\rho \tilde\sigma_U\sigma_Z  C(\tilde\psi_{UZ}, W_{ij})
\ \  
\text{for all $1\le i,j \le n$},
\]
which immediately implies that 
\begin{equation}\label{eq:flex-bivariate-exp}(\beta-\tilde\beta)\sigma_Z^2 C(\psi_{Z}, W_{ij}) + \rho \sigma_U\sigma_Z C(\psi_{UZ}, W_{ij}) - \tilde\rho \tilde\sigma_U\sigma_Z C(\tilde\psi_{UZ}, W_{ij}) = 0
\ \ 
\text{for $1\le i,j \le n$}.
\end{equation}
From the condition in the theorem and Proposition \ref{prop:bivariate-scaled-identity}, $\psi_{Z} \ne \psi_{UZ}$. 
If further $\psi_{UZ} \ne \tilde\psi_{UZ}$, then, no matter $\psi_{Z} = \tilde \psi_{UZ}$ or not, by linear independence, we have $\rho \sigma_U\sigma_Z = 0$ and consequently $\rho = 0$, leading to a contradiction. 
Thus, we must have $\psi_{UZ} = \tilde\psi_{UZ}$.
Equation \eqref{eq:flex-bivariate-exp} then simplifies to 
\[
(\beta-\tilde\beta)\sigma_Z^2 C(\psi_{Z}, W_{ij}) + \sigma_Z (\rho \sigma_U - \tilde\rho \tilde\sigma_U )C(\psi_{UZ}, W_{ij}) = 0
\ \ 
\text{for $1\le i,j \le n$}.
\]
By linear independence, this implies that 
$\beta = \tilde\beta$ and $\rho\sigma_U = \tilde\rho\tilde\sigma_U$.

Finally, we prove $\psi_U=\tilde\psi_U, \sigma_U = \tilde\sigma_U$, $\rho = \tilde\rho$ and $\sigma_\epsilon^2 = \tilde\sigma_\epsilon^2$. 
From the expression of $\var(Y)$ and the previous two parts, we have  
\begin{align}\label{eq:varY_para}
    \sigma_U^2 C(\psi_U, W_{ij}) + \sigma_\epsilon^2 (I_n)_{ij} = \tilde\sigma_U^2 C(\tilde\psi_U, W_{ij})+ \tilde\sigma_\epsilon^2(I_n)_{ij}
\ \ 
\text{for $1\le i,j \le n$}, 
\end{align}
which immediately implies 
\[\sigma_U^2 C(\psi_U, W_{ij})  = \tilde\sigma_U^2 C(\tilde\psi_U, W_{ij})
\ \ 
\text{for $1\le i \ne j \le n$}. 
\] 
By linear independence, if $\psi_U \neq \tilde\psi_U$, 
then we have $\sigma_U^2 = \tilde\sigma_U^2=0$, leading to a contradiction. 
Thus, we must have $\psi_U = \tilde\psi_U$. Consequently, 
$
(\sigma_U^2 - \tilde\sigma_U^2) C(\psi_U, W_{ij}) = 0
$
for $1\le i \ne j \le n$. 
From Lemma \ref{lemma:linear_indep_less_K} and the linear independence, we must have $\sigma_U^2 = \tilde\sigma_U^2$. 
From \eqref{eq:varY_para} with $i=j$, we then have $\sigma_\epsilon^2 = \tilde\sigma_\epsilon^2$. 
Note that, from the second part, $\rho\sigma_U=\tilde\rho \tilde\sigma_U$. 
We can then derive that $\rho = \tilde\rho$. 

as well, and all parameters are identified.
\end{proof}

\subsection{A remark on the necessity of $\rho \neq 0$} 
We must assume the domain restriction that $\rho \neq 0$ or we can not be able to distinguish from observable quantities the case when $\rho \neq 0$ but $C_Z = C_{UZ}$ and there is meaningful confounding from the case when $\rho = 0$ and there is no confounding.

\begin{proposition}\label{prop:bivariate-flexible-rho-zero-nonid}
Consider $n$ spatial locations with a distance matrix $W$, 
and assume that the data generating process follows \eqref{eq:simp-model} and \eqref{eq:bivariate-cov-framework}. If it is possible that $\rho = 0$, then the parameters are not identifiable.
\end{proposition}

\begin{proof}[Proof of Proposition \ref{prop:bivariate-flexible-rho-zero-nonid}]
We exhibit a class of counterexamples. Consider data generating parameters $(\psi_U, \psi_Z, \psi_{UZ}, \beta, \sigma_Z, \sigma_U, \rho, \sigma_\epsilon)$ with $\psi_U=\psi_Z=\psi_{UZ}$. Then consider alternative parameters $(\tilde\psi_U, \tilde\psi_Z, \tilde\psi_{UZ}, \allowbreak \tilde\beta, \allowbreak \tilde\sigma_Z, \allowbreak\tilde\sigma_U, \tilde\rho, \tilde\sigma_\epsilon)$  where: $\tilde\sigma_Z = \sigma_Z, \tilde\psi_Z=\psi_Z, \tilde\psi_{UZ}=\psi_Z, \tilde\psi_U=\psi_U, \tilde\sigma_\epsilon=\sigma_\epsilon$, and $\tilde\beta$ is chosen so that $\tilde\sigma^2_U = \beta\rho\sigma_U\sigma_Z + \sigma_U^2 - \tilde\beta\sigma_Z\{(\tilde\beta-\beta)\sigma_Z-\rho\sigma_U\}$ is nonzero, and finally $\tilde\rho= \{(\tilde\beta-\beta)\sigma_Z-\rho\sigma_U\}/\tilde\sigma_U$. 

For notation's sake, let $C_Z$ denote the matrix whose $ij$th entry is given by $C(\psi_Z, W_{ij})$, and analogously for $C_U, C_{UZ}$; similarly let $\tilde C_Z$ denote matrix whose $ij$th entry is given by $C(\tilde\psi_Z, W_{ij})$, and analogously for $\tilde{C}_U, \tilde{C}_{UZ}$. Then we it remains to show that the following equations hold:

\[\sigma_Z^2 C_Z= \tilde \sigma_Z \tilde C_Z,\]
\[\beta \sigma_Z^2 C_Z + \rho \sigma_U \sigma_Z C_{UZ} =\tilde\beta \tilde\sigma_Z^2 \tilde C_Z +\tilde\rho \tilde\sigma_U \tilde\sigma_Z \tilde C_{UZ}, \]
\[\beta^2 \sigma_Z^2 C_Z + 2\beta\rho\sigma_U\sigma_Z C_{UZ} + \sigma_U^2 C_U + \sigma_\epsilon^2 I_n = \tilde\beta^2 \tilde\sigma_Z^2 \tilde C_Z + 2\tilde\beta\tilde\rho\tilde\sigma_U\tilde\sigma_Z \tilde C_{UZ} + \tilde\sigma_U^2\tilde C_U + \tilde\sigma_\epsilon^2 I_n.  \]

The first equation corresponding to $\var(Z)$ is automatically satisfied. The remaining equations simplify to
\begin{equation}\label{eq:bivariate-rho-zero-cov}
\beta \sigma_Z +\rho\sigma_U= \tilde\beta \sigma_Z + \tilde\rho \tilde\sigma_U ,
\end{equation}
\begin{equation}\label{eq:bivariate-rho-zero-var}
\beta\rho\sigma_U\sigma_Z+\sigma_U^2 = \tilde\beta\tilde\rho\tilde\sigma_U \sigma_Z + \tilde\sigma_U^2.
\end{equation}

Note that our selections for $\tilde\rho, \tilde\sigma_U$ give $\tilde\rho\tilde\sigma_U = (\tilde\beta-\beta)\sigma_Z -\rho\sigma_U$, and thus \eqref{eq:bivariate-rho-zero-cov} is satisfied. As for \eqref{eq:bivariate-rho-zero-var}, it suffices to show that:
\[\tilde\rho = \frac{\beta\rho\sigma_U\sigma_Z + \sigma_U^2-\tilde\sigma_U^2}{\tilde\beta\tilde\sigma_U \sigma_Z}.\] But this holds so long as $\beta\rho\sigma_U\sigma_Z+\sigma_U^2-\tilde\sigma_U^2 = \tilde\beta\sigma_Z\{(\tilde\beta-\beta)\sigma_Z-\rho\sigma_U\}$, which holds by definition.
\end{proof}

\section{Proofs for Parsimonious Mat\'ern model}\label{sec:proofs-matern}

In this case, we consider the case in which the smoothness parameters are all known. 
The corollary below shows the identifiability of the parsimonious Mat\'ern model. 

\begin{theorem}\label{cor:bivariate-matern-id-known-smoothness}
Consider $n$ spatial locations with a distance matrix $W$, and assume that the data generating process follows \eqref{eq:simp-model} and \eqref{eq:pars-matern}.
If the smoothness parameters $\nu_U$ and $\nu_Z$ are known and different, $E(YZ^\top)E(ZZ^\top)^{-1}$ is not proportional to the identity matrix, and the distance matrix $W$ has at least one nonzero element, then 
all the model parameters, including the treatment effect $\beta$,
are identifiable.
\end{theorem}
Here we impose a slightly different form of condition from theorems in the main paper. 
Instead of restricting the true data generating parameters, we require that the observable quantity $E(YZ^\top)E(ZZ^\top)^{-1}$ is not proportional to the identity matrix. This implies that $\rho\neq 0$ and $\\nu_U\neq \nu_Z$, but is stronger and is invoked due to the complicated form of the Mat\'ern function.

To prove Theorem \ref{cor:bivariate-matern-id-known-smoothness}, we need the following lemma about the properties of the Bessel function of the second kind.

\begin{lemma}\label{lem:bessel-properties}
Let $K_\nu$ denote the modified Bessel function of the second kind of order $\nu$. Then the following equations and inequalities hold:\\
(i) For any $z, \nu \in \mathbb{R}$,
\[\lim_{z\to\infty} \frac{K_\nu(z)}{\sqrt{\pi} e^{-z}/(2z)^{1/2}} = 1;\]
(ii) For $\nu_1 > \nu_2$ and any $z \in \mathbb{R}^+$, \[K_{\nu_1}(z) > K_{\nu_2}(z);\]
(iii) For any $z, \nu \in \mathbb{R}$, \[\frac{\partial}{\partial z} z^\nu K_\nu(z) = -z^\nu K_{\nu-1}(z);\]
(iv) For positive real $z$ and real $\nu$, \[K_\nu(z) > 0.\]

\end{lemma}

\begin{proof}[Proof of Lemmas \ref{lem:bessel-properties}]
The limit in (i) can be found at \cite[\href{https://dlmf.nist.gov/10.25\#E3}{(10.25.3)}]{NIST:DLMF}. The inequality in (ii) can be found on \cite[\href{https://dlmf.nist.gov/10.37\#E1}{(10.37)}]{NIST:DLMF}. The derivative expression in (iii) is given at \cite[\href{https://dlmf.nist.gov/10.29\#E4}{(10.29.4)}]{NIST:DLMF}. 
Note that $K_\nu(z) = K_{-\nu}(z)$ for all positive real $z$, as shown by \cite[\href{https://dlmf.nist.gov/10.27\#E3}{Section 10.27}]{NIST:DLMF}. 
The inequality in (iv) can then be shown by  \cite[\href{https://dlmf.nist.gov/10.37}{Section 10.37}]{NIST:DLMF}. 
\end{proof}

\begin{lemma}\label{lem:matern-fixed-smoothness-distinguishable}
Consider the Mat\'ern covariance function: 
$$
C(\phi, \nu; w) = \frac{2^{1-\nu}}{\Gamma(\nu)}\left( \frac{w}{\phi} \right)^\nu  K_\nu \left(\frac{w}{\phi}\right),
$$
where $\phi, \nu$ are positive and $w$ is nonnegative. 
For any given $\nu>0$ and $w>0$, the function $C(\phi, \nu; w)$ is strictly increasing in $\phi$, in the sense that $C(\phi_1, \nu; w) > C(\phi_2, \nu; w)$ for any $\phi_1 > \phi_2 > 0$. 
\end{lemma}

\begin{proof}[Proof of Lemma \ref{lem:matern-fixed-smoothness-distinguishable}]
It suffices to prove that $z^\nu K_\nu(z)$ is strictly decreasing in $z \in (0, \infty)$. 
From Lemma \ref{lem:bessel-properties}(iii) and (iv), we have 
\[
\frac{\partial}{\partial z} z^\nu K_\nu(z) = -z^\nu K_{\nu-1}(z) < 0.
\] 
Therefore, Lemma \ref{lem:matern-fixed-smoothness-distinguishable} holds. 
\end{proof}

\begin{proof}[Proof of Theorem \ref{cor:bivariate-matern-id-known-smoothness}]
First, from \eqref{eq:joint_ZY}, we have, for $1\le i,j\le n$, 
\begin{align*}\var(Z)_{ij} &= \sigma_Z^2 C(\phi, \nu_Z; W_{ij}),\\
\var(Y)_{ij} &= \beta^2 \sigma_Z^2 C(\phi, \nu_Z; W_{ij}) + 2\beta\rho \sigma_U\sigma_ZC(\phi, 2^{-1}(\nu_U+\nu_Z); W_{ij}) +  \sigma_U^2 C(\phi, \nu_U; W_{ij}) + \sigma_\epsilon^2(I_n)_{ij} , \\
\cov (Y,Z)_{ij} &= \beta \sigma_Z^2 C(\phi, \nu_Z; W_{ij}) + \rho\sigma_U\sigma_Z C(\phi, 2^{-1}(\nu_U+\nu_Z); W_{ij}).\end{align*} 
Recall that that $(\nu_U, \nu_Z)$ are known here. 
Let $(\beta, \phi,\sigma_Z^2, \sigma_U^2, \sigma_\epsilon^2, \rho)$ denote the true data generating parameters, and suppose that $(\tilde\beta, \tilde\phi,\tilde\sigma_Z^2, \tilde\sigma_U^2, \tilde\sigma_\epsilon^2, \tilde\rho)$ yield the same observed data distribution. 
Below we prove that $(\tilde\beta, \tilde\phi,\tilde\sigma_Z^2, \tilde\sigma_U^2, \tilde\sigma_\epsilon^2, \tilde\rho)$ must be the same as $(\beta, \phi,\sigma_Z^2, \sigma_U^2, \sigma_\epsilon^2, \rho)$. 

First, from the expression of $\var(Z)$, we must have for each $i,j$:
\[\sigma_Z^2 C(\phi,\nu_Z;W_{ij}) = \tilde\sigma_Z^2 C(\tilde\phi,\nu_Z;W_{ij}).\]
Because $C(\tilde\phi,\nu_Z; 0) =C(\phi,\nu_Z;0)= 1$ and $W_{ii}=0$ for all $i$, we must have $\sigma_Z^2 = \tilde\sigma_Z^2$, which further implies that for all $i,j$,  
\[ C(\phi,\nu_Z;W_{ij}) = C(\tilde\phi,\nu_Z;W_{ij}).\] 
Because $W$ contains at least one positive element, from Lemma \ref{lem:matern-fixed-smoothness-distinguishable}, 
the above equation cannot hold if $\phi \neq \tilde \phi$. Thus, we must have $\phi = \tilde \phi$. 

Second, we consider the expression of $\cov(Y,Z)$. 
From the diagonal entries of $\cov(Y,Z)$, we obtain
\begin{equation}\label{eq:bivariate-matern-fixed-smooth-diag}\beta \sigma_Z^2 + \rho \sigma_Z\sigma_U = \tilde\beta \sigma_Z^2 + \tilde\rho\sigma_Z\tilde\sigma_U.\end{equation} 
Next, our assumption that $\cov(Y,Z)\var(Z)^{-1}$ is not proportional to the identity matrix implies that 
\[ \cov(Y,Z)\var(Z)^{-1} - \beta I_n =\Sigma_{UZ}\Sigma_{ZZ}^{-1} \]
is not proportional to the identity matrix, where $\Sigma_{UZ}$ and $\Sigma_{ZZ}$ are defined in \eqref{eq:pars-matern}.
If $C(\phi, \nu_Z; W_{ij}) = C((\phi, 2^{-1}(\nu_Z+\nu_U), W_{ij})$ for all $i,j$, then $\Sigma_{UZ}\Sigma_{ZZ}^{-1} = \rho \sigma_U/\sigma_Z \cdot I_n$, leading to a contradiction. 
Thus, there must exist $(i', j')$ such that $W_{i'j'}>0$ and $C(\phi, \nu_Z; W_{i'j'}) \neq C((\phi, 2^{-1}(\nu_Z+\nu_U), W_{i'j'})$. 
Let $a=C(\phi, \nu_Z; W_{i'j'})$ and $b=C(\phi, 2^{-1}(\nu_Z+\nu_U); W_{i'j'})$. By construction, $a\neq b$. Moreover, from Lemma  \ref{lem:bessel-properties}(iv), both $a$ and $b$ are positive. 
From $\cov(Y,Z)_{i', j'}$, we then have  
\begin{equation}\label{eq:bivariate-matern-fixed-smooth-offdiag}a\beta \sigma_Z^2 + b\rho \sigma_Z\sigma_U = a \tilde\beta \sigma_Z^2 +b \tilde\rho \sigma_Z\tilde\sigma_U. \end{equation}
Multiplying \eqref{eq:bivariate-matern-fixed-smooth-diag} by $b$ and subtracting from \eqref{eq:bivariate-matern-fixed-smooth-offdiag}, we have
$
(a-b)\beta \sigma_Z^2 = (a-b)\tilde\beta \sigma_Z^2
$. Since both $\sigma_Z$ and $a-b$ are nonzero, we must have $\beta = \tilde\beta$. From \eqref{eq:bivariate-matern-fixed-smooth-diag}, we further have $\rho\sigma_U = \tilde\rho\tilde\sigma_U$. 

Third, we consider the expression of $\var(Y)$. From the above, for each $i,j$, we have 
\begin{align*} &\sigma_U^2 C(\phi, \nu_U; W_{ij}) + \sigma_\epsilon^2(I_n)_{ij}  = \tilde\sigma_U^2 C(\phi, \nu_U; W_{ij}) +\tilde\sigma_\epsilon^2(I_n)_{ij}  \\
&\implies (\sigma_U^2 - \tilde\sigma_U^2)C(\phi, \nu_U; W_{ij}) = (\sigma_\epsilon^2 - \tilde\sigma_\epsilon^2)(I_n)_{ij}.\end{align*}
Because there exist $i \ne j$ such that $W_{ij}>0$, we must have $(\sigma_U^2 - \tilde\sigma_U^2)C(\phi, \nu_U; W_{ij}) =0$, which further implies that $\sigma_U^2 = \tilde\sigma_U^2$ due to Lemma  \ref{lem:bessel-properties}(iv). 
Subsequently, we have $\sigma_\epsilon^2 = \tilde\sigma_\epsilon^2$, which further implies that $\rho = \tilde\rho$.

From the above, we derive Theorem \ref{cor:bivariate-matern-id-known-smoothness}. 
\end{proof}

\subsection{Proof of Theorem \ref{thm:bivariate-matern-id-full-asymptotic}}

To prove Theorem \ref{thm:bivariate-matern-id-full-asymptotic}, we need the following two lemmas.

\begin{lemma}\label{lem:matern-asymptotic-2-linind-two-parameter}
Let $S$ be a set with $\sup S = \infty$, that is, for every finite $L$, $S$ contains some $w$ with $w>L$. Then the family of Mat\'ern covariance functions for all $\phi>0$ and $\nu >0$ is 2-linearly independent with respect to $S$. That is, for any $(\nu_1, \phi_1) \neq (\nu_2, \phi_2)$ with $\nu_1, \phi_1,\nu_2, \phi_2 > 0$, 
\begin{equation}\label{eq:matern-2-linind-asymptotic-1} 
a_1 C(\phi_1, \nu_1; w) + a_2 C(\phi_2, \nu_2; w)=0
\end{equation} 
for all $w\in S$ if and only if $a_1=a_2=0$.
\end{lemma}

\begin{proof}[Proof of Lemma \ref{lem:matern-asymptotic-2-linind-two-parameter}]
If either $a_1, a_2$ is zero, then both must be zero due to Lemma \ref{lem:bessel-properties}(iv). 
Below we suppose that both $a_1$ and $a_2$ are nonzero.
Without loss of generality, we further assume that one of the two cases holds: 
\begin{enumerate}
    \item[(i)] $\phi_1 > \phi_2$; 
    \item[(ii)]  $\phi_1 = \phi_2$, $\nu_1 > \nu_2$. 
\end{enumerate}

From \eqref{eq:matern-2-linind-asymptotic-1}, for any $w\in S$,  we have 
\begin{align*}
    a_1 \frac{2^{1-\nu_1}}{\Gamma(\nu_1)} (w/\phi_1)^{\nu_1} K_{\nu_1}(w/\phi_1) + a_2 \frac{2^{1-\nu_2}}{\Gamma(\nu_2)} (w/\phi_2)^{\nu_2} K_{\nu_2}(w/\phi_2) =0. 
\end{align*}
Using the fact that from Lemma \ref{lem:bessel-properties}(iv)
both terms are nonzero, this implies that 
\begin{align*}
    w^{\nu_1 - \nu_2} \frac{K_{\nu_1}(w/\phi_1)}{K_{\nu_2}(w/\phi_2)} = B
\end{align*}
for some constant $B$ that depends only on $(a_1, a_2, \nu_1, \nu_2, \phi_1, \phi_2)$ and does not depend on $w$. 
Consequently, for any $w\in S$, 
\begin{align*}
    B & = w^{\nu_1 - \nu_2} \frac{K_{\nu_1}(w/\phi_1)}{K_{\nu_2}(w/\phi_2)} \\
    &= w^{\nu_1 - \nu_2} \frac{K_{\nu_1}(w/\phi_1)}{\sqrt{\pi} e^{-(w/\phi_1)}/(2w/\phi_1)^{1/2}}
    \frac{\sqrt{\pi} e^{-(w/\phi_2)}/(2w/\phi_2)^{1/2}}{K_{\nu_2}(w/\phi_2)}
    \frac{\sqrt{\pi} e^{-(w/\phi_1)}/(2w/\phi_1)^{1/2}}{\sqrt{\pi} e^{-(w/\phi_2)}/(2w/\phi_2)^{1/2}}\\
    & =  \frac{K_{\nu_1}(w/\phi_1)}{\sqrt{\pi} e^{-(w/\phi_1)}/(2w/\phi_1)^{1/2}}
    \frac{\sqrt{\pi} e^{-(w/\phi_2)}/(2w/\phi_2)^{1/2}}{K_{\nu_2}(w/\phi_2)}
    \cdot
    \frac{\phi_1^{1/2}}{\phi_2^{1/2}}
    \cdot 
    w^{\nu_1 - \nu_2}
    \exp\big\{(1/\phi_2 - 1/\phi_1)w\big\} 
\end{align*}
From Lemma \ref{lem:bessel-properties}(i), by letting $w\in S$ go to infinity, we must have  
\begin{align}\label{eq:lim_h}
    B = \frac{\phi_1^{1/2}}{\phi_2^{1/2}} \cdot \lim_{w\in S, w\rightarrow \infty} 
    \left[ w^{\nu_1 - \nu_2}
    \exp\big\{(1/\phi_2 - 1/\phi_1)w\big\} \right]. 
\end{align}
In case (i), we have $1/\phi_2 - 1/\phi_1>0$, and consequently the right-hand side of \eqref{eq:lim_h} diverges to infinity. 
In case (ii), we have $w^{\nu_1 - \nu_2}
    \exp\{(1/\phi_2 - 1/\phi_1)w\} = w^{\nu_1 - \nu_2}$ with $\nu_1 - \nu_2>0$, 
and thus the right-hand side of \eqref{eq:lim_h} also diverges to infinity. 
Therefore, both cases lead to contradiction. 

From the above, we must have $a_1=a_2=0$. Therefore, Lemma \ref{lem:matern-asymptotic-2-linind-two-parameter} holds.
\end{proof}

\begin{lemma}\label{lem:matern-asymptotic-3-linind}
Let $S$ be a set with $\sup S = \infty$, that is, for every finite $L$, $S$ contains some $w$ with $w>L$. Then the family of Mat\'ern covariance functions for a fixed $\phi>0$ and all $\nu>0$ is 3-linearly independent with respect to $S$. That is, for mutually distinct $\nu_1, \nu_2, \nu_3$, 
\begin{equation}\label{eq:matern-4}  a_1 C(\phi, \nu_1; w) + a_2 C(\phi, \nu_2; w) + a_3C(\phi, \nu_3;w) = 0\end{equation} for all $w\in S$ if and only if $a_1=a_2=a_3=0$.
\end{lemma}
\begin{proof}[Proof of Lemma \ref{lem:matern-asymptotic-3-linind}] 
Assume without loss of generality that $\nu_1 > \nu_2 > \nu_3$. From Lemma \ref{lem:matern-asymptotic-2-linind-two-parameter}, if any of $a_1, a_2$ and $a_3$ is $0$, then $a_1=a_2=a_3=0$. 
Below we suppose that 
$a_1, a_2, a_3$ are all nonzero.

We can rewrite \eqref{eq:matern-4} as:
\[
C(\phi, \nu_1; w)  + b_2 C(\phi, \nu_2; w)  + b_3 C(\phi, \nu_3; w)  = 0,
\] 
where $b_2 = a_2/a_1, b_3=a_3/a_1$.
From Lemma \ref{lem:bessel-properties}(iv), both $C(\phi, \nu_2; w)$ and $C(\phi, \nu_3; w)$ are positive for all $w$.  We can then rewrite the above equation as 
\begin{equation}\label{eq:matern-expression}
\frac{C(\phi, \nu_2; w) }{C(\phi, \nu_3; w) } 
\left\{ \frac{C(\phi, \nu_1; w) }{C(\phi, \nu_2; w) } + b_2 \right\}  = -b_3.
\end{equation} 
Letting $z = w/\phi$, we then have 
\begin{align*}
    & \frac{2^{1-\nu_2}/\Gamma(\nu_2)}{2^{1-\nu_3}/\Gamma(\nu_3)}z^{\nu_2-\nu_3} \frac{K_{\nu_2}(z)}{K_{\nu_3}(z)}
\left\{ \frac{2^{1-\nu_1}/\Gamma(\nu_1)}{2^{1-\nu_2}/\Gamma(\nu_2)}z^{\nu_1-\nu_2} \frac{K_{\nu_1}(z)}{K_{\nu_2}(z)} + b_2 \right\}  = -b_3\\
\Longrightarrow \  &
c_1 \cdot z^{\nu_2-\nu_3} \frac{K_{\nu_2}(z)}{K_{\nu_3}(z)}
\left\{ c_2 \cdot z^{\nu_1-\nu_2} \frac{K_{\nu_1}(z)}{K_{\nu_2}(z)} + b_2 \right\}  = -b_3, 
\end{align*}
where $c_1 = {2^{1-\nu_2}/\Gamma(\nu_2)}/\{2^{1-\nu_3}/\Gamma(\nu_3)\}>0$ and 
$c_2 = {2^{1-\nu_1}/\Gamma(\nu_1)}/\{2^{1-\nu_2}/\Gamma(\nu_2)\}>0$ do not depend on $z$. 
From Lemma \ref{lem:bessel-properties}(ii), for any $w\in S$ and $z=w/\phi$, we have 
\begin{align*}
    -b_3 = c_1 \cdot z^{\nu_2-\nu_3} \frac{K_{\nu_2}(z)}{K_{\nu_3}(z)}
\left\{ c_2 \cdot z^{\nu_1-\nu_2} \frac{K_{\nu_1}(z)}{K_{\nu_2}(z)} + b_2 \right\}  
\ge c_1 \cdot z^{\nu_2-\nu_3} \frac{K_{\nu_2}(z)}{K_{\nu_3}(z)}
\left( c_2 \cdot z^{\nu_1-\nu_2} + b_2 \right).   
\end{align*}
For sufficiently large $w\in S$ and $z=w/\phi$, we then have, $c_2 \cdot z^{\nu_1-\nu_2} + b_2>0$. Consequently, from Lemma \ref{lem:bessel-properties}(ii), for sufficiently large $w\in S$ and $z=w/\phi$, 
\begin{align*}
    -b_3 \ge c_1 \cdot z^{\nu_2-\nu_3} \frac{K_{\nu_2}(z)}{K_{\nu_3}(z)}
\left( c_2 \cdot z^{\nu_1-\nu_2} + b_2 \right)
\ge  c_1 \cdot z^{\nu_2-\nu_3}
\left( c_2 \cdot z^{\nu_1-\nu_2} + b_2 \right), 
\end{align*}
which must diverge to infinity as $w\rightarrow \infty$. This leads to a contradiction. 

From the above, Lemma \ref{lem:matern-asymptotic-3-linind} holds.
\end{proof}

\begin{proof}[Proof of Theorem \ref{thm:bivariate-matern-id-full-asymptotic}]
By the same logic as Theorem \ref{cor:bivariate-matern-id-known-smoothness}, for all $1\le i, j\le n$, 
\begin{align*}\var(Z)_{ij} &= \sigma_Z^2 C(\phi, \nu_Z; W_{ij}),\\
\var(Y)_{ij} &= \beta^2 \sigma_Z^2 C(\phi, \nu_Z; W_{ij}) + 2\beta\rho \sigma_U\sigma_ZC(\phi, 2^{-1}(\nu_U+\nu_Z); W_{ij}) +  \sigma_U^2 C(\phi, \nu_U; W_{ij}) + \sigma_\epsilon^2(I_n)_{ij} , \\
\cov (Y,Z)_{ij} &= \beta \sigma_Z^2 C(\phi, \nu_Z; W_{ij}) + \rho\sigma_U\sigma_Z C(\phi, 2^{-1}(\nu_U+\nu_Z); W_{ij}).\end{align*} 
Let $(\beta, \rho, \sigma_U, \sigma_Z, \phi, \nu_U, \nu_Z, \sigma_\epsilon)$ denote the true data generating parameters, and suppose that $(\tilde\beta, \tilde\rho, \tilde\sigma_U, \tilde\sigma_Z, \tilde\phi, \tilde\nu_U, \tilde\nu_Z, \tilde\sigma_\epsilon)$  yield the same observed data distribution. 
Below we prove that $(\tilde\beta, \tilde\rho, \tilde\sigma_U, \tilde\sigma_Z, \tilde\phi, \tilde\nu_U, \tilde\nu_Z, \tilde\sigma_\epsilon)$ must be the same as $(\beta, \rho, \sigma_U, \sigma_Z, \phi, \nu_U, \nu_Z, \sigma_\epsilon)$. 

First, we considering $\var(Z)$. For every $i,j$, we have 
\[\sigma_Z^2 C(\phi, \nu_Z; W_{ij}) = \tilde\sigma_Z^2 C(\tilde\phi, \tilde\nu_Z; W_{ij}). \]
Note that both $\sigma_Z^2$ and $\tilde\sigma_Z^2$ are positive. 
From Lemma \ref{lem:matern-asymptotic-2-linind-two-parameter}, we must have $\nu_Z=\tilde\nu_Z$ and $\tilde\phi=\phi$.
Considering further the case where $i=j$, we can immediately obtain $\tilde\sigma_Z^2 = \sigma_Z^2$, since both covariance functions take the value 1 when $W_{ii}=0$.

We then consider $\cov (Y,Z)$. After some arrangement and noting our identified parameters above, we have, for all $i,j,$
\begin{equation}\label{eq:pars-covyz} (\beta-\tilde{\beta}) \sigma_Z^2 
C(\phi, \nu_Z; W_{ij}) + \rho\sigma_U\sigma_Z 
C(\phi, 2^{-1}(\nu_U+\nu_Z); W_{ij}) -\tilde\rho\tilde\sigma_U\sigma_Z  C(\phi, 2^{-1}(\tilde\nu_U+\nu_Z); W_{ij}) = 0.\end{equation} By assumption, $\nu_U\neq \nu_Z=\tilde\nu_Z \neq \tilde\nu_U$. If further $\nu_U \neq\tilde \nu_U$, then $\nu_Z, 2^{-1}(\nu_U+\nu_Z)$ and $2^{-1}(\tilde\nu_U+\nu_Z)$ are mutually distinct, which, by Lemma \ref{lem:matern-asymptotic-3-linind}, implies that $\beta = \tilde\beta$ and $\rho \sigma_U = \tilde\rho \tilde\sigma_U =0$.  
If $\nu_U = \tilde\nu_U$, then we also have $\beta = \tilde\beta$ and $\rho\sigma_U = \tilde\rho \tilde\sigma_U$. 
This is because 3-linear independence implies 2-linear independence. 
Thus, we must have $\beta = \tilde\beta$ and $\rho\sigma_U = \tilde\rho \tilde\sigma_U$.

Because $\beta = \tilde\beta$, from \eqref{eq:pars-covyz}, we then have $\rho\sigma_U\sigma_Z C(\phi, 2^{-1}(\tilde\nu_U+\nu_Z); W_{ij}) = \tilde\rho\tilde\sigma_U \sigma_Z C(\phi, 2^{-1}(\nu_U+\nu_Z); W_{ij})$ for all $i,j$. 
From $\var(Y)$,  we can then obtain that, for all $i,j$, 
\[\sigma_U^2 C(\phi, \nu_U; W_{ij}) - \tilde\sigma_U^2 C(\phi, \tilde\nu_U; W_{ij}) + (\sigma_\epsilon^2 - \tilde\sigma_\epsilon^2) (I_n)_{ij} =0.\]
For $i \neq j$, this simplifies to $\sigma_U^2 C(\phi, \nu_U; W_{ij}) - \tilde\sigma_U^2 C(\phi, \tilde\nu_U; W_{ij})=0$. If $\nu_U\neq \tilde\nu_U$, then, by Lemma \ref{lem:matern-asymptotic-2-linind-two-parameter} we have $\sigma_U=\tilde\sigma_U=0$, leading to a contradiction.  Thus, we must have $\nu_U=\tilde\nu_U$ and consequently  $\sigma^2_U = \tilde{\sigma}^2_U$. 
Subsequently we can derive that $\sigma_\epsilon^2 = \tilde\sigma_\epsilon^2$. 
Since  $\rho\sigma_U = \tilde\rho \tilde\sigma_U$, we also have $\rho = \tilde\rho$. 

From the above, Theorem \ref{thm:bivariate-matern-id-full-asymptotic} holds. 
\end{proof}

\section{More about linear independence of some other covariance functions}\label{sec:linear-non-ind}

\subsection{Spherical Covariance}

The spherical covariance function is parametrized as 
$$
f_\phi(w) = \left\{ 1 - 1.5\left(\frac{w}{\phi}\right) + 0.5\left(\frac{w}{\phi}\right)^3\right\} \cdot \mathbf{1}_{w\le\phi}, 
\quad \text{where } \phi > 0.
$$ 
The $K$-linear independence of the family is complicated in general to state for the full domain of $\phi$, 
due to the indicator function component. In particular, it is obvious that linear independence does not occur if the given values of $w$ were all sufficiently large.  
For convenience, we 
restrict $\phi$ from below so that there are sufficiently many pairwise distances smaller than or equal to the minimum possible $\phi$.

\begin{proposition}\label{spherical-3-lin-ind}
Consider the family of spherical covariance functions with range bounded from below by some fixed $c>0$, i.e.,  
$$
\mathcal{F}_{c}= \left\{f_\phi: f_\phi(w) = \left\{ 1 - 1.5\left(\frac{w}{\phi}\right) + 0.5\left(\frac{w}{\phi}\right)^3\right\} \cdot \mathbf{1}_{w\le\phi}, \ \phi \ge c\right\}. 
$$
This family is 3-linearly independent with respect to any set $\mathcal{S}$ containing at least four values smaller than or equal to $c$.
\end{proposition}

\begin{proof}[Proof of Proposition \ref{spherical-3-lin-ind}]
Consider any mutually distinct $\phi_1, \phi_2, \phi_3\in [c, \infty)$. 
Suppose 
that 
for $w\in \mathcal{S}$ and $a_1, a_2, a_3\in \mathbb{R}$,  
\begin{equation}\label{eq:spherical-3-lind-zero}
a_1f_{\phi_1}(w) + a_2f_{\phi_2}(w) + a_3f_{\phi_3}(w)=0.
\end{equation} 
By definition, this implies that, for $w\in \mathcal{S}$ and $w\le c$,  
\[(a_1+a_2+a_3) -1.5\left(\frac{a_1}{\phi_1} + \frac{a_2}{\phi_2}+\frac{a_3}{\phi_3} \right)w + 0.5 \left(\frac{a_1}{\phi_1^3} + \frac{a_2}{\phi_2^3}
+\frac{a_3}{\phi_3^3} \right)w^3 =0\]
Note that the expression on the left-hand side is a cubic polynomial in $w$, and thus cannot have four distinct roots. Thus, we must have 
\[\begin{pmatrix} 1 & 1 & 1 \\
\frac{1}{\phi_1} & \frac{1}{\phi_2} & \frac{1}{\phi_3} \\
\frac{1}{\phi_1^3} & \frac{1}{\phi_2^3} & \frac{1}{\phi_3^3}\end{pmatrix} \begin{pmatrix} a_1 \\ a_2 \\ a_3 \end{pmatrix} = \textbf{0}.\]
We can verify that the determinant of the matrix on the left-hand side can be computed as 
\[
\frac{(\phi_1-\phi_2)(\phi_1-\phi_3)(\phi_2-\phi_3)\,(\phi_1\phi_2+\phi_1\phi_3+\phi_2\phi_3)}
{\phi_1^3\phi_2^3\phi_3^3},
\] which must be nonzero since $\phi_1, \phi_2, \phi_3$ are all distinct. Thus, we must have $a_1 = a_2 = a_3 = 0$. 
From the above, we can then derive Proposition \ref{spherical-3-lin-ind}. 
\end{proof}

\begin{proposition}\label{spherical-4-non-lin-ind}
Consider the family of spherical covariance functions with range bounded from below by some fixed $c>0$, defined as in Proposition \ref{spherical-3-lin-ind}. 
For $K \ge 4$, the family is not $K$-linearly independent with respect to any set $\mathcal{S}\subset [0, \infty)$ is not dense in $[0, \infty)$.

\end{proposition} 
\begin{proof}[Proof of Proposition \ref{spherical-4-non-lin-ind}]
It suffices to consider the case with $K=4$. 
Below we 
construct distinct $\phi_1,\phi_2,\phi_3, \phi_4 >0$ and $(a_1, a_2, a_3, a_4) \neq (0,0,0,0)$ such that for all $w \in S$,
\begin{align}\label{eq:sphe_4}
    a_1f_{\phi_1}(w) + a_2f_{\phi_2}(w) + a_3f_{\phi_3}(w) + a_4f_{\phi_4}(w) = 0
\end{align}
where $f_{\phi_j}$, $1\le j \le 4$, is defined as in Proposition \ref{spherical-3-lin-ind}.

Consider any set $\mathcal{S}\subset [0, \infty)$ that is not dense in $[0, \infty)$.
We can find $\phi_1 > \phi_2> \phi_3> \phi_4$ such that $S\cap [\phi_4, \phi_1] = \emptyset$, i.e., 
every element $s\in S$ is either larger than $\phi_1$ or smaller than $\phi_4$. For all elements of $\mathcal{S}$ that are larger than $\phi_1$, the equation in \eqref{eq:sphe_4} holds for any choice of $(a_1, a_2, a_3, a_4)$. For the other elements of $s$, 
the equation in \eqref{eq:sphe_4} holds as long as 
\[\begin{pmatrix} 1 & 1 & 1 & 1\\
\phi_1^{-1} & \phi_2^{-1} & \phi_3^{-1} & \phi_4^{-1}\\
\phi_1^{-3} & \phi_2^{-3} & \phi_3^{-3} & \phi_4^{-3}\end{pmatrix} \begin{pmatrix}a_1 \\ a_2 \\ a_3 \\ a_4\end{pmatrix} = \begin{pmatrix} 0 \\ 0 \\ 0 \\ 0 \end{pmatrix},   \]
which must admit a nontrivial solution $(a_1, a_2, a_3, a_4) \neq (0,0,0,0)$.

From the above, Proposition \ref{spherical-4-non-lin-ind} holds. 
\end{proof}

\subsection{Wave Covariance}

\begin{proposition}\label{prop:wave-cov-not-k-distinguishable} 
Consider the family of wave covariance functions:
\[\mathcal{F} =\left\{ f_\phi(w) = \begin{cases} \frac{\sin (w/\phi)}{w/\phi}, & w > 0\\
1, & w = 0
\end{cases}, \quad \phi > 0\right\}\]
For any two distinct and positive $\phi\neq \phi'$, $f_{\phi}(w) = f_{\phi'}(w)$ for infinitely many $w\ge 0$. Consequently, for any finite integer $M\ge 1$, the family may fail to be 2-linearly independent with respect to a set of size $M$.
\end{proposition}

\begin{proof}[Proof of Proposition \ref{prop:wave-cov-not-k-distinguishable} ]
Without loss of generality, we assume $\phi > \phi'$ and let $\alpha = \phi/\phi' >1$.
For any $w>0$, 

\begin{align*}
& f_{\phi}(w) = f_{\phi'}(w) \\ \iff &
\frac{\sin(w/\phi)}{w/\phi} = \frac{ \sin(w/\phi')}{w/\phi'}\\ \iff &\sin(w/\phi) = \frac{1}{\alpha} \sin(w/\phi')\\
\iff &
\sin z = \frac{1}{\alpha}\sin (\alpha z),  
\end{align*}
where $z=w/\phi$. 
Thus, it suffices to find infinitely many $z>0$ such that 
\begin{align}\label{eq:sinz}
    \sin z = \sin (\alpha z)/\alpha. 
\end{align}

Consider first the case where $\alpha$ is a rational number. 
In this case, we can write $\alpha = p/q$ for two positive integer $p, q$. 
Equation \eqref{eq:sinz} will then hold  with $z = nq\pi$, for any integer $n$.

Consider then the case where $\alpha$ is an irrational number. 
Consider any positive integer $n$. 
Let $k= \lfloor \alpha n \rfloor$ be the largest integer smaller than or equal to $\alpha n$. 
Because $\alpha$ is an irrational number, we must have $k < \alpha n < k+1$. 
Let $z_1 = k\pi/\alpha$ and  $z_2 = (k+1)\pi/\alpha$. 
We then have (i) $z_1 < n \pi < z_2$, 
(ii) $\sin (\alpha z_1) = \sin (\alpha z_2) = 0$. 
Moreover, since $z_2 - z_1 = \pi/\alpha < \pi$, then $\sin(z_1)$ and $\sin(z_2)$ must both be nonzero and have different signs. 
Let $g(z) =  \sin z - \sin (\alpha z)/\alpha$. 
We then have 
$g(z_1) \cdot g(z_2) = \sin(z_1) \cdot \sin(z_2) < 0$. Thus $g(z_1)$ and $g(z_2)$ must both be nonzero and have different signs.
Because $g$ is a continuous function, there must exist $\tilde{z} \in (z_1, z_2) \subset ((n-1)\pi, (n+1)\pi)$ such that $g(\tilde{z})=0$, i.e., equation \eqref{eq:sinz} holds with $z=\tilde{z}$.

From the above, Proposition \ref{prop:wave-cov-not-k-distinguishable} holds. 
\end{proof}

\section{Positive Definiteness}\label{sec:pd}

In this section, we briefly discuss the conditions for  positive definiteness of the covariance matrix for $(Z,U)$ under various model assumptions in the main text. 
We first present a useful lemma below.

\begin{lemma}\label{lem:pd-schur}
Let $\Sigma$ be a symmetric matrix defined as 
\[
\Sigma \equiv \begin{pmatrix}
\Sigma_{UU} & \Sigma_{UZ} \\
\Sigma_{ZU} & \Sigma_{ZZ}
\end{pmatrix}, 
\text{ where } \Sigma_{UU}, \Sigma_{ZZ}, \Sigma_{UZ} = \Sigma_{ZU}^\top \in \mathbb{R}^{n\times n}. 
\]
\begin{itemize}
    \item[(i)] $\Sigma$ is positive definite if and only if both $\Sigma_{ZZ}$ and $\Sigma_{UU} - \Sigma_{UZ} \Sigma_{ZZ}^{-1} \Sigma_{ZU}$ is positive definite. 

    \item[(ii)] If $\lambda_{\min}(\Sigma_{ZZ})> 0$ and  $\lambda_{\min}(\Sigma_{UU}) \cdot \lambda_{\min}(\Sigma_{ZZ}) > \sigma_{\max}\{ (\Sigma_{UZ})\}^2$, then $\Sigma$ is positive definite, where $\lambda_{\min}(\cdot)$ and $\sigma_{\max}(\cdot)$ denote, respectively, the smallest eigenvalue and the largest singular value of a matrix. 
\end{itemize}
\end{lemma}
\begin{proof}[Proof of Lemma \ref{lem:pd-schur}]
(i) follows from \citet[Theorem 1.12(a)]{zhang2006schur}. 
Below we prove (ii). 
Suppose that  $\lambda_{\min}(\Sigma_{ZZ})> 0$ and   $\lambda_{\min}(\Sigma_{UU}) \cdot \lambda_{\min}(\Sigma_{ZZ}) > \{\sigma_{\max} (\Sigma_{UZ})\}^2$. Then $\Sigma_{ZZ}$ must be positive definite. 
Consider any $x \in \mathbb{R}^{n}$ and $x \ne 0$. We have 
\begin{align*}
    x^\top (\Sigma_{UU} - \Sigma_{UZ} \Sigma_{ZZ}^{-1} \Sigma_{ZU}) x & = x^\top \Sigma_{UU} x - x^\top \Sigma_{UZ} \Sigma_{ZZ}^{-1} \Sigma_{ZU} x \\
    & \ge \lambda_{\min}(\Sigma_{UU}) \cdot \|x\|_2^2 - 1/\lambda_{\min}(\Sigma_{ZZ}) \cdot \|\Sigma_{ZU} x \|_2^2\\
    & \ge \lambda_{\min}(\Sigma_{UU}) \cdot \|x\|_2^2 - \{\sigma_{\max} (\Sigma_{UZ})\}^2/\lambda_{\min}(\Sigma_{ZZ}) \cdot \|x\|_2^2\\
    & = \frac{\lambda_{\min}(\Sigma_{UU}) \cdot \lambda_{\min}(\Sigma_{ZZ})- \{\sigma_{\max} (\Sigma_{UZ})\}^2}{\lambda_{\min}(\Sigma_{ZZ})}  \cdot \|x\|_2^2,
\end{align*}
which must be positive. 
Thus, $\Sigma_{UU} - \Sigma_{UZ} \Sigma_{ZZ}^{-1} \Sigma_{ZU}$ must be positive definite. (ii) then follows immediately from (i). 
\end{proof}
 
The covariance or precision matrices for $(U,Z)$ in Sections 
\ref{sec:car-papadogeorgu}, \ref{sec:leroux_car}, 
\ref{sec:bivariate_spatial}, and \ref{sec:matern} 
all share a structure similar to that in Lemma~\ref{lem:pd-schur}. 
We can then use Lemma~\ref{lem:pd-schur} to derive sufficient conditions for their positive definiteness.

Moreover, the positive definiteness of the general (not necessarily parsimonious) 
Mat\'ern covariance structure has been studied in 
\citet{gneiting2010matern} and \citet{onnen2021estimation}. 
Since the exponential covariance function is a special case of the 
Mat\'ern covariance function with smoothness parameter $\nu = 1/2$, 
these results also imply positive definiteness for the covariance structure in 
Section~\ref{sec:bivariate_spatial} when the exponential covariance 
function is used.

\bibliographystyle{plainnat}
\bibliography{paper-ref}

\end{document}